\documentclass[reqno]{amsart}

\usepackage{fullpage}

\usepackage{ bbold }
\usepackage{booktabs}

\usepackage{color,soul}
\definecolor{lightblue}{rgb}{.90,.95,1}
\sethlcolor{lightblue}

\usepackage[]{tikz}
\usetikzlibrary{arrows}
\usetikzlibrary{calc}
\usepackage[T2A]{fontenc}
\usepackage{amssymb, amsmath, amsthm}
\usepackage[utf8]{inputenc}
\usepackage[english,russian]{}
\usepackage{xcolor}
\usepackage{graphicx}
\usepackage{caption}
\usepackage{subcaption}

\theoremstyle{plain}
\newtheorem{Th}{Theorem}[section]
\newtheorem{lemma}[Th]{Lemma}
\newtheorem{remark}[Th]{Remark}
\newtheorem{prop}[Th]{Proposition}
\newtheorem{cor}[Th]{Corollary}
\newtheorem{definition}[Th]{Definition}

\newcommand{\bb}[1]{\ensuremath{{{\color{gray}\blacklozenge}}_{\raisebox{-1pt}{\tiny $\mathrm{#1}$}}}}
\newcommand{\ww}[1]{\lozenge_{\raisebox{-1pt}{\tiny $\mathrm{#1}$}}}
\newcommand{\V}[1]{\mathcal{V}_{\raisebox{-1pt}{\tiny#1}}}
\newcommand{\bblack}[1]{\ensuremath{{{\color{black}\blacklozenge}}_{\raisebox{-1pt}{\tiny $\mathrm{#1}$}}}}
\newcommand{\wwhite}[1]{\ensuremath{{{\color{white}\blacklozenge}}_{\raisebox{-1pt}{\tiny $\mathrm{#1}$}}}}

\newcommand{\om}{\Omega}
\newcommand{\intom}{\mathrm{Int}\, \Omega}
\newcommand{\dV}{\partial\V{}}
\newcommand{\vl}{v_{\raisebox{-1pt}{\tiny$\lambda$}}}
\newcommand{\vlbr}{v_{\raisebox{-1pt}{\tiny $\overline\lambda$}}} 
\newcommand{\uR}{u_{\raisebox{-1pt}{\tiny R}}}
\newcommand{\uI}{u_{\raisebox{-1pt}{\tiny I}}}

\newcommand{\intbb}{\mathrm {Int}\bb{}}

\newcommand{\intb}[1]{\mathrm {Int}\bb{#1}}
\newcommand{\intw}[1]{\mathrm {Int}\ww{#1}}

\newcommand{\dintb}[1]{\partial_{\mathrm {int}}\bb{#1}}
\newcommand{\dintw}[1]{\partial_{\mathrm {int}}\ww{#1}}
\newcommand{\diom}{\partial_{\mathrm {int}}\Omega}
\newcommand{\dib}{\partial_{\mathrm {int}}\bb{}}
\newcommand{\diw}{\partial_{\mathrm {int}}\ww{}}

\newcommand{\clom}{\overline{\om}}
\newcommand{\clb}{\bar{\bb{}}}
\newcommand{\clw}{\bar{\ww{}}}
\newcommand{\clbb}[1]{\bar{\ensuremath{{\color{gray}\blacklozenge}}}_{\raisebox{-1pt}{\tiny $\mathrm{#1}$}}}
\newcommand{\clww}[1]{\bar{\ensuremath{{\lozenge}}}_{\raisebox{-1pt}{\tiny $\mathrm{#1}$}}}

\newcommand{\dom}{\partial\Omega}
\newcommand{\db}{\partial\ensuremath{{\color{gray}\blacklozenge}}}
\newcommand{\dw}{\partial\lozenge}
\newcommand{\dbb}[1]{\partial\bb{#1}}
\newcommand{\dww}[1]{\partial\ww{#1}}

\newcommand{\bbb}{u_0}
\newcommand{\www}{v_0}

\newcommand{\F}{F}
\newcommand{\G}{G}
\newcommand{\HH}{H}
\newcommand{\ff}[1]{f_{{#1}}}

\newcommand{\Cmd}{C_{\Omega^\delta}}

\newcommand{\re}{\mathrm {Re}}
\newcommand{\im}{\mathrm {Im}}

\newcommand{\mdel}{M^\delta}
\newcommand{\Ftilda}{\widetilde{F}}

\newcommand{\old}[1]{}

\begin{document}

\title{Dominos in hedgehog domains}
\large

\author[Marianna Russkikh]{Marianna Russkikh$^\mathrm{\sharp}$}

\thanks{\textsc{
${}^\mathrm{\sharp}$ 
Department of Mathematics, Massachusetts Institute of Technology. 77 Massachusetts Avenue
Cambridge, MA 02139-4307, USA}}



\thanks{{\it E-mail addresses:} \texttt{Russkikh@mit.edu}}

\begin{abstract} 
We introduce a new class of discrete approximations of planar domains that we call ``hedgehog domains''. In particular, this class of approximations contains two-step Aztec diamonds and similar shapes. We show that fluctuations of the height function of a random dimer tiling on hedgehog discretizations of a planar domain converge in the scaling limit to the Gaussian Free Field with Dirichlet boundary conditions. Interestingly enough, in this case the dimer model coupling function satisfies the same Riemann-type boundary conditions as fermionic observables in the Ising model.

In addition, using the same factorization of the double-dimer model coupling function as in~\cite{rus},  we show that in the case of approximations by hedgehog domains the expectation of the double-dimer height function is harmonic in the scaling limit.
\end{abstract}

\maketitle

\tableofcontents


\section{Introduction}
A dimer covering of a graph is a subset of edges that covers every vertex exactly once. The dimer model is a random covering of a given graph by dimers. In this paper, we are interested in uniform random coverings of finite subgraphs (or domains) of the square lattice. Such a dimer covering may be viewed as a random tiling of a domain on the dual lattice by dominos $2\times 1$.

\begin{figure}
\begin{center}
\begin{tikzpicture}[x={(0.7cm,0cm)}, y={(0cm,0.7cm)}]
\begin{scope}
\draw (0,0) -- (2,0) -- (2,1) -- (3,1) -- (3,3) -- (0,3)-- cycle;

\draw  [draw, fill=gray!20](1,0) rectangle (2,1);
\draw  [draw, fill=gray!20](2,1) rectangle (3,2);
\draw  [draw, fill=gray!20](1,2) rectangle (2,3);
\draw  [draw, fill=gray!20](0,1) rectangle (1,2);

\draw[line width=1.3pt] (0,0) rectangle (2,1);
\draw[line width=1.3pt] (1,1) rectangle (3,2);
\draw[line width=1.3pt] (1,2) rectangle (3,3);
\draw[line width=1.3pt] (0,1) rectangle (1,3);

\draw[ line width=2.3pt] (0,0) -- (0,1) -- (2,1) -- (2,2);
\draw[white, line width=1pt,dashed] (0,0) -- (0,1) -- (2,1) -- (2,2);

\fill[black] (0,0)  node[below,font=\tiny]{$0$};
\fill[black] (0,1) node[left,font=\tiny]{$1$};
\fill[black] (0,2)  node[left,font=\tiny]{$0$};
\fill[black] (0,3) node[left,font=\tiny]{$1$};
\fill[black] (1,3) node[above,font=\tiny]{$2$};
\fill[black] (2,3) node[above,font=\tiny]{$1$};
\fill[black] (3,3) node[above,font=\tiny]{$2$};
\fill[black] (3,2) node[right,font=\tiny]{$3$};
\fill[black] (3,1) node[right,font=\tiny]{$2$};
\fill[black] (2,0) node[below,font=\tiny]{$0$};
\fill[black] (1,0) node[below,font=\tiny]{$-1$};
\fill[black] (0.7,1) node[below,font=\tiny]{$2$};
\fill[black] (1.7,1) node[above,font=\tiny]{$1$};
\fill[black] (2.2,2) node[above,font=\tiny]{$4$};
\fill[black] (0.8,2) node[above,font=\tiny]{$3$};

\fill[black] (0,0) circle (1.5pt);
\fill[black] (0,1) circle (1.5pt);
\fill[black] (0,2) circle (1.5pt);
\fill[black] (0,3) circle (1.5pt);
\fill[black] (1,3) circle (1.5pt);
\fill[black] (2,3) circle (1.5pt);
\fill[black] (3,3) circle (1.5pt);
\fill[black] (3,2) circle (1.5pt);
\fill[black] (3,1) circle (1.5pt);
\fill[black] (1,0) circle (1.5pt);
\fill[black] (0,1) circle (1.5pt);
\fill[black] (1,1) circle (1.5pt);
\fill[black] (2,1) circle (1.5pt);
\fill[black] (2,2) circle (1.5pt);
\fill[black] (1,2) circle (1.5pt);
\end{scope}

\begin{scope}[xshift=5cm]
\draw  (-1,3.5) rectangle (0,2.5) node (v3) {};

\draw  [draw, fill=gray!20](-1,4.5) rectangle (0,3.5) node (v1) {};

\draw [draw, fill=gray!20] (v1) rectangle (1,2.5) node (v4) {};

\draw  [draw, fill=gray!20](-2,3.5) rectangle (-1,2.5) node (v2) {};

\draw  [draw, fill=gray!20](v2) rectangle (0,1.5) node (v5) {};

\draw  (v3) rectangle (1,1.5);
\draw  [draw, fill=gray!20](v4) rectangle (2,1.5);
\draw  [draw, fill=gray!20](v5) rectangle (1,0.5);

\draw[draw, line width=1.5pt]  (-0.5,4) edge (-0.5,2);
\draw (-0.55,3.6) node[anchor=west]{$i$};
\draw (-0.45,2.4) node[anchor=east]{$-i$};
\draw[draw, line width=1.5pt]  (-1.5,3) edge (0.5,3);
\draw (0.1,2.95) node[anchor=south,font=\small]{$1$};
\draw (-1.1,2.95) node[anchor=south,font=\small]{$-1$};

\draw[draw, line width=1.5pt]  (0.5,3) edge (0.5,1);
\draw (0.45,2.6) node[anchor=west]{$i$};
\draw (0.55,1.4) node[anchor=east]{$-i$};
\draw[draw, line width=1.5pt]  (-0.5,2) edge (1.5,2);
\draw (1.1,1.95) node[anchor=south,font=\small]{$1$};
\draw (-0.05,1.95) node[anchor=south,font=\small]{$-1$};

\path (0.5,3) node[]{$\bblack{}$};
\path (-1.5,3) node[]{$\bblack{}$};
\path (-0.5,4) node[]{$\bblack{}$};
\path (-0.5,2) node[]{$\bblack{}$};
\path (-0.5,3) node[]{$\wwhite{}$};
\path (-0.5,3) node[]{$\ww{}$};

\path (0.5,2) node[]{$\wwhite{}$};
\path (0.5,2) node[]{$\ww{}$};
\path (0.5,1) node[]{$\bblack{}$};
\path (1.5,2) node[]{$\bblack{}$};
\end{scope}

\end{tikzpicture} 
 \begin{tikzpicture}[x={(0.4cm,0cm)}, y={(0cm,0.4cm)}]
 
 \draw  [draw, white](-2,5) rectangle (-1,6);
\draw  [draw, fill=gray!20](2,4) rectangle (3,5);

\draw  [draw, fill=gray!20](4,4) rectangle (5,5);
\draw  [draw, fill=gray!20](4,6) rectangle (5,7);

\draw  [draw, fill=gray!20](6,2) rectangle (7,3);
\draw  [draw, fill=gray!20](6,4) rectangle (7,5);
\draw  [draw, fill=gray!20](6,6) rectangle (7,7);
\draw  [draw, fill=gray!20](6,8) rectangle (7,9);

\draw  [draw, fill=gray!20](8,4) rectangle (9,5);
\draw  [draw, fill=gray!20](8,6) rectangle (9,7);
\draw  [draw, fill=gray!20](8,8) rectangle (9,9);

\draw  [draw, fill=gray!20](10,4) rectangle (11,5);
\draw  [draw, fill=gray!20](10,6) rectangle (11,7);

\draw  [draw, fill=gray!90](1,3) rectangle (2,4);
\draw  [draw, fill=gray!90](1,5) rectangle (2,6);

\draw  [draw, fill=gray!90](3,3) rectangle (4,4);
\draw  [draw, fill=gray!90](3,5) rectangle (4,6);
\draw  [draw, fill=gray!90](3,7) rectangle (4,8);

\draw  [draw, fill=gray!90](5,1) rectangle (6,2);
\draw  [draw, fill=gray!90](5,3) rectangle (6,4);
\draw  [draw, fill=gray!90](5,5) rectangle (6,6);
\draw  [draw, fill=gray!90](5,7) rectangle (6,8);
\draw  [draw, fill=gray!90](5,9) rectangle (6,10);

\draw  [draw, fill=gray!90](7,1) rectangle (8,2);
\draw  [draw, fill=gray!90](7,3) rectangle (8,4);
\draw  [draw, fill=gray!90](7,5) rectangle (8,6);
\draw  [draw, fill=gray!90](7,7) rectangle (8,8);
\draw  [draw, fill=gray!90](7,9) rectangle (8,10);

\draw  [draw, fill=gray!90](9,3) rectangle (10,4);
\draw  [draw, fill=gray!90](9,5) rectangle (10,6);
\draw  [draw, fill=gray!90](9,7) rectangle (10,8);
\draw  [draw, fill=gray!90](9,9) rectangle (10,10);

\draw  [draw, fill=gray!90](11,3) rectangle (12,4);
\draw  [draw, fill=gray!90](11,5) rectangle (12,6);
\draw  [draw, fill=gray!90](11,7) rectangle (12,8);


\draw[draw, line width=1pt,black!90]  
(5,2)--(6,2)--(6,1)--(8,1)--(8,3)--(12,3)--(12,8)--(10,8)--(10,10)--(5,10)--(5,8)--(3,8)--(3,6)--(1,6)--(1,3)--(5,3)--cycle;
\draw (6,0.65) node[font=\tiny]{$\bf{0}$};
\draw (7,0.65) node[font=\tiny]{\bf{-1}};
\draw (8,0.65) node[font=\tiny]{$\bf{0}$};
\draw (8.3,1.8) node[font=\tiny]{$\bf{1}$};
\draw (8.3,2.7) node[font=\tiny]{$\bf{0}$};
\draw (9,2.65) node[font=\tiny]{\bf{-1}};
\draw (10,2.65) node[font=\tiny]{$\bf{0}$};
\draw (11,2.65) node[font=\tiny]{\bf{-1}};
\draw (12.3,2.65) node[font=\tiny]{$\bf{0}$};

\draw (12.3,4) node[font=\tiny]{$\bf{1}$};
\draw (12.3,5) node[font=\tiny]{$\bf{0}$};
\draw (12.3,6) node[font=\tiny]{$\bf{1}$};
\draw (12.3,7) node[font=\tiny]{$\bf{0}$};
\draw (12.3,8.35) node[font=\tiny]{$\bf{1}$};

\draw (11.2,8.35) node[font=\tiny]{$\bf{2}$};
\draw (10.3,8.35) node[font=\tiny]{$\bf{1}$};

\draw (10.3,9.15) node[font=\tiny]{$\bf{0}$};
\draw (10.3,10.35) node[font=\tiny]{$\bf{1}$};

\draw (9,10.35) node[font=\tiny]{$\bf{2}$};
\draw (8,10.35) node[font=\tiny]{$\bf{1}$};
\draw (7,10.35) node[font=\tiny]{$\bf{2}$};
\draw (6,10.35) node[font=\tiny]{$\bf{1}$};
\draw (4.7,10.35) node[font=\tiny]{$\bf{2}$};

\draw (4.7,9.25) node[font=\tiny]{$\bf{3}$};
\draw (4.7,8.35) node[font=\tiny]{$\bf{2}$};

\draw (3.9,8.35) node[font=\tiny]{$\bf{1}$};
\draw (2.7,8.35) node[font=\tiny]{$\bf{2}$};

\draw (2.7,7.35) node[font=\tiny]{$\bf{3}$};
\draw (2.7,6.35) node[font=\tiny]{$\bf{2}$};

\draw (1.9,6.35) node[font=\tiny]{$\bf{1}$};
\draw (0.7,6.35) node[font=\tiny]{$\bf{2}$};

\draw (0.7,5) node[font=\tiny]{$\bf{3}$};
\draw (0.7,4) node[font=\tiny]{$\bf{2}$};
\draw (0.7,2.65) node[font=\tiny]{$\bf{3}$};

\draw (2,2.65) node[font=\tiny]{$\bf{4}$};
\draw (3,2.65) node[font=\tiny]{$\bf{3}$};
\draw (4,2.65) node[font=\tiny]{$\bf{4}$};
\draw (4.7,2.65) node[font=\tiny]{$\bf{3}$};

\draw (4.7,1.65) node[font=\tiny]{$\bf{2}$};
\draw (5.7,1.65) node[font=\tiny]{$\bf{1}$};
\end{tikzpicture}
\caption{{\bf Left:} a domino tiling of a domain; an edge-path $\gamma$ from $z_0$ to $z$ and the height along this path: $h^\delta(z_0)=0$, $h^\delta(z)=-4$. {\bf Center:} Weights of the Kasteleyn matrix~$K_\Omega$ on the square lattice (proposed by Kenyon in~\cite{Kdom}). {\bf Right:} A Temperleyan domain.  
The difference of the height function at two boundary vertices is related to the amount of the winding of the boundary (the number of left turns minus the number of right turns) between them. 
}\label{hf_cf}\label{temp}\end{center}
\end{figure}
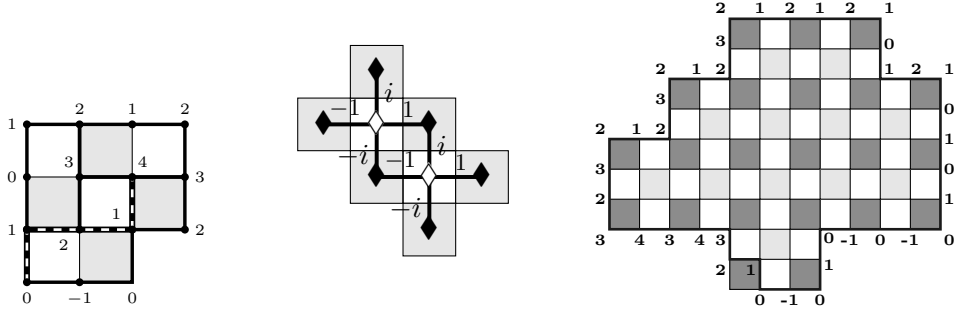

Thurston~\cite{Ter} introduced the height function of a domino tiling which assigns real values to all vertices as follows. 
Fix a vertex $z_0$ and set $h(z_0)=0$.
For every other vertex $z$ in the tiling, take an edge-path $\gamma$ from $z_0$ to $z$. The height along $\gamma$ changes by $\pm 1$ if the traversed edge does not cross a domino from the tiling or by $\mp 3$ otherwise (depending on the colour of the square on the left of the traversed edge), see Fig.~\ref{hf_cf}. 
Vice versa, a domino tiling can be reconstructed from the values of the height function. Thus, one can think of a random domino tiling as a random height function on the vertex set of the domain. One key question in the dimer model concerns the large-scale behavior of the expectation of Thurston's height function and of its fluctuations, see for instance~\cite{Kdom, BLR, BG, CKP}. 

In our paper, we use the classical approach of Kenyon based on the Kasteleyn theory of the dimer model on planar graphs. Kasteleyn~\cite{Kast} showed that the partition function of the dimer model can be evaluated as the determinant of a signed adjacency matrix $K_\Omega$, whose rows are indexed by the black vertices and columns are  indexed by the white vertices, the {\it Kasteleyn matrix}, see Fig~\ref{hf_cf}. The local statistics for the uniform measure on dimer configurations can be computed using the inverse Kasteleyn matrix, 
see~\cite{Klocstat}. The latter can be viewed as a function of two squares (one black $u\in\bb{}$, one white $v\in\ww{}$), called the {\it coupling function}~\cite{Kdom}. The main properties of the coupling function $C_\Omega\colon\clb\times\clw\to\mathbb{C}$ are the following: 
 \begin{enumerate}
\item[$\rhd$] if $v\in\ww{}$, then $C_\Omega(u,v)$ is a {discrete holomorphic} function of $u$ with a simple pole at $v$; 
\item[$\rhd$]if $u$ and $v$ are adjacent squares, then $|C_\Omega(u,v)|$ is equal to the probability that the domino  $[uv]$ is contained in a random domino tiling of~$\Omega$;
\item[$\rhd$]moreover, all the joint probabilities to see a collection of dominos $\{[u_kv_k]\}_{k=1}^n$ in a random domino tiling of~$\Omega$, can be expressed via $C_\Omega$ as $n\times n$ determinants.
\end{enumerate}
In other words, the study of the local statistics of random tilings can be reduced to the study of the convergence of discrete holomorphic functions~\cite{Kdom, rus}. We are interested in the scaling limit of the dimer model on the square lattice as the mesh size $\delta$ tends to zero.
Assuming that the functions $\frac{1}{\delta}C_{\Omega^\delta}(\cdot,v)$ are uniformly bounded away from $v$, it follows from
the Arzel\`a–Ascoli theorem that subsequential limits exist. To show that such a limit is unique, we study the
boundary conditions of the coupling function and show that they survive as the mesh size tends to zero and determine the limit uniquely. 

The classical situation in which boundary conditions of the coupling function can be easily described is {\it Temperleyan} discretizations, see Fig.~\ref{temp}. These are discrete domains in which all corner squares have even coordinates. More precisely, consider a checkerboard tiling of a discrete plane with unit squares, and split the set of black squares into two sets $\bb{1}$ and $\bb{0}$ (dark grey and light grey squares on~Fig.~\ref{temp}).
A domain in which all corner squares are of type $\bb{1}$ is called almost Temperleyan domain. To obtain a  Temperleyan domain one removes one black square of type $\bb{1}$ adjacent to the boundary from an almost Temperleyan domain. As the coupling function is discrete holomorphic, its real part (living on $\bb{0}$) and imaginary part (living on $\bb{1}$) are discrete harmonic. 
 In the Temperleyan case, the real part of the coupling function has Dirichlet boundary conditions. 
Kenyon~\cite{Kdom, KGff} used this approach to prove the conformal invariance of the limiting distribution  of the height function in the case of Temperleyan discretizations.

Note that the height function on the boundary does not depend on a domino tiling, and is completely determined by the shape of the boundary. In a Temperleyan domain the boundary values of the height function are related to the winding of the boundary,
see Fig~\ref{temp}. 
Kenyon has shown~\cite{KGff}, that the fluctuations of the height function on Temperleyan discretizations of a planar domain converge in the scaling limit (as the mesh size tends to zero) to the Gaussian Free Field with Dirichlet boundary conditions. This result has been extended for piecewise Temperleyan discretizations in~\cite{rus}. The latter discrete domains correspond to mixed Dirichlet and Neumann boundary conditions for the real part of the coupling function, with prescribed number of changes between them.

\begin{figure}
\begin{center}
 \begin{tikzpicture}[x={(0.38cm,0cm)}, y={(0cm,0.38cm)}]
\begin{scope}
\draw  [draw, fill=gray!20](0,6) rectangle (1,7);
\draw  [draw, fill=gray!20](0,4) rectangle (1,5);

\draw  [draw, fill=gray!20](2,2) rectangle (3,3);
\draw  [draw, fill=gray!20](2,4) rectangle (3,5);
\draw  [draw, fill=gray!20](2,6) rectangle (3,7);
\draw  [draw, fill=gray!20](2,8) rectangle (3,9);

\draw  [draw, fill=gray!20](4,0) rectangle (5,1);
\draw  [draw, fill=gray!20](4,2) rectangle (5,3);
\draw  [draw, fill=gray!20](4,4) rectangle (5,5);
\draw  [draw, fill=gray!20](4,6) rectangle (5,7);
\draw  [draw, fill=gray!20](4,8) rectangle (5,9);
\draw  [draw, fill=gray!20](4,10) rectangle (5,11);

\draw  [draw, fill=gray!20](6,0) rectangle (7,1);
\draw  [draw, fill=gray!20](6,2) rectangle (7,3);
\draw  [draw, fill=gray!20](6,4) rectangle (7,5);
\draw  [draw, fill=gray!20](6,6) rectangle (7,7);
\draw  [draw, fill=gray!20](6,8) rectangle (7,9);
\draw  [draw, fill=gray!20](6,10) rectangle (7,11);

\draw  [draw, fill=gray!20](8,0) rectangle (9,1);
\draw  [draw, fill=gray!20](8,2) rectangle (9,3);
\draw  [draw, fill=gray!20](8,4) rectangle (9,5);
\draw  [draw, fill=gray!20](8,6) rectangle (9,7);
\draw  [draw, fill=gray!20](8,8) rectangle (9,9);

\draw  [draw, fill=gray!20](10,8) rectangle (11,9);
\draw  [draw, fill=gray!20](10,6) rectangle (11,7);
\draw  [draw, fill=gray!20](10,4) rectangle (11,5);
\draw  [draw, fill=gray!20](10,2) rectangle (11,3);
\draw  [draw, fill=gray!20](10,0) rectangle (11,1);

\draw  [draw, fill=gray!20](12,10) rectangle (13,11);
\draw  [draw, fill=gray!20](12,8) rectangle (13,9);
\draw  [draw, fill=gray!20](12,6) rectangle (13,7);
\draw  [draw, fill=gray!20](12,4) rectangle (13,5);
\draw  [draw, fill=gray!20](12,2) rectangle (13,3);
\draw  [draw, fill=gray!20](12,0) rectangle (13,1);

\draw  [draw, fill=gray!20](14,12) rectangle (15,13);
\draw  [draw, fill=gray!20](14,10) rectangle (15,11);
\draw  [draw, fill=gray!20](14,8) rectangle (15,9);
\draw  [draw, fill=gray!20](14,6) rectangle (15,7);
\draw  [draw, fill=gray!20](14,4) rectangle (15,5);
\draw  [draw, fill=gray!20](14,2) rectangle (15,3);
\draw  [draw, fill=gray!20](14,0) rectangle (15,1);

\draw  [draw, fill=gray!20](16,12) rectangle (17,13);
\draw  [draw, fill=gray!20](16,10) rectangle (17,11);
\draw  [draw, fill=gray!20](16,8) rectangle (17,9);
\draw  [draw, fill=gray!20](16,6) rectangle (17,7);
\draw  [draw, fill=gray!20](16,4) rectangle (17,5);
\draw  [draw, fill=gray!20](16,2) rectangle (17,3);

\draw  [draw, fill=gray!20](18,6) rectangle (19,7);
\draw  [draw, fill=gray!20](18,4) rectangle (19,5);

\draw  [draw, fill=gray!90](1,7) rectangle (2,8);
\draw  [draw, fill=gray!90](1,5) rectangle (2,6);

\draw  [draw, fill=gray!90](3,3) rectangle (4,4);
\draw  [draw, fill=gray!90](3,5) rectangle (4,6);
\draw  [draw, fill=gray!90](3,7) rectangle (4,8);
\draw  [draw, fill=gray!90](3,9) rectangle (4,10);

\draw  [draw, fill=gray!90](5,1) rectangle (6,2);
\draw  [draw, fill=gray!90](5,3) rectangle (6,4);
\draw  [draw, fill=gray!90](5,5) rectangle (6,6);
\draw  [draw, fill=gray!90](5,7) rectangle (6,8);
\draw  [draw, fill=gray!90](5,9) rectangle (6,10);
\draw  [draw, fill=gray!90](5,11) rectangle (6,12);

\draw  [draw, fill=gray!90](7,1) rectangle (8,2);
\draw  [draw, fill=gray!90](7,3) rectangle (8,4);
\draw  [draw, fill=gray!90](7,5) rectangle (8,6);
\draw  [draw, fill=gray!90](7,7) rectangle (8,8);
\draw  [draw, fill=gray!90](7,9) rectangle (8,10);
\draw  [draw, fill=gray!90](7,11) rectangle (8,12);

\draw  [draw, fill=gray!90](9,1) rectangle (10,2);
\draw  [draw, fill=gray!90](9,3) rectangle (10,4);
\draw  [draw, fill=gray!90](9,5) rectangle (10,6);
\draw  [draw, fill=gray!90](9,7) rectangle (10,8);
\draw  [draw, fill=gray!90](9,9) rectangle (10,10);

\draw  [draw, fill=gray!90](11,9) rectangle (12,10);
\draw  [draw, fill=gray!90](11,7) rectangle (12,8);
\draw  [draw, fill=gray!90](11,5) rectangle (12,6);
\draw  [draw, fill=gray!90](11,3) rectangle (12,4);
\draw  [draw, fill=gray!90](11,1) rectangle (12,2);

\draw  [draw, fill=gray!90](13,11) rectangle (14,12);
\draw  [draw, fill=gray!90](13,9) rectangle (14,10);
\draw  [draw, fill=gray!90](13,7) rectangle (14,8);
\draw  [draw, fill=gray!90](13,5) rectangle (14,6);
\draw  [draw, fill=gray!90](13,3) rectangle (14,4);
\draw  [draw, fill=gray!90](13,1) rectangle (14,2);

\draw  [draw, fill=gray!90](15,13) rectangle (16,14);
\draw  [draw, fill=gray!90](15,11) rectangle (16,12);
\draw  [draw, fill=gray!90](15,9) rectangle (16,10);
\draw  [draw, fill=gray!90](15,7) rectangle (16,8);
\draw  [draw, fill=gray!90](15,5) rectangle (16,6);
\draw  [draw, fill=gray!90](15,3) rectangle (16,4);
\draw  [draw, fill=gray!90](15,1) rectangle (16,2);

\draw  [draw, fill=gray!90](17,13) rectangle (18,14);
\draw  [draw, fill=gray!90](17,11) rectangle (18,12);
\draw  [draw, fill=gray!90](17,9) rectangle (18,10);
\draw  [draw, fill=gray!90](17,7) rectangle (18,8);
\draw  [draw, fill=gray!90](17,5) rectangle (18,6);
\draw  [draw, fill=gray!90](17,3) rectangle (18,4);

\draw  [draw, fill=gray!90](19,7) rectangle (20,8);
\draw  [draw, fill=gray!90](19,5) rectangle (20,6);

\draw[draw, line width=2pt] (0,4) -- (0,8) -- (2,8) -- (2,10) -- (4,10) -- (4,12) -- (8,12) -- (8,10) -- (10,10) --  (12,10) -- (12,12) -- (14,12) -- (14,14) -- (18,14) -- (18,8) -- (20,8) -- (20,4) -- (18,4) -- (18,2) -- (16,2) -- (16,0) -- (4,0)-- (4,2) -- (2,2) -- (2,4)
-- cycle;

\draw (-0.3,3.65) node[font=\tiny]{$\bf{0}$};

\draw (-0.35,5) node[font=\tiny]{\bf{-1}};
\draw (-0.3,6) node[font=\tiny]{$\bf{0}$};
\draw (-0.35,7) node[font=\tiny]{\bf{-1}};

\draw (-0.3,8.35) node[font=\tiny]{$\bf{0}$};

\draw (0.8,8.35) node[font=\tiny]{$\bf{1}$};
\draw (1.7,8.35) node[font=\tiny]{$\bf{0}$};

\draw (1.65,9.2) node[font=\tiny]{\bf{-1}};

\draw (1.7,10.35) node[font=\tiny]{$\bf{0}$};
\draw (2.8,10.35) node[font=\tiny]{$\bf{1}$};
\draw (3.7,10.35) node[font=\tiny]{$\bf{0}$};

\draw (3.65,11.2) node[font=\tiny]{\bf{-1}};

\draw (3.7,12.35) node[font=\tiny]{$\bf{0}$};
\draw (5,12.35) node[font=\tiny]{$\bf{1}$};
\draw (6,12.35) node[font=\tiny]{$\bf{0}$};
\draw (7,12.35) node[font=\tiny]{$\bf{1}$};
\draw (8.3,12.35) node[font=\tiny]{$\bf{0}$};

\draw (8.4,11.2) node[font=\tiny]{\bf{-1}};
\draw (8.3,10.35) node[font=\tiny]{$\bf{0}$};

\draw (9,10.35) node[font=\tiny]{$\bf{1}$};
\draw (10,10.35) node[font=\tiny]{$\bf{0}$};
\draw (11,10.35) node[font=\tiny]{$\bf{1}$};
\draw (11.7,10.35) node[font=\tiny]{$\bf{0}$};

\draw (11.65,11.2) node[font=\tiny]{\bf{-1}};
\draw (11.7,12.35) node[font=\tiny]{$\bf{0}$};

\draw (12.8,12.35) node[font=\tiny]{$\bf{1}$};
\draw (13.7,12.35) node[font=\tiny]{$\bf{0}$};
\draw (13.65,13.2) node[font=\tiny]{\bf{-1}};
\draw (13.7,14.35) node[font=\tiny]{$\bf{0}$};

\draw (15,14.35) node[font=\tiny]{$\bf{1}$};
\draw (16,14.35) node[font=\tiny]{$\bf{0}$};
\draw (17,14.35) node[font=\tiny]{$\bf{1}$};
\draw (18.3,14.35) node[font=\tiny]{$\bf{0}$};

\draw (18.4,13) node[font=\tiny]{\bf{-1}};
\draw (18.3,12) node[font=\tiny]{$\bf{0}$};
\draw (18.4,11) node[font=\tiny]{\bf{-1}};
\draw (18.3,10) node[font=\tiny]{$\bf{0}$};
\draw (18.4,9.2) node[font=\tiny]{\bf{-1}};
\draw (18.3,8.35) node[font=\tiny]{$\bf{0}$};

\draw (19.2,8.35) node[font=\tiny]{$\bf{1}$};
\draw (20.3,8.35) node[font=\tiny]{$\bf{0}$};

\draw (20.4,7) node[font=\tiny]{\bf{-1}};
\draw (20.3,6) node[font=\tiny]{$\bf{0}$};
\draw (20.4,5) node[font=\tiny]{\bf{-1}};
\draw (20.3,3.65) node[font=\tiny]{$\bf{0}$};

\draw (19.2,3.65) node[font=\tiny]{$\bf{1}$};
\draw (18.3,3.65) node[font=\tiny]{$\bf{0}$};

\draw (18.4,2.8) node[font=\tiny]{\bf{-1}};
\draw (18.3,1.65) node[font=\tiny]{$\bf{0}$};

\draw (17.2,1.65) node[font=\tiny]{$\bf{1}$};
\draw (16.3,1.65) node[font=\tiny]{$\bf{0}$};

\draw (16.4,0.8) node[font=\tiny]{\bf{-1}};
\draw (16.3,-0.35) node[font=\tiny]{$\bf{0}$};

\draw (15,-0.35) node[font=\tiny]{$\bf{1}$};
\draw (14,-0.35) node[font=\tiny]{$\bf{0}$};
\draw (13,-0.35) node[font=\tiny]{$\bf{1}$};
\draw (12,-0.35) node[font=\tiny]{$\bf{0}$};
\draw (11,-0.35) node[font=\tiny]{$\bf{1}$};
\draw (10,-0.35) node[font=\tiny]{$\bf{0}$};
\draw (9,-0.35) node[font=\tiny]{$\bf{1}$};
\draw (8,-0.35) node[font=\tiny]{$\bf{0}$};
\draw (7,-0.35) node[font=\tiny]{$\bf{1}$};
\draw (6,-0.35) node[font=\tiny]{$\bf{0}$};
\draw (5,-0.35) node[font=\tiny]{$\bf{1}$};
\draw (3.7,-0.35) node[font=\tiny]{$\bf{0}$};

\draw (3.65,0.75) node[font=\tiny]{\bf{-1}};
\draw (3.7,1.65) node[font=\tiny]{$\bf{0}$};

\draw (3,1.65) node[font=\tiny]{$\bf{1}$};
\draw (1.7,1.65) node[font=\tiny]{$\bf{0}$};

\draw (1.65,2.75) node[font=\tiny]{\bf{-1}};
\draw (1.7,3.65) node[font=\tiny]{$\bf{0}$};

\draw (1,3.65) node[font=\tiny]{$\bf{1}$};

\draw (7,14) node[]{$\bf{\mathrm{Even}}\,\, \bf{\mathrm{domain} }$};
\end{scope}

\begin{scope}[xshift=8.3cm]
\draw  [draw, fill=gray!20](0,6) rectangle (1,7);
\draw  [draw, fill=gray!20](0,4) rectangle (1,5);

\draw  [draw, fill=gray!20](2,2) rectangle (3,3);
\draw  [draw, fill=gray!20](2,4) rectangle (3,5);
\draw  [draw, fill=gray!20](2,6) rectangle (3,7);
\draw  [draw, fill=gray!20](2,8) rectangle (3,9);

\draw  [draw, fill=gray!20](4,0) rectangle (5,1);
\draw  [draw, fill=gray!20](4,2) rectangle (5,3);
\draw  [draw, fill=gray!20](4,4) rectangle (5,5);
\draw  [draw, fill=gray!20](4,6) rectangle (5,7);
\draw  [draw, fill=gray!20](4,8) rectangle (5,9);
\draw  [draw, fill=gray!20](4,10) rectangle (5,11);

\draw  [draw, fill=gray!20](6,0) rectangle (7,1);
\draw  [draw, fill=gray!20](6,2) rectangle (7,3);
\draw  [draw, fill=gray!20](6,4) rectangle (7,5);
\draw  [draw, fill=gray!20](6,6) rectangle (7,7);
\draw  [draw, fill=gray!20](6,8) rectangle (7,9);
\draw  [draw, fill=gray!20](6,10) rectangle (7,11);

\draw  [draw, fill=gray!20](8,0) rectangle (9,1);
\draw  [draw, fill=gray!20](8,2) rectangle (9,3);
\draw  [draw, fill=gray!20](8,4) rectangle (9,5);
\draw  [draw, fill=gray!20](8,6) rectangle (9,7);
\draw  [draw, fill=gray!20](8,8) rectangle (9,9);

\draw  [draw, fill=gray!20](10,8) rectangle (11,9);
\draw  [draw, fill=gray!20](10,6) rectangle (11,7);
\draw  [draw, fill=gray!20](10,4) rectangle (11,5);
\draw  [draw, fill=gray!20](10,2) rectangle (11,3);
\draw  [draw, fill=gray!20](10,0) rectangle (11,1);

\draw  [draw, fill=gray!20](12,10) rectangle (13,11);
\draw  [draw, fill=gray!20](12,8) rectangle (13,9);
\draw  [draw, fill=gray!20](12,6) rectangle (13,7);
\draw  [draw, fill=gray!20](12,4) rectangle (13,5);
\draw  [draw, fill=gray!20](12,2) rectangle (13,3);
\draw  [draw, fill=gray!20](12,0) rectangle (13,1);

\draw  [draw, fill=gray!20](14,12) rectangle (15,13);
\draw  [draw, fill=gray!20](14,10) rectangle (15,11);
\draw  [draw, fill=gray!20](14,8) rectangle (15,9);
\draw  [draw, fill=gray!20](14,6) rectangle (15,7);
\draw  [draw, fill=gray!20](14,4) rectangle (15,5);
\draw  [draw, fill=gray!20](14,2) rectangle (15,3);
\draw  [draw, fill=gray!20](14,0) rectangle (15,1);

\draw  [draw, fill=gray!20](16,12) rectangle (17,13);
\draw  [draw, fill=gray!20](16,10) rectangle (17,11);
\draw  [draw, fill=gray!20](16,8) rectangle (17,9);
\draw  [draw, fill=gray!20](16,6) rectangle (17,7);
\draw  [draw, fill=gray!20](16,4) rectangle (17,5);
\draw  [draw, fill=gray!20](16,2) rectangle (17,3);

\draw  [draw, fill=gray!20](18,6) rectangle (19,7);
\draw  [draw, fill=gray!20](18,4) rectangle (19,5);

\draw  [draw, fill=gray!90](1,7) rectangle (2,8);
\draw  [draw, fill=gray!90](1,5) rectangle (2,6);

\draw  [draw, fill=gray!90](3,3) rectangle (4,4);
\draw  [draw, fill=gray!90](3,5) rectangle (4,6);
\draw  [draw, fill=gray!90](3,7) rectangle (4,8);
\draw  [draw, fill=gray!90](3,9) rectangle (4,10);

\draw  [draw, fill=gray!90](5,1) rectangle (6,2);
\draw  [draw, fill=gray!90](5,3) rectangle (6,4);
\draw  [draw, fill=gray!90](5,5) rectangle (6,6);
\draw  [draw, fill=gray!90](5,7) rectangle (6,8);
\draw  [draw, fill=gray!90](5,9) rectangle (6,10);
\draw  [draw, fill=gray!90](5,11) rectangle (6,12);

\draw  [draw, fill=gray!90](7,1) rectangle (8,2);
\draw  [draw, fill=gray!90](7,3) rectangle (8,4);
\draw  [draw, fill=gray!90](7,5) rectangle (8,6);
\draw  [draw, fill=gray!90](7,7) rectangle (8,8);
\draw  [draw, fill=gray!90](7,9) rectangle (8,10);
\draw  [draw, fill=gray!90](7,11) rectangle (8,12);

\draw  [draw, fill=gray!90](9,1) rectangle (10,2);
\draw  [draw, fill=gray!90](9,3) rectangle (10,4);
\draw  [draw, fill=gray!90](9,5) rectangle (10,6);
\draw  [draw, fill=gray!90](9,7) rectangle (10,8);
\draw  [draw, fill=gray!90](9,9) rectangle (10,10);

\draw  [draw, fill=gray!90](11,9) rectangle (12,10);
\draw  [draw, fill=gray!90](11,7) rectangle (12,8);
\draw  [draw, fill=gray!90](11,5) rectangle (12,6);
\draw  [draw, fill=gray!90](11,3) rectangle (12,4);
\draw  [draw, fill=gray!90](11,1) rectangle (12,2);

\draw  [draw, fill=gray!90](13,11) rectangle (14,12);
\draw  [draw, fill=gray!90](13,9) rectangle (14,10);
\draw  [draw, fill=gray!90](13,7) rectangle (14,8);
\draw  [draw, fill=gray!90](13,5) rectangle (14,6);
\draw  [draw, fill=gray!90](13,3) rectangle (14,4);
\draw  [draw, fill=gray!90](13,1) rectangle (14,2);

\draw  [draw, fill=gray!90](15,13) rectangle (16,14);
\draw  [draw, fill=gray!90](15,11) rectangle (16,12);
\draw  [draw, fill=gray!90](15,9) rectangle (16,10);
\draw  [draw, fill=gray!90](15,7) rectangle (16,8);
\draw  [draw, fill=gray!90](15,5) rectangle (16,6);
\draw  [draw, fill=gray!90](15,3) rectangle (16,4);
\draw  [draw, fill=gray!90](15,1) rectangle (16,2);

\draw  [draw, fill=gray!90](17,13) rectangle (18,14);
\draw  [draw, fill=gray!90](17,11) rectangle (18,12);
\draw  [draw, fill=gray!90](17,9) rectangle (18,10);
\draw  [draw, fill=gray!90](17,7) rectangle (18,8);
\draw  [draw, fill=gray!90](17,5) rectangle (18,6);
\draw  [draw, fill=gray!90](17,3) rectangle (18,4);

\draw  [draw, fill=gray!90](19,7) rectangle (20,8);
\draw  [draw, fill=gray!90](19,5) rectangle (20,6);

\draw[draw, line width=1pt] (0,4) -- (0,8) -- (2,8) -- (2,10) -- (4,10) -- (4,12) -- (8,12) -- (8,10) -- (10,10) --  (12,10) -- (12,12) -- (14,12) -- (14,14) -- (18,14) -- (18,8) -- (20,8) -- (20,4) -- (18,4) -- (18,2) -- (16,2) -- (16,0) -- (4,0)-- (4,2) -- (2,2) -- (2,4)
-- cycle;

\draw (-0.3,3.65) node[font=\tiny]{$\bf{0}$};

\draw (-0.35,5) node[font=\tiny]{\bf{-1}};
\draw (-0.3,6) node[font=\tiny]{$\bf{0}$};
\draw (-0.35,7) node[font=\tiny]{\bf{-1}};

\draw (-0.3,8.35) node[font=\tiny]{$\bf{0}$};

\draw (0.8,8.35) node[font=\tiny]{$\bf{1}$};
\draw (1.7,8.35) node[font=\tiny]{$\bf{0}$};

\draw (1.65,9.2) node[font=\tiny]{\bf{-1}};

\draw (1.7,10.35) node[font=\tiny]{$\bf{0}$};
\draw (2.8,10.35) node[font=\tiny]{$\bf{1}$};
\draw (3.7,10.35) node[font=\tiny]{$\bf{0}$};

\draw (3.65,11.2) node[font=\tiny]{\bf{-1}};

\draw (3.7,12.35) node[font=\tiny]{$\bf{0}$};
\draw (5,12.35) node[font=\tiny]{$\bf{1}$};
\draw (6,12.35) node[font=\tiny]{$\bf{0}$};
\draw (7,12.35) node[font=\tiny]{$\bf{1}$};
\draw (8.3,12.35) node[font=\tiny]{$\bf{0}$};

\draw (8.4,11.2) node[font=\tiny]{\bf{-1}};
\draw (8.3,10.35) node[font=\tiny]{$\bf{0}$};

\draw (9,10.35) node[font=\tiny]{$\bf{1}$};
\draw (10,10.35) node[font=\tiny]{$\bf{0}$};
\draw (11,10.35) node[font=\tiny]{$\bf{1}$};
\draw (11.7,10.35) node[font=\tiny]{$\bf{0}$};

\draw (11.65,11.2) node[font=\tiny]{\bf{-1}};
\draw (11.7,12.35) node[font=\tiny]{$\bf{0}$};

\draw (12.8,12.35) node[font=\tiny]{$\bf{1}$};
\draw (13.7,12.35) node[font=\tiny]{$\bf{0}$};
\draw (13.65,13.2) node[font=\tiny]{\bf{-1}};
\draw (13.7,14.35) node[font=\tiny]{$\bf{0}$};

\draw (15,14.35) node[font=\tiny]{$\bf{1}$};
\draw (16,14.35) node[font=\tiny]{$\bf{0}$};
\draw (17,14.35) node[font=\tiny]{$\bf{1}$};
\draw (18.3,14.35) node[font=\tiny]{$\bf{0}$};

\draw (18.4,13) node[font=\tiny]{\bf{-1}};
\draw (18.3,12) node[font=\tiny]{$\bf{0}$};
\draw (18.4,11) node[font=\tiny]{\bf{-1}};
\draw (18.3,10) node[font=\tiny]{$\bf{0}$};
\draw (18.4,9.2) node[font=\tiny]{\bf{-1}};
\draw (18.3,8.35) node[font=\tiny]{$\bf{0}$};

\draw (19.2,8.35) node[font=\tiny]{$\bf{1}$};
\draw (20.3,8.35) node[font=\tiny]{$\bf{0}$};

\draw (20.4,7) node[font=\tiny]{\bf{-1}};
\draw (20.3,6) node[font=\tiny]{$\bf{0}$};
\draw (20.4,5) node[font=\tiny]{\bf{-1}};
\draw (20.3,3.65) node[font=\tiny]{$\bf{0}$};

\draw (19.2,3.65) node[font=\tiny]{$\bf{1}$};
\draw (18.3,3.65) node[font=\tiny]{$\bf{0}$};

\draw (18.4,2.8) node[font=\tiny]{\bf{-1}};
\draw (18.3,1.65) node[font=\tiny]{$\bf{0}$};

\draw (17.2,1.65) node[font=\tiny]{$\bf{1}$};
\draw (16.3,1.65) node[font=\tiny]{$\bf{0}$};

\draw (16.4,0.8) node[font=\tiny]{\bf{-1}};
\draw (16.3,-0.35) node[font=\tiny]{$\bf{0}$};

\draw (15,-0.35) node[font=\tiny]{$\bf{1}$};
\draw (14,-0.35) node[font=\tiny]{$\bf{0}$};
\draw (13,-0.35) node[font=\tiny]{$\bf{1}$};
\draw (12,-0.35) node[font=\tiny]{$\bf{0}$};
\draw (11,-0.35) node[font=\tiny]{$\bf{1}$};
\draw (10,-0.35) node[font=\tiny]{$\bf{0}$};
\draw (9,-0.35) node[font=\tiny]{$\bf{1}$};
\draw (8,-0.35) node[font=\tiny]{$\bf{0}$};
\draw (7,-0.35) node[font=\tiny]{$\bf{1}$};
\draw (6,-0.35) node[font=\tiny]{$\bf{0}$};
\draw (5,-0.35) node[font=\tiny]{$\bf{1}$};
\draw (3.7,-0.35) node[font=\tiny]{$\bf{0}$};

\draw (3.65,0.75) node[font=\tiny]{\bf{-1}};
\draw (3.7,1.65) node[font=\tiny]{$\bf{0}$};

\draw (3,1.65) node[font=\tiny]{$\bf{1}$};
\draw (1.7,1.65) node[font=\tiny]{$\bf{0}$};

\draw (1.65,2.75) node[font=\tiny]{\bf{-1}};
\draw (1.7,3.65) node[font=\tiny]{$\bf{0}$};

\draw (1,3.65) node[font=\tiny]{$\bf{1}$};

\draw (7,14) node[]{$\bf{\mathrm{Hedgehog}}\,\, \bf{\mathrm{domain} }$};

\draw[draw, line width=2pt] (0,4) -- (0,8) -- (2,8) -- (2,10) -- (4,10) -- (4,12) -- (8,12) -- (8,10) -- (10,10) --  (12,10) -- (12,12) -- (14,12) -- (14,14) -- (18,14) -- (18,8) -- (20,8) -- (20,4) -- (18,4) -- (18,2) -- (16,2) -- (16,0) -- (4,0)-- (4,2) -- (2,2) -- (2,4)
-- cycle;
\draw[draw, line width=2pt] (18,10)--(16,10);
\draw (17.2,10.35) node[font=\tiny]{\bf{1}};
\draw (16.2,10.35) node[font=\tiny]{$\bf{0}$};

\draw[draw, line width=2pt] (10,10)--(10,8);
\draw (9.65,9.35) node[font=\tiny]{\bf{-1}};
\draw (9.65,8.35) node[font=\tiny]{$\bf{0}$};

\draw[draw, line width=2pt] (8,0)--(8,2);
\draw (7.65,0.75) node[font=\tiny]{\bf{-1}};
\draw (7.65,1.75) node[font=\tiny]{$\bf{0}$};

\draw[draw, line width=2pt] (12,0)--(12,2);
\draw (11.65,0.75) node[font=\tiny]{\bf{-1}};
\draw (11.65,1.75) node[font=\tiny]{$\bf{0}$};

\end{scope}
\end{tikzpicture}
\end{center}
\caption{ 
Even domain: the boundary height function is almost trivial, it varies between three values. Hedgehog domains form a subclass of even ones.}\label{b_h_f}
\end{figure}
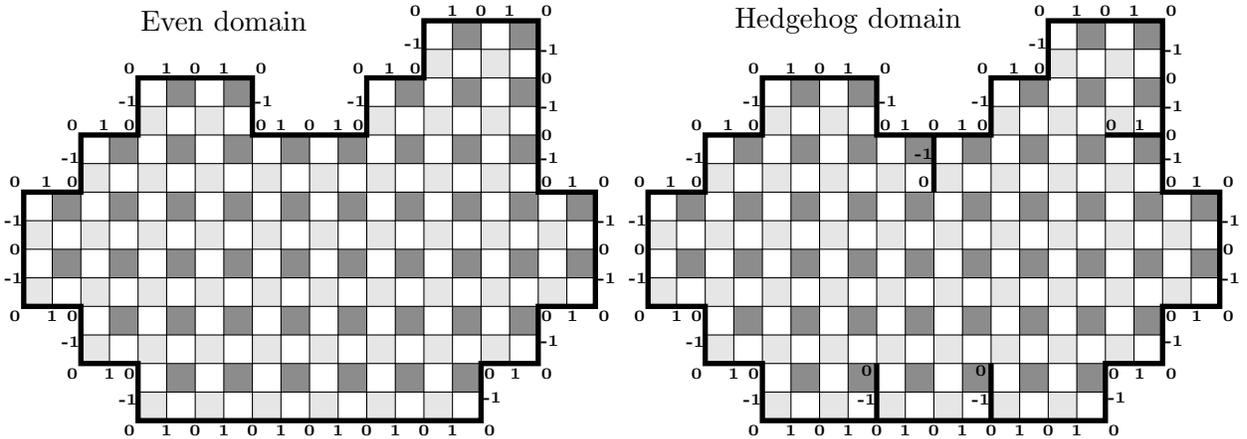

However, it seems that the most natural discretizations are given by {\it even} domains, domains with all edges of even length (e.g., see a discussion in~\cite[Section 8]{Kdom}). Such a domain obviously always has domino tilings. Furthermore, the boundary height function in this case is almost trivial, see Fig~\ref{b_h_f}. Unfortunately, in this case the boundary conditions for the coupling function are much less transparent, 
so that the following question is still open.

 \medskip
 
\noindent{\bf Open problem (\cite{Kdom}).} {\it Prove that the fluctuations of the height function on even discretizations of a planar domain converge in the scaling limit (as the mesh size tends to zero) to the Gaussian Free Field with Dirichlet boundary conditions.}

 \medskip

In particular, Temperleyan discretizations are never even domains and, more generally, situations when an even domain can be treated as a piecewise Temperleyan one are very rare.

In this paper, we introduce a special subclass of even discretizations. We call this type of domains {\it hedgehog domains}, see Fig.~\ref{b_h_f}, a formal definition is given in Section~\ref{section even}. The class of hedgehog discretizations contains {\it two-step Aztec diamonds} (see~\cite[Fig. $4$]{CEP}), which is one of the Aztec diamond-type shapes considered in~\cite{CEP}. Compared to the class of (piecewise) Temperlyan discretizations, the class of hedgehog discretizations is of a very different nature.
In particular,
\[
\lbrace \rm{(piecewise)\,\, Temperlyan\,\, discretizations\rbrace \cap \lbrace Hedgehog\,\, discretizations}\rbrace =\varnothing.
\]
 \old{One can think that Temperlyan discretizations are discretizations without white corner squares. Then the class of piecewise Temperlyan discretizations is a subclass of discretizations with fixed finite number of white corner squares as the mesh size tends to zero. While in hedgehog discretizations the number of both white and black corner squares goes to infinity as the mesh size tends to zero. 
 As for even discretizations of a given domain:
 \[
\lbrace \rm{Temperlyan\,\, discretizations\rbrace \cap \lbrace Even\,\, discretizations}\rbrace =\varnothing,
\]
\[
\varnothing\neq(\rm{\lbrace Piecewise\,\,Temperlyan\,\, discr.\rbrace \cap \lbrace Even\,\, discr.\rbrace) \varsubsetneq \lbrace Piecewise\,\,Temperlyan\,\, discr.\rbrace },
\]
 \[
\lbrace \rm{Hedgehog\,\, discretizations\rbrace \subset \lbrace Even\,\, discretizations}\rbrace.
\]
Hedgehog discretizations are of different nature 
in terms of the boundary conditions 
and the scaling limit of the coupling function as well.} 
The most conceptual difference between (piecewise) Temperlyan and hedgehog domains lies in boundary conditions satisfied by the coupling function. 
In the Temperleyan case the coupling function 
$\frac1\delta C_{\om^\delta}(\cdot,\www^\delta)$ satisfies 
$\mathrm {Re} [C_{\om^\delta}(\cdot,\www^\delta)]=0$ on the boundary, see~\cite{Kdom}. 
In the piecewise Temperleyn~\cite{rus} case one deals with mixed $\mathrm {Re} [C_{\om^\delta}(\cdot,\www^\delta)]=0 \,\, / \,\,\mathrm {Im} [C_{\om^\delta}(\cdot,\www^\delta)]=0$ boundary conditions
with the number of changes prescribed in advance. In sharp contrast, hedgehog domains give rise to Riemann-type boundary conditions of the coupling function, see Section~\ref{coupl}.
This is consistent with the fact that scaling limits of the coupling function in the Temperlyan~\cite{Kdom} and piecewise Temperlyan~\cite{rus} cases are different, though being different from each other they have the same conformal covariance~$(0,1)$, while the scaling limit of the coupling function in the case of hedgehog approximations has conformal covariance~$(\frac12, \frac12)$, see Theorem~\ref{convF} and Proposition~\ref{f}.

We prove the convergence of the fluctuations of the height function to the Gaussian Free Field in the case of hedgehog discretizations. Our result is based on the following convergence theorem for the dimer coupling function.

\old{In this paper, we introduce a special subclass of even discretizations and prove the convergence of the fluctuations of the height function to the Gaussian Free Field in this new case. We call this type of domains {\it hedgehog domains}, see Fig.~\ref{b_h_f}, a formal definition is given in Section~\ref{section even}. The class of hedgehog discretizations contains {\it two-step Aztec diamonds} (see~\cite[Fig. $4$]{CEP}), which is one of the Aztec diamond-type shapes considered in~\cite{CEP}.
Our result is based on the following convergence theorem for the dimer coupling function.}

\begin{Th}\label{main-convF}
Let $\om^\delta$ be a sequence of discrete hedgehog domains of mesh size $\delta$ approximating a simply connected domain $\om$.  Let $\www^{\delta}$ approximate a point $v\in\om$. Then, $\frac1\delta C_{\om^\delta}(\cdot,\www^\delta)$ converges uniformly on compact subsets of  $\om\smallsetminus\{v\}$ to a continuous holomorphic function satisfying Riemann-type boundary conditions 
\[
\im[f(z)\sqrt{n(z)}]=0,\quad z\in\dom,
\]
where $n(z)$ is the outer normal to the boundary at $z$.
\end{Th}

For a more precise statement, see Theorem~\ref{convF}. In particular, note that we do {\it not} assume that the boundary $\dom$ is smooth. 
The main advantage of hedgehog-type discretizations is that the primitive of the square of 
s-holomorphic version of the coupling function is constant at (a half of) boundary vertices, see Proposition~\ref{half_dirichlet}. 
This allows to obtain, after slight modification, Riemann boundary conditions for the coupling function on a discrete level, see Proposition~\ref{second_half_dirichlet}, and to use discrete complex analysis techniques originally developed in the Ising model context~\cite{Stas07}.

Due to~\cite{KGff} the dimer height function converges to the Gaussian Free Field in the setup of Theorem~\ref{main-convF}.

 \begin{cor}\label{main-cor2}
 Let $\om$ be a bounded simply connected domain in $\mathbb{C}$, whose boundary $\dom$ contains a straight segment. Let $\om^\delta$ be a hedgehog domain approximating $\om$. Let $h^\delta$ be the height function of~$\om^\delta$. Then $h^\delta - \mathbb{E}{h^\delta}$ converges weakly in distribution to the Gaussian Free Field on $\om$ with Dirichlet boundary conditions, as $\delta$ tends to $0$.
 \end{cor}

The existence of a straight part of the boundary is a technical condition, which we use in the proof. However, we believe that it can be relaxed.

In addition, we generalise the result of~\cite{rus}, on convergence of the mean values of the double-dimer height functions for hedgehog domains.  A double-dimer configuration is the union of two dimer coverings, we will consider coverings of a pair of domains $(\om^\delta$, $\widehat{\om}^\delta)$ that differ by two squares. Namely, let $\widehat{\om}^\delta$ be obtained from $\om^\delta$ by removing black and white squares $u_0$ and $v_0$, adjacent to the boundary. We define the {\it double-dimer coupling function} on $\om^\delta=\om^\delta\cup\widehat{\om}^\delta$ as the difference of the two dimer coupling functions on domains $\om^\delta$ and $\widehat{\om}^\delta$
\[C_{\operatorname{dbl-d}, \om^\delta}(u, v) := C_{\om^\delta}(u, v) - C_{\widehat{\om}^\delta}(u, v).\]
Similarly, the height function $h_{\operatorname{dbl-d}, \om^\delta}$ in the double-dimer model is the difference of the two height functions corresponding to two independent uniform dimer coverings of the domains $\om^\delta$ and $\widehat{\om}^\delta$.

In~\cite{rus} it was shown that in the double-dimer model the coupling function $C_{\operatorname{dbl-d}}(u,v)$ has a factorization into a product of two discrete holomorphic functions $\F(u)$ and $\G(v)$, and therefore for any discrete domain the expectation of the height function of the double-dimer model can be interpreted as a primitive of $\F(u)\G(v)$.
In the case of hedgehog approximations the functions $F$ and $G$ solve the discrete Riemann boundary value problems described in Section~\ref{RH}. This allows us to prove the following theorem.

\begin{Th}\label{double} 
Let $\om$ be a bounded simply connected domain. Let $\bbb$ and $\www$ be points on straight parts of the boundary of $\om$.  Suppose that a sequence of discrete hedgehog domains $\om^\delta$ on a grid with mesh size $\delta$ approximates the domain $\om$. 
 Let black and white squares $\bbb^\delta$ and $\www^\delta$ of the domain $\om^\delta$ tend to boundary points $\bbb$ and $\www$ of the domain $\om$. Finally, let $h^\delta$ be the height function of a uniform double-dimer configuration on $\om$. Then $\mathbb{E}h^\delta$ converges to the harmonic measure $\operatorname{hm}_{\om}(\,\cdot\,, (\bbb\www))$ of the boundary arc~$(\bbb\www)$ on the domain $\om$.
 \end{Th}

 \bigskip
 
\noindent{\bf Organization of the paper.} The paper is organised as follows. The notation and basic definitions are given in Section~\ref{not}. Section~\ref{4}  is devoted to the Riemann-type boundary value problem. In Section~\ref{3}, we show that the coupling function satisfies Riemann-type boundary conditions on hedgehog domains, and we prove Theorem~\ref{main-convF}.
 Section~\ref{5} contains the proof of Corollary~\ref{main-cor2}. Finally, in Section~\ref{6} we prove Theorem~\ref{double}. 
\bigskip
\medskip











\noindent{\bf Acknowledgements.} 
The author thanks Dmitry Chelkak for many helpful discussions, and Stanislav Smirnov for valuable insights and support. The author also thanks the anonymous referee for helpful comments and suggestions. Research of Section~\ref{6} was supported by the Russian Science Foundation grant 19-71-30002. 
Part of the work was done while the author was at the University of Geneva, supported by the NCCR SwissMAP of the SNSF and ERC AG COMPASP.  The author also received partial support from the Swiss NSF grant P2GEP2\_184555.
\medskip

\section{Notation}\label{not}
We will use the same notations as in~\cite{rus}. Put $\lambda = e^{i\frac{\pi}{4}}$ and $\bar{\lambda} = e^{-i\frac{\pi}{4}}$.
Consider a checkerboard tiling $\mathbb{C}^\delta$ of $\mathbb{C}$ with  
squares, each square has side $\delta$ and centered at a lattice point of $$\left\{\left(\frac{\delta n}{\sqrt{2}},\frac{\delta m}{\sqrt{2}}\right)\, | \,n, m \in \mathbb{Z}; n+m \in 2\mathbb{Z}\right\}$$
(see~Fig.~\ref{rotation}). The pair $(n,m)$ is called the coordinates of a point on this lattice. 
A discrete bounded simply connected domain $\om^\delta$ is defined as a non-empty bounded connected component of the complement of 
a connected union of edges 
of the lattice~$\mathbb{C}^\delta$. The boundary of such a domain is the boundary of corresponding open set. 
Note that $\om^\delta$ is allowed to have slits on the boundary, see~Fig.~\ref{rotation}. Let us call the edges on the slits of the boundary {\it slit boundary edges}.
Let $\V{}^\delta$ be the vertex set of $\om^\delta$. 
We will denote the set of black squares by $\bb{}^\delta$ and the set of white squares of $\om^\delta$ by $\ww{}^\delta$. Thus, $\om^\delta = \bb{}^\delta \sqcup \ww{}^\delta$. Let the coordinates of a square be the coordinates of its center. Then we can define the sets $\bb{0}^\delta$ and $\bb{1}^\delta$ of black squares of $\om^\delta$ and the sets $\ww{0}^\delta$ and $\ww{1}^\delta$  of white squares by the following properties:
\begin{enumerate}
\item[($\bb{0}^\delta$)] both coordinates are even and the sum of coordinates is divisible by~$4$;
\item[($\bb{1}^\delta$)] both coordinates are even and the sum of coordinates is not divisible by~$4$;
\item[($\ww{0}^\delta$)] both coordinates are odd and the sum of coordinates is not divisible by~$4$;
\item[($\ww{1}^\delta$)] both coordinates are odd and the sum of coordinates is divisible by~$4$.
\end{enumerate}
Let us divide the set vertex $\mathcal{V}^\delta$ into three sets $\mathcal{V}_\circ^\delta$ , $\mathcal{V}_\bullet^\delta$  and $\mathcal{V}_\diamond^\delta$ as is shown on Fig.~\ref{rotation}. 

Define $\dV^\delta$ to be the set of vertices on the boundary. 
Let $\partial_{\operatorname{out}}\om^\delta$ be the set of squares adjacent to~$\om^\delta$ but not in $\om^\delta$. Let us glue to the slits the copy of each square of $\om^\delta$ adjacent to a slit boundary edge. Denote this set of glued squares by $\partial_{\operatorname{slit}}\om^\delta$.
Let $\dom^\delta=\partial_{\operatorname{out}}\om^\delta \sqcup \partial_{\operatorname{slit}}\om^\delta$.
In other words, $\dom^\delta$ is an abstract set of squares which are across the boundary from the squares of $\om^\delta$.
Let $\db ^\delta$ and $\dw^\delta$ be the sets of black and white faces of $\dom^\delta$ correspondingly.  
Let $\tilde{\partial}\mathcal{V}^\delta$ be the set of vertices adjacent to $\dom^\delta$ but not in $\dV^\delta$.
Let $\diom^\delta$ be the set of interior faces that have a common edge with the boundary of $\om^\delta$. Similarly define sets $\dib^\delta$ and $\diw^\delta$ ($\diom^\delta=\dib^\delta\sqcup\diw^\delta$). Let $\intom^\delta=\om^\delta\smallsetminus\diom^\delta$. Let us denote by $\clom^\delta$ the set $\om^\delta\sqcup\dom^\delta$, define also sets $\clb^\delta$ and $\clw^\delta$, to be exact: $\clb^\delta=\bb{}^\delta\sqcup\db^\delta$ and $\clw^\delta=\ww{}^\delta\sqcup\dw^\delta$.   
In the same way we define the sets 
$\partial\mathcal{V}_\circ^\delta$,  
$\partial\mathcal{V}_\bullet^\delta$,  
$\partial\mathcal{V}_\diamond^\delta$,
$\tilde\partial\mathcal{V}_\circ^\delta$,  
$\tilde\partial\mathcal{V}_\bullet^\delta$,  
$\tilde\partial\mathcal{V}_\diamond^\delta$,
$\dbb{0,1}^\delta$, 
$\dww{0,1}^\delta$, 
$\dintb{0,1}^\delta$, 
$\dintw{0,1}^\delta$,
$\intb{0,1}^\delta$, $\intw{0,1}^\delta$,
$\clbb{0,1}^\delta$ and $\clww{0,1}^\delta$.

We say that a discrete domain $\om^\delta$ approximates a simply connected domain $\om$ if $\om^\delta\to\om$
in the sense of Carath\'eodory, see~\cite[Section 3.2]{C+S}.

\begin{figure}
\begin{center}
\begin{tikzpicture}[x={(0.46cm,0.46cm)}, y={(0.46cm,-0.46cm)}]

\draw  [draw, fill=gray!20](0,6) --(0,7)-- (1,7)--(1,6)--cycle;
\draw  [draw, fill=gray!20](0,4)--(0,5)-- (1,5)--(1,4)--cycle;

\draw  [draw, fill=gray!20](2,2)  --(2,3)-- (3,3)--(3,2)--cycle;
\draw  [draw, fill=gray!20](2,4)  --(2,5)-- (3,5)--(3,4)--cycle;
\draw  [draw, fill=gray!20](2,6)  --(2,7)-- (3,7)--(3,6)--cycle;
\draw  [draw, fill=gray!20](2,8)  --(2,9)-- (3,9)--(3,8)--cycle;

\draw  [draw, fill=gray!20](4,0) --(4,1)-- (5,1)--(5,0)--cycle;
\draw  [draw, fill=gray!20](4,2) --(4,3)-- (5,3)--(5,2)--cycle;
\draw  [draw, fill=gray!20](4,4) --(4,5)-- (5,5)--(5,4)--cycle;
\draw  [draw, fill=gray!20](4,6) --(4,7)-- (5,7)--(5,6)--cycle;
\draw  [draw, fill=gray!20](4,8) --(4,9)-- (5,9)--(5,8)--cycle;
\draw  [draw, fill=gray!20](4,10) --(4,11)-- (5,11)--(5,10)--cycle;

\draw  [draw, fill=gray!20](6,0) --(6,1)-- (7,1)--(7,0)--cycle;
\draw  [draw, fill=gray!20](6,2) --(6,3)-- (7,3)--(7,2)--cycle;
\draw  [draw, fill=gray!20](6,4) --(6,5)-- (7,5)--(7,4)--cycle;
\draw  [draw, fill=gray!20](6,6) --(6,7)-- (7,7)--(7,6)--cycle;
\draw  [draw, fill=gray!20](6,8) --(6,9)-- (7,9)--(7,8)--cycle;
\draw  [draw, fill=gray!20](6,10) --(6,11)-- (7,11)--(7,10)--cycle;

\draw  [draw, fill=gray!20](8,0) --(8,1)--  (9,1)--(9,0)--cycle;
\draw  [draw, fill=gray!20](8,2) --(8,3)-- (9,3)--(9,2)--cycle;
\draw  [draw, fill=gray!20](8,4) --(8,5)-- (9,5)--(9,4)--cycle;
\draw  [draw, fill=gray!20](8,6) --(8,7)-- (9,7)--(9,6)--cycle;
\draw  [draw, fill=gray!20](8,8) --(8,9)-- (9,9)--(9,8)--cycle;

\draw  [draw, fill=gray!20](10,8) --(10,9)-- (11,9)--(11,8)--cycle;
\draw  [draw, fill=gray!20](10,6) --(10,7)-- (11,7)--(11,6)--cycle;
\draw  [draw, fill=gray!20](10,4) --(10,5)-- (11,5)--(11,4)--cycle;
\draw  [draw, fill=gray!20](10,2) --(10,3)-- (11,3)--(11,2)--cycle;
\draw  [draw, fill=gray!20](10,0) --(10,1)-- (11,1)--(11,0)--cycle;

\draw  [draw, fill=gray!20](12,10) --(12,11)-- (13,11)--(13,10)--cycle;
\draw  [draw, fill=gray!20](12,8) --(12,9)-- (13,9)--(13,8)--cycle;
\draw  [draw, fill=gray!20](12,6) --(12,7)-- (13,7)--(13,6)--cycle;
\draw  [draw, fill=gray!20](12,4) --(12,5)-- (13,5)--(13,4)--cycle;
\draw  [draw, fill=gray!20](12,2) --(12,3)-- (13,3)--(13,2)--cycle;
\draw  [draw, fill=gray!20](12,0) --(12,1)-- (13,1)--(13,0)--cycle;

\draw  [draw, fill=gray!20](14,12) --(14,13)-- (15,13)--(15,12)--cycle;
\draw  [draw, fill=gray!20](14,10) --(14,11)-- (15,11)--(15,10)--cycle;
\draw  [draw, fill=gray!20](14,8) --(14,9)-- (15,9)--(15,8)--cycle;
\draw  [draw, fill=gray!20](14,6) --(14,7)-- (15,7)--(15,6)--cycle;
\draw  [draw, fill=gray!20](14,4) --(14,5)-- (15,5)--(15,4)--cycle;
\draw  [draw, fill=gray!20](14,2) --(14,3)-- (15,3)--(15,2)--cycle;
\draw  [draw, fill=gray!20](14,0) --(14,1)-- (15,1)--(15,0)--cycle;

\draw  [draw, fill=gray!20](18,8) --(18,9)-- (19,9)--(19,8)--cycle;
\draw  [draw, fill=gray!20](20,6) --(20,7)-- (21,7)--(21,6)--cycle;
\draw  [draw, fill=gray!20](20,4) --(20,5)-- (21,5)--(21,4)--cycle;
\draw  [draw, fill=gray!20](19,3) --(19,4)-- (20,4)--(20,3)--cycle;

\draw  [draw, fill=gray!20](18,2) --(18,3)-- (19,3)--(19,2)--cycle;
\draw  [draw, fill=gray!20](17,1) --(17,2)-- (18,2)--(18,1)--cycle;
\draw  [draw, fill=gray!20](16,0) --(16,1)-- (17,1)--(17,0)--cycle;
\draw  [draw, fill=gray!20](15,-1) --(15,0)-- (16,0)--(16,-1)--cycle;
\draw [draw] (19,9)--(20,9)--(20,8)--(21,8)--(21,4);
\draw [draw] (14, 0)--(14,-1)--(15,-1);
\path[draw, fill=black] (19,9) circle[radius=0.07cm];
\path[draw, fill=black] (21,7) circle[radius=0.07cm];
\path[draw, fill=black] (21,5) circle[radius=0.07cm];
\path[draw, fill=black] (19,3) circle[radius=0.07cm];
\path[draw, fill=black] (17,1) circle[radius=0.07cm];
\path[draw, fill=black] (15,-1) circle[radius=0.07cm];


\draw  [draw, fill=gray!20](16,12) --(16,13)-- (17,13)--(17,12)--cycle;
\draw  [draw, fill=gray!20](16,10) --(16,11)-- (17,11)--(17,10)--cycle;
\draw  [draw, fill=gray!20](16,8) --(16,9)-- (17,9)--(17,8)--cycle;
\draw  [draw, fill=gray!20](16,6) --(16,7)-- (17,7)--(17,6)--cycle;
\draw  [draw, fill=gray!20](16,4) --(16,5)-- (17,5)--(17,4)--cycle;
\draw  [draw, fill=gray!20](16,2) --(16,3)-- (17,3)--(17,2)--cycle;

\draw  [draw, fill=gray!20](18,6) --(18,7)-- (19,7)--(19,6)--cycle;
\draw  [draw, fill=gray!20](18,4) --(18,5)-- (19,5)--(19,4)--cycle;

\draw  [draw, fill=gray!20](1,7) --(1,8)-- (2,8)--(2,7)--cycle;
\draw  [draw, fill=gray!20](1,5) --(1,6)-- (2,6)--(2,5)--cycle;

\draw  [draw, fill=gray!20](3,3) --(3,4)-- (4,4)--(4,3)--cycle;
\draw  [draw, fill=gray!20](3,5) --(3,6)-- (4,6)--(4,5)--cycle;
\draw  [draw, fill=gray!20](3,7) --(3,8)-- (4,8)--(4,7)--cycle;
\draw  [draw, fill=gray!20](3,9) --(3,10)-- (4,10)--(4,9)--cycle;

\draw  [draw, fill=gray!20](5,1) --(5,2)-- (6,2)--(6,1)--cycle;
\draw  [draw, fill=gray!20](5,3) --(5,4)-- (6,4)--(6,3)--cycle;
\draw  [draw, fill=gray!20](5,5) --(5,6)-- (6,6)--(6,5)--cycle;
\draw  [draw, fill=gray!20](5,7) --(5,8)-- (6,8)--(6,7)--cycle;
\draw  [draw, fill=gray!20](5,9) --(5,10)-- (6,10)--(6,9)--cycle;
\draw  [draw, fill=gray!20](5,11) --(5,12)-- (6,12)--(6,11)--cycle;

\draw  [draw, fill=gray!20](7,1) --(7,2)-- (8,2)--(8,1)--cycle;
\draw  [draw, fill=gray!20](7,3) --(7,4)-- (8,4)--(8,3)--cycle;
\draw  [draw, fill=gray!20](7,5) --(7,6)-- (8,6)--(8,5)--cycle;
\draw  [draw, fill=gray!20](7,7) --(7,8)-- (8,8)--(8,7)--cycle;
\draw  [draw, fill=gray!20](7,9) --(7,10)-- (8,10)--(8,9)--cycle;
\draw  [draw, fill=gray!20](7,11) --(7,12)-- (8,12)--(8,11)--cycle;

\draw  [draw, fill=gray!20](9,1) --(9,2)-- (10,2)--(10,1)--cycle;
\draw  [draw, fill=gray!20](9,3) --(9,4)-- (10,4)--(10,3)--cycle;
\draw  [draw, fill=gray!20](9,5) --(9,6)-- (10,6)--(10,5)--cycle;
\draw  [draw, fill=gray!20](9,7) --(9,8)-- (10,8)--(10,7)--cycle;
\draw  [draw, fill=gray!20](9,9) --(9,10)-- (10,10)--(10,9)--cycle;

\draw  [draw, fill=gray!20](11,1) --(11,2)-- (12,2)--(12,1)--cycle;
\draw  [draw, fill=gray!20](11,3) --(11,4)-- (12,4)--(12,3)--cycle;
\draw  [draw, fill=gray!20](11,5) --(11,6)-- (12,6)--(12,5)--cycle;
\draw  [draw, fill=gray!20](11,7) --(11,8)-- (12,8)--(12,7)--cycle;
\draw  [draw, fill=gray!20](11,9) --(11,10)-- (12,10)--(12,9)--cycle;

\draw  [draw, fill=gray!20](13,1) --(13,2)-- (14,2)--(14,1)--cycle;
\draw  [draw, fill=gray!20](13,3) --(13,4)-- (14,4)--(14,3)--cycle;
\draw  [draw, fill=gray!20](13,5) --(13,6)-- (14,6)--(14,5)--cycle;
\draw  [draw, fill=gray!20](13,7) --(13,8)-- (14,8)--(14,7)--cycle;
\draw  [draw, fill=gray!20](13,9) --(13,10)-- (14,10)--(14,9)--cycle;
\draw  [draw, fill=gray!20](13,11) --(13,12)-- (14,12)--(14,11)--cycle;

\draw  [draw, fill=gray!20](15,1) --(15,2)-- (16,2)--(16,1)--cycle;
\draw  [draw, fill=gray!20](15,3) --(15,4)-- (16,4)--(16,3)--cycle;
\draw  [draw, fill=gray!20](15,5) --(15,6)-- (16,6)--(16,5)--cycle;
\draw  [draw, fill=gray!20](15,7) --(15,8)-- (16,8)--(16,7)--cycle;
\draw  [draw, fill=gray!20](15,9) --(15,10)-- (16,10)--(16,9)--cycle;
\draw  [draw, fill=gray!20](15,11) --(15,12)-- (16,12)--(16,11)--cycle;
\draw  [draw, fill=gray!20](15,13) --(15,14)-- (16,14)--(16,13)--cycle;

\draw  [draw, fill=gray!20](17,3) --(17,4)-- (18,4)--(18,3)--cycle;
\draw  [draw, fill=gray!20](17,5) --(17,6)-- (18,6)--(18,5)--cycle;
\draw  [draw, fill=gray!20](17,7) --(17,8)-- (18,8)--(18,7)--cycle;
\draw  [draw, fill=gray!20](17,9) --(17,10)-- (18,10)--(18,9)--cycle;
\draw  [draw, fill=gray!20](17,11) --(17,12)-- (18,12)--(18,11)--cycle;
\draw  [draw, fill=gray!20](17,13) --(17,14)-- (18,14)--(18,13)--cycle;

\draw  [draw, fill=gray!20](19,5) --(19,6)-- (20,6)--(20,5)--cycle;
\draw  [draw, fill=gray!20](19,7) --(19,8)-- (20,8)--(20,7)--cycle;

\draw[draw,dashed, line width=0.5pt] (2,8)--(8,2);
\draw[draw,dashed, line width=0.5pt] (4,10)--(12,2);
\draw[draw,dashed, line width=0.5pt] (10,8)--(16,2);
\draw[draw,dashed, line width=0.5pt] (12,10)--(18,4);
\draw[draw,dashed, line width=0.5pt] (14,12)--(16,10);

\draw[draw,dashed, line width=0.5pt] (2,4)--(8,10);
\draw[draw,dashed, line width=0.5pt] (4,2)--(10,8);
\draw[draw,dashed, line width=0.5pt] (8,2)--(16,10);
\draw[draw,dashed, line width=0.5pt] (12,2)--(18,8);

\draw[draw,dashed, line width=0.5pt] (0,6)--(6,0)--(8,2)--(10,0)--(12,2)--(14,0)--(20,6)--(16,10)--(18,12)--(16,14)--(10,8)--(6,12)--cycle;

\path (0.5,6.5) node[]{$\bb{1}^{\flat}$};
\path (0.5,4.5) node[]{$\bb{1}^{\text{-}}$};
\path (0.5,5.5) node[]{$\ww{0}^{\text{-}}$};
\path (1.5,4.5) node[]{$\ww{1}^{\text{-}}$};

\path (2.5,2.5) node[]{$\bb{1}^{\text{-}}$};
\path (2.5,3.5) node[]{$\ww{0}^{\text{-}}$};
\path (3.5,2.5) node[]{$\ww{1}^{\text{-}}$};
\path (2.5,8.5) node[]{$\bb{1}^{\flat}$};

\path (4.5,0.5) node[]{$\bb{1}^{\text{-}}$};
\path (4.5,1.5) node[]{$\ww{0}^{\text{-}}$};
\path (5.5,0.5) node[]{$\ww{1}^{\text{-}}$};
\path (4.5,10.5) node[]{$\bb{1}^{\flat}$};

\path (6.5,0.5) node[]{$\bb{1}^{\sharp}$};
\path (7.5,0.5) node[]{$\ww{1}^{\sharp}$};
\path (6.5,6.5) node[]{$\bf \uI$};
\path (5.5,6.5) node[]{$\bf \vlbr$};

\path (8.5,0.5) node[]{$\bb{1}^{\text{-}}$};
\path (8.5,1.5) node[]{$\ww{0}^{\text{-}}$};
\path (9.5,0.5) node[]{$\ww{1}^{\text{-}}$};
\path (8.5,4.5) node[]{$\bb{1}$};

\path (10.5,0.5) node[]{$\bb{1}^{\sharp}$};
\path (11.5,0.5) node[]{$\ww{1}^{\sharp}$};
\path (10.5,2.5) node[]{$\bb{1}$};
\path (10.5,4.5) node[]{$\bb{1}$};
\path (10.5,6.5) node[]{$\bb{1}$};
\path (10.5,8.5) node[]{$\bb{1}^{\flat}$};

\path (12.5,0.5) node[]{$\bb{1}^{\text{-}}$};
\path (12.5,1.5) node[]{$\ww{0}^{\text{-}}$};
\path (13.5,0.5) node[]{$\ww{1}^{\text{-}}$};
\path (12.5,4.5) node[]{$\bb{1}$};
\path (12.5,6.5) node[]{$\bb{1}$};
\path (12.5,10.5) node[]{$\bb{1}^{\flat}$};

\path (14.5,0.5) node[]{$\bb{1}^{\sharp}$};
\path (15.5,0.5) node[]{$\ww{1}^{\sharp}$};
\path (14.5,12.5) node[]{$\bb{1}^{\flat}$};

\path (1.5,7.5) node[]{$\bb{0}^{\flat}$};
\path (0.5,7.5) node[]{$\ww{0}^{\flat}$};

\path (3.5,9.5) node[]{$\bb{0}^{\flat}$};
\path (2.5,9.5) node[]{$\ww{0}^{\flat}$};

\path (5.4,7.6) node[]{$\bf \uR$};
\path (6.5,7.5) node[]{$\bf \vl$};
\path (5.5,11.5) node[]{$\bb{0}^{\flat}$};
\path (4.5,11.5) node[]{$\ww{0}^{\flat}$};

\path (7.5,1.5) node[]{$\bb{0}^{\sharp}$};
\path (7.5,11.5) node[]{$\bb{0}^{+}$};
\path (6.5,11.5) node[]{$\ww{0}^{+}$};
\path (7.5,10.5) node[]{$\ww{1}^{+}$};

\path (9.5,3.5) node[]{$\bb{0}$};
\path (9.5,5.5) node[]{$\bb{0}$};
\path (9.5,9.5) node[]{$\bb{0}^{+}$};
\path (8.5,9.5) node[]{$\ww{0}^{+}$};
\path (9.5,8.5) node[]{$\ww{1}^{+}$};

\path (11.5,1.5) node[]{$\bb{0}^{\sharp}$};
\path (11.5,3.5) node[]{$\bb{0}$};
\path (11.5,5.5) node[]{$\bb{0}$};
\path (11.5,9.5) node[]{$\bb{0}^{\flat}$};
\path (10.5,9.5) node[]{$\ww{0}^{\flat}$};

\path (13.5,3.5) node[]{$\bb{0}$};
\path (13.5,5.5) node[]{$\bb{0}$};
\path (13.5,11.5) node[]{$\bb{0}^{\flat}$};
\path (12.5,11.5) node[]{$\ww{0}^{\flat}$};

\path (15.5,1.5) node[]{$\bb{0}^{\sharp}$};
\path (15.5,13.5) node[]{$\bb{0}^{\flat}$};
\path (14.5,13.5) node[]{$\ww{0}^{\flat}$};

\path (16.5,2.5) node[]{$\bb{1}^{\sharp}$};
\path (17.5,2.5) node[]{$\ww{1}^{\sharp}$};
\path (16.5,10.5) node[]{$\bb{1}^{\sharp}$};
\path (17.5,10.5) node[]{$\ww{1}^{\sharp}$};
\path (16.5,12.5) node[]{$\bb{1}$};

\path (17.5,3.5) node[]{$\bb{0}^{\sharp}$};
\path (17.5,9.5) node[]{$\bb{0}^{+}$};
\path (16.5,9.5) node[]{$\ww{0}^{+}$};
\path (17.5,8.5) node[]{$\ww{1}^{+}$};
\path (17.5,11.5) node[]{$\bb{0}^{\sharp}$};
\path (17.5,13.5) node[]{$\bb{0}^{+}$};
\path (16.5,13.5) node[]{$\ww{0}^{+}$};
\path (17.5,12.5) node[]{$\ww{1}^{+}$};

\path (18.5,4.5) node[]{$\bb{1}^{\sharp}$};
\path (19.5,4.5) node[]{$\ww{1}^{\sharp}$};

\path (19.5,5.5) node[]{$\bb{0}^{\sharp}$};
\path (19.5,7.5) node[]{$\bb{0}^{+}$};
\path (18.5,7.5) node[]{$\ww{0}^{+}$};
\path (19.5,6.5) node[]{$\ww{1}^{+}$};

\path[draw, fill=white] (6,6) circle[radius=0.07cm];
\path[draw, fill=white] (8,6) circle[radius=0.07cm];
\path[draw, fill=white] (10,6) circle[radius=0.07cm];
\path[draw, fill=white] (12,6) circle[radius=0.07cm];
\path[draw, fill=white] (14,6) circle[radius=0.07cm];

\path[draw, fill=white] (8,4) circle[radius=0.07cm];
\path[draw, fill=white] (10,4) circle[radius=0.07cm];
\path[draw, fill=white] (12,4) circle[radius=0.07cm];
\path[draw, fill=white] (14,4) circle[radius=0.07cm];

\path[draw, fill=white] (10,2) circle[radius=0.07cm];

\path[draw, fill=white] (6,8) circle[radius=0.07cm];



\draw[draw, line width=2pt] (0,4) -- (0,8) -- (2,8) -- (2,10) -- (4,10) -- (4,12) -- (8,12) -- (8,10) -- (10,10) --  (12,10) -- (12,12) -- (14,12) -- (14,14) -- (18,14) -- (18,8) -- (20,8) -- (20,4) -- (18,4) -- (18,2) -- (16,2) -- (16,0) -- (4,0)-- (4,2) -- (2,2) -- (2,4)
-- cycle;
\draw[draw, line width=2pt] (18,10)--(16,10);
\draw[draw, line width=2pt] (10,10)--(10,8);
\draw[draw, line width=2pt] (8,0)--(8,2);
\draw[draw, line width=2pt] (12,0)--(12,2);

\path[draw, fill=white] (0,6) circle[radius=0.07cm];
\path[draw, fill=white] (0,8) circle[radius=0.07cm];
\path[draw, fill=white] (2,8) circle[radius=0.07cm];
\path[draw, fill=white] (2,10) circle[radius=0.07cm];
\path[draw, fill=white] (4,10) circle[radius=0.07cm];
\path[draw, fill=white] (4,12) circle[radius=0.07cm];
\path[draw, fill=white] (6,12) circle[radius=0.07cm];
\path[draw, fill=white] (8,12) circle[radius=0.07cm];
\path[draw, fill=white] (8,10) circle[radius=0.07cm];
\path[draw, fill=white] (10,8) circle[radius=0.07cm];
\path[draw, fill=white] (10,10) circle[radius=0.07cm];
\path[draw, fill=white] (12,10) circle[radius=0.07cm];
\path[draw, fill=white] (12,12) circle[radius=0.07cm];
\path[draw, fill=white] (14,12) circle[radius=0.07cm];
\path[draw, fill=white] (14,14) circle[radius=0.07cm];
\path[draw, fill=white] (16,14) circle[radius=0.07cm];
\path[draw, fill=white] (18,14) circle[radius=0.07cm];
\path[draw, fill=white] (18,12) circle[radius=0.07cm];
\path[draw, fill=white] (16,10) circle[radius=0.07cm];
\path[draw, fill=white] (18,10) circle[radius=0.07cm];
\path[draw, fill=white] (18,8) circle[radius=0.07cm];
\path[draw, fill=white] (20,8) circle[radius=0.07cm];
\path[draw, fill=white] (20,6) circle[radius=0.07cm];
\path[draw, fill=white] (20,4) circle[radius=0.07cm];
\path[draw, fill=white] (18,4) circle[radius=0.07cm];
\path[draw, fill=white] (18,2) circle[radius=0.07cm];
\path[draw, fill=white] (16,2) circle[radius=0.07cm];
\path[draw, fill=white] (16,0) circle[radius=0.07cm];
\path[draw, fill=white] (14,0) circle[radius=0.07cm];
\path[draw, fill=white] (12,2) circle[radius=0.07cm];
\path[draw, fill=white] (12,0) circle[radius=0.07cm];
\path[draw, fill=white] (10,0) circle[radius=0.07cm];
\path[draw, fill=white] (8,2) circle[radius=0.07cm];
\path[draw, fill=white] (8,0) circle[radius=0.07cm];
\path[draw, fill=white] (6,0) circle[radius=0.07cm];
\path[draw, fill=white] (4,0) circle[radius=0.07cm];
\path[draw, fill=white] (4,2) circle[radius=0.07cm];
\path[draw, fill=white] (2,2) circle[radius=0.07cm];
\path[draw, fill=white] (2,4) circle[radius=0.07cm];
\path[draw, fill=white] (0,4) circle[radius=0.07cm];

\path[draw, fill=black] (13,5) circle[radius=0.07cm];
\path[draw, fill=black] (13,3) circle[radius=0.07cm];

\path[draw, fill=black] (9,3) circle[radius=0.06cm];
\path[draw, fill=black] (11,3) circle[radius=0.06cm];
\path[draw, fill=black] (13,3) circle[radius=0.06cm];
\path[draw, fill=black] (9,5) circle[radius=0.06cm];
\path[draw, fill=black] (11,5) circle[radius=0.06cm];
\path[draw, fill=black] (13,5) circle[radius=0.06cm];

\path[draw, fill=white] (9.9,2.9)--(10.1,2.9)--(10.1,3.1)--(9.9,3.1)--cycle;
\path[draw, fill=white] (11.9,2.9)--(12.1,2.9)--(12.1,3.1)--(11.9,3.1)--cycle;
\path[draw, fill=white] (8.9,3.9)--(9.1,3.9)--(9.1,4.1)--(8.9,4.1)--cycle;
\path[draw, fill=white] (10.9,3.9)--(11.1,3.9)--(11.1,4.1)--(10.9,4.1)--cycle;
\path[draw, fill=white] (12.9,3.9)--(13.1,3.9)--(13.1,4.1)--(12.9,4.1)--cycle;

\path[draw, fill=white] (9.9,4.9)--(10.1,4.9)--(10.1,5.1)--(9.9,5.1)--cycle;
\path[draw, fill=white] (11.9,4.9)--(12.1,4.9)--(12.1,5.1)--(11.9,5.1)--cycle;

\path[draw, fill=white] (8.9,5.9)--(9.1,5.9)--(9.1,6.1)--(8.9,6.1)--cycle;
\path[draw, fill=white] (10.9,5.9)--(11.1,5.9)--(11.1,6.1)--(10.9,6.1)--cycle;

\path (10.5,3.5) node[]{$\ww{0}$};
\path (10.5,5.5) node[]{$\ww{0}$};
\path (9.5,6.5) node[]{$\ww{1}$};
\path (9.5,4.5) node[]{$\ww{1}$};
\path (12.5,5.5) node[]{$\ww{0}$};
\path (12.5,3.5) node[]{$\ww{0}$};
\path (11.5,4.5) node[]{$\ww{1}$};

\path[draw, fill=white] (5.9,6.9)--(6.1,6.9)--(6.1,7.1)--(5.9,7.1)--cycle;
\path[draw, fill=black] (7,7) circle[radius=0.06cm];
\path[draw, fill=black] (5,7) circle[radius=0.06cm];
\path (6,7) node[anchor=north, font=\tiny]{$\bf z$};

\path (6,8) node[anchor=north, font=\tiny]{$ \bf z_1$};
\path (7,7) node[anchor=south, font=\tiny]{$ \bf z_2$};
\path (6,6) node[anchor=south, font=\tiny]{$ \bf z_3$};
\path (5,7) node[anchor=north, font=\tiny]{$ \bf z_4$};

\end{tikzpicture}
\caption{
A hedgehog domain $\om^\delta$ with the boundary (thick line) and portions of the sets ${\color{gray}\blacklozenge}^\delta$ and $\lozenge^\delta$. A portion of faces of $\partial_{\operatorname{out}}\om^\delta$. The vertex set $\mathcal{V}^\delta = \mathcal{V}^\delta_\circ \sqcup \mathcal{V}^\delta_\bullet \sqcup \mathcal{V}^\delta_\diamond$  of $\om^\delta$ consists of white vertices ($\mathcal{V}^\delta_\circ$), black vertices ($\mathcal{V}^\delta_\bullet$) and ``diamond'' ones ($\mathcal{V}^\delta_\diamond$). 
The boundary vertex set ${\partial}\mathcal{V}^\delta_\circ$ (white vertices on the thick line) and a portion of the set $\tilde{\partial}\mathcal{V}^\delta_\bullet$.
S-holomorphic functions are defined on $\mathcal{V}^\delta_\diamond$, and their projections on $\protect\ww{}^\delta\sqcup\protect\bb{}^\delta$. 
Given a vertex $z\in \mathcal{V}^\delta_\diamond$, we denote the neighbouring squares by $\protect\uI\in\protect\bb{1}^\delta, \protect\vlbr\in\protect\ww{1}^\delta, \protect\uR\in\protect\bb{0}^\delta$ and $\protect\vl\in\protect\ww{0}^\delta$.  The set $\diom^\delta$ is composed of subsets 
$\partial^{+}_{\mathrm {int}}\om^\delta$, 
$\partial^{-}_{\mathrm {int}}\om^\delta$, 
$\partial^{\sharp}_{\mathrm {int}}\om^\delta$ and
$\partial^{\flat}_{\mathrm {int}}\om^\delta$. 
}\label{ed}
\label{s-h}
\label{rotation}
\end{center}
\end{figure}

Let $\F^\delta~\colon\clb^\delta~\to~\mathbb{C}$ be a function.  
The following definition is similar to~\cite{C+S,KLD,Mer01}.
Let us define discrete operators 
$\partial^\delta\colon \mathbb{C}^{{\bb{}}^\delta} \to\mathbb{C}^{\ww{}^\delta}$ and 
$\bar{\partial}^\delta\colon  \mathbb{C}^{{\bb{}}^\delta} \to\mathbb{C}^{\ww{}^\delta}$ by the formulas: 
\[
[\partial^\delta \F^\delta](v)=\frac12\left(\frac{\F^\delta(v+\delta\lambda)-
 \F^\delta(v-\delta\lambda)}{2\delta\lambda}+\frac{\F^\delta(v+\delta\bar{\lambda})-
 \F^\delta(v-\delta\bar{\lambda})}{2\delta\bar{\lambda}}\right), \]

\[ [\bar{\partial}^\delta \F^\delta](v)=\frac12\left(\frac{\F^\delta(v+\delta\lambda)-
 \F^\delta(v-\delta\lambda)}{2\delta\bar{\lambda}}+\frac{\F^\delta(v+\delta\bar{\lambda})-
 \F^\delta(v-\delta\bar{\lambda})}{2\delta\lambda}\right),
\]
  where $v\in\ww{}^\delta$. Note that, If $[\bar{\partial}^\delta \F^\delta](v)=0$, then the two terms involved in the definition of $[\partial \F^\delta]$ are equal to each other.


\begin{definition}
{\rm A function $\F^\delta\colon\clb^\delta\to\mathbb{C}$ is called} discrete holomorphic {\rm in $\om^\delta$ if 
$[\bar{\partial}^\delta\F^\delta](v) = 0$ for all $v\in\ww{}^\delta$. Also, we always assume that $\F^\delta$ is real on~$\clbb{0}^\delta$ and purely imaginary on~$\clbb{1}^\delta$. }

\end{definition}

 
 Define the discrete Laplacian of $\F^\delta$ by
$$\Delta^\delta\F^\delta(u)=\frac {\F^\delta(u+2\delta\lambda)+\F^\delta(u+2\delta\bar{\lambda})+\F^\delta(u-2\delta\lambda)+\F^\delta(u-2\delta\bar{\lambda})-4\F^\delta(u)}{4\delta^2},$$ where $u\in\bb{}^\delta$. Note that $\Delta^\delta \F^\delta(u)=4[\partial^\delta\bar{\partial}^\delta\F^\delta](u)=4[\bar{\partial}^\delta\partial^\delta\F^\delta](u),$ 
the factorisation was noted in~\cite{C+S, KLD, Mer01}.

A function  $\F^\delta\colon\clb^\delta\to\mathbb{C}$ 
 is called {\it discrete harmonic} in $\om^\delta$ if it satisfies $\Delta^\delta \F^\delta(u)=0$ for all $u\in\bb{}^\delta$. 

It is easy to see that discrete harmonic functions satisfy the maximum principle:
$$\max_{u\in\om^\delta}\F^\delta(u)=\max_{u\in\dom^\delta}\F^\delta(u).$$

The notion of s-holomorphicity, which is a version of discrete holomorphicity, was introduced in~\cite{Stas, Stas07}.

\begin{definition}\label{proj_s_hol}
{\rm A function $\F_{\operatorname{s-hol}}^\delta\colon \mathcal{V}_\diamond^\delta\to\mathbb{C}$ is called} s-holomorphic {\rm on $\mathcal{V}_\diamond^\delta$ if for each pair of vertices $z_1, z_2\in\mathcal{V}_\diamond^\delta$ of the same square $a\in\om^\delta$
\[\mathrm{Proj}_{\tau(a)} [F_{\operatorname{s-hol}}^\delta(z_1)] = \mathrm{Proj}_{\tau(a)} [F_{\operatorname{s-hol}}^\delta(z_2)],\]
where ${\mathrm{Proj}}_{\tau(a)}[z]=\tau(a)\cdot\re\left[z\cdot\overline{\tau(a)}\right]$ and $\tau(a)$ is $1$, $i$, $\lambda$ or $\bar{\lambda}$ if the square $a$ is a square of type $\bb{0}^\delta$, $\bb{1}^\delta$, $\ww{0}^\delta$ or $\ww{1}^\delta$ correspondingly.}
\end{definition}

\begin{remark}
An s-holomorphic function $\F_{\operatorname{s-hol}}^\delta\colon \mathcal{V}_\diamond^\delta\to\mathbb{C}$ can be extended 
to the faces of $\clom^\delta.$
For any $a\in\clom^\delta$
\[
\F_{\operatorname{s-hol}}^\delta(a):= {\mathrm{Proj}}_{\tau(a)}[F_{\operatorname{s-hol}}^\delta(z)],
\]
where $z\in\mathcal{V}_\diamond^\delta$ is a vertex of the square $a$.
\end{remark}

From now onwards, we will think that s-holomorphic functions are defined on the set $\mathcal{V}_\diamond^\delta\sqcup\clom^\delta$ rather than on the set $\mathcal{V}_\diamond^\delta$ only.

\begin{remark}\label{bij:hol_s-hol}
There is a bijection between s-holomorphic functions on $\mathcal{V}_\diamond^\delta\sqcup\clom^\delta$ and holomorphic functions on~$\bar{\bb{}}^\delta$: 

1. Let a function $F^\delta_{\operatorname{s-hol}}$ be s-holomorphic. It is easy to check that $\mathrm{Proj}_{\tau(\cdot)} [F_{\operatorname{s-hol}}^\delta]|_{\bar{\bb{}}^\delta}$ is holomorphic; 

2. Let $\F^\delta\colon \bar{\bb{}}^\delta\to\mathbb{C}$ be a discrete holomorphic function. Let $F^\delta_{\operatorname{s-hol}}$ be a function defined as follows:
\[
\begin{cases} 
\begin{array}{llll}
F^\delta_{\operatorname{s-hol}}(u)=\F^\delta(u)  \quad & \operatorname{if}\, u\in \bb{}^\delta; 
& F^\delta_{\operatorname{s-hol}}(\vl)=\frac{\lambda}{\sqrt{2}}\cdot(\F^\delta(\uR)-i\F^\delta(\uI)) \quad & \operatorname{if}\, \vl\in\ww{0}^\delta; \\
F^\delta_{\operatorname{s-hol}}(z)=\F^\delta(\uR)+\F^\delta(\uI)  \quad & \operatorname{if}\, z\in \mathcal{V}_\diamond^\delta; 
&  F^\delta_{\operatorname{s-hol}}(\vlbr)=\frac{\bar\lambda}{\sqrt{2}}\cdot(\F^\delta(\uR)+i\F^\delta(\uI)) \quad & \operatorname{if}\, \vlbr\in\ww{1}^\delta,\\
\end{array}
\end{cases}
\] 
\old{where $z\in \mathcal{V}_\diamond^\delta$ and $\uI, \vlbr, \uR, \vl$ are adjacent to the vertex $z$ squares of types $\bb{1}^\delta$, $\ww{1}^\delta$, $\bb{0}^\delta$ and $\ww{0}^\delta$ correspondingly (see Fig.~\ref{s-h}).}
where $z\in \mathcal{V}_\diamond^\delta$ and $\uI, \vlbr, \uR, \vl$ are squares of types $\bb{1}^\delta$, $\ww{1}^\delta$, $\bb{0}^\delta$ and $\ww{0}^\delta$ adjacent to the vertex $z$ (see Fig.~\ref{s-h}). Discrete holomorphicity of $\F^\delta$ guarantees that the function $F^\delta_{\operatorname{s-hol}}$ is well defined. 
Note that 
\[
\frac{\lambda}{\sqrt{2}}\cdot(\F^\delta(\uR)-i\F^\delta(\uI))={\mathrm{Proj}}_{\lambda}[\F^\delta(\uR)+\F^\delta(\uI)],
\]
\[
\frac{\bar\lambda}{\sqrt{2}}\cdot(\F^\delta(\uR)+i\F^\delta(\uI))={\mathrm{Proj}}_{\bar\lambda}[\F^\delta(\uR)+\F^\delta(\uI)].
\]
Therefore the function $\F_{\operatorname{s-hol}}^\delta$ is s-holomorphic on $\mathcal{V}_\diamond^\delta\sqcup\clom^\delta$.
\end{remark}

\section{The boundary value problem for s-holomorphic functions}\label{4}
In this section we introduce {\it hedgehog domains}. We study the specific discrete boundary value problem of Riemann type on hedgehog domains.

\subsection{Hedgehog domains}\label{section even}
To define a hedgehog domain let us define a {\it dashed square lattice} (see Fig.~\ref{ed}), a lattice where each square has side $2\sqrt{2}\delta$ and centered at a lattice point of 
\[
\left\{\left(\frac{4\delta n+1}{\sqrt{2}},\frac{4\delta m+1}{\sqrt{2}}\right)\, | \,n, m \in \mathbb{Z}\right\}.
\]

\begin{definition}
{\rm A discrete simply connected domain $\om^\delta$ is called a}  hedgehog domain {\rm if it is composed of a finite number of squares $2\delta\times2\delta$ with vertices in $\mathcal{V}^\delta_\circ$ and each such square has either zero or two consecutive edges on the boundary of $\om^\delta$.
Equivalently, one can take an arbitrary simply-connected union of $2\sqrt{2}\delta\times2\sqrt{2}\delta$ squares of the dashed square lattice and add $2\delta\times 2\delta$ right triangles to each of its boundary edges, see Fig. 3.}
\end{definition}

\begin{remark}
Note that the boundary of a hedgehog domain need not be a simple curve.
\end{remark}


Let us divide the set  $\diom^\delta$ of a hedgehog domain into four sets 
$\partial^{+}_{\mathrm {int}}\om^\delta$, 
$\partial^{-}_{\mathrm {int}}\om^\delta$, 
$\partial^{\sharp}_{\mathrm {int}}\om^\delta$ and
$\partial^{\flat}_{\mathrm {int}}\om^\delta$: note that each $\delta\times\delta$ square $a\in\diom^\delta$ belongs to exactly one square $2\delta\times2\delta$ with vertices in $\mathcal{V}_\circ$ touching the boundary, if north-east and south-east (resp., N-W and S-W; N-W and N-E; S-W and S-E) sides of this $2\delta\times2\delta$ square belong to the boundary of $\om^\delta$, then $a\in\partial^{+}_{\mathrm {int}}\om^\delta$ (resp., $a\in\partial^{-}_{\mathrm {int}}\om^\delta$; 
$a\in\partial^{\sharp}_{\mathrm {int}}\om^\delta$;
$a\in\partial^{\flat}_{\mathrm {int}}\om^\delta$). In the same way as above we define the sets  
$\partial^{+}_{\mathrm {int}}\bb{0,1}^\delta$,
$\partial^{+}_{\mathrm {int}}\ww{0,1}^\delta$,
$\partial^{-}_{\mathrm {int}}\bb{0,1}^\delta$,
$\partial^{-}_{\mathrm {int}}\ww{0,1}^\delta$,
$\partial^{\flat}_{\mathrm {int}}\bb{0,1}^\delta$,
$\partial^{\flat}_{\mathrm {int}}\ww{0,1}^\delta$,
$\partial^{\sharp}_{\mathrm {int}}\bb{0,1}^\delta$ and
$\partial^{\sharp}_{\mathrm {int}}\ww{0,1}^\delta$,
see Fig.~\ref{ed}. 
Note that the sets $\partial\mathcal{V}_\bullet^\delta$,  
$\tilde\partial\mathcal{V}_\circ^\delta$,  
$\partial^{+}_{\mathrm {int}}\bb{1}^\delta$, 
$\partial^{-}_{\mathrm {int}}\bb{0}^\delta$, 
$\partial^{\sharp}_{\mathrm {int}}\ww{0}^\delta$ and
$\partial^{\flat}_{\mathrm {int}}\ww{1}^\delta$ of hedgehog domain $\om^\delta$ are empty.

\subsection{Riemann boundary value problem for s-holomorphic functions}
Fermionic observables in the Ising model are s-holomorphic functions with Riemann-type boundary conditions~\cite{CHI, CI, CS, HS, Stas, Stas07}. In Section~\ref{3} we show that the dimer coupling function on hedgehog domains can be considered as an s-holomorphic function satisfying Riemann boundary conditions.
\begin{definition}\label{RH}
Let $\om^\delta$ be a hedgehog domain. Let $\www\in\mathrm{Int}\ww{0}^\delta$.
We say that $F^\delta_{\operatorname{s-hol}}$ solves the discrete Riemann boundary value problem 
$\bf{\operatorname{\bf RBVP}(\om^\delta,\www)}$ 
if 

$
\begin{cases}
F^\delta_{\operatorname{s-hol}} \text{ is  s-holomorphic  on }
(\clom^\delta\smallsetminus\{\www\})\sqcup\mathcal{V}_{\diamond}^\delta;\\
[ \bar{\partial}F^\delta_{\operatorname{s-hol}}](\www)=\frac{\lambda}{4\delta^2};\\
\im[F^\delta_{\operatorname{s-hol}}(z)\cdot{\sqrt{(n(z))}}]=0, \quad z\in\partial\mathcal{V}_{\diamond}^\delta.
\end{cases}$
\end{definition}

\subsection{The primitive of the square of the s-holomorphic function}\label{h-h=ff}
 This definition was introduced by Smirnov in~\cite{Stas}. 
 The primitive of the square of the s-holomorphic function is a crucial tool for the analysis of fermionic observables in the Ising model via boundary value problems for s-holomorphic functions~\cite{Stas, Stas07, CS}.
 In the present paper we use this approach to study the scaling limit of the dimer coupling function. 
 

\begin{definition}\label{defH}
Let a function $\F_{\operatorname{s-hol}}^\delta$ be s-holomorphic.
Let us define a function $\HH^\delta\colon \mathcal{V}_{\bullet}^\delta\sqcup\mathcal{V}_{\circ}^\delta 
\to \mathbb{R}$ by the equality
\begin{equation} \label{def H}
H^\delta_\bullet(z_2)-H^\delta_\circ(z_1)=(F^\delta_{\operatorname{s-hol}}(a))^2\cdot(z_2 - z_1),
\end{equation}
where $H^\delta_\bullet$ (resp., $H^\delta_\circ$) is a restriction of the function $\HH^\delta$ to the set $\mathcal{V}_{\bullet}^\delta$ (resp., $\mathcal{V}_{\circ}^\delta$), and $z_1\in\mathcal{V}_{\circ}^\delta$, $z_2\in\mathcal{V}_{\bullet}^\delta$ are two vertices of the same square $a$. Note that the function $\HH^\delta$ can be viewed as the primitive of the real part of the square of the s-holomorphic function $\F_{\operatorname{s-hol}}^\delta$, see~\cite[Proposition 3.6]{CS}.
\end{definition}

\begin{remark}\label{int} 1. The primitive ${\HH}^\delta$ is defined up to an additive constant.

2. S-holomorphicity of the function $F_{\operatorname{s-hol}}^\delta$ guarantees that 
\[
(\HH^\delta(z_1) - \HH^\delta(z_2)) + (\HH^\delta(z_4) - \HH^\delta(z_1)) + (\HH^\delta(z_3) - \HH^\delta(z_4)) + (\HH^\delta(z_2) - \HH^\delta(z_3)) = 0,
\]
 since $|F_{\operatorname{s-hol}}^\delta(z)|^2=(F_{\operatorname{s-hol}}^\delta(\protect\uR))^2-(F_{\operatorname{s-hol}}^\delta(\protect\uI))^2=i((F_{\operatorname{s-hol}}^\delta(\protect\vlbr))^2-(F_{\operatorname{s-hol}}^\delta(\protect\vl))^2),$ where vertices $z, z_1, z_2, z_3, z_4$ and squares $\uI, \vlbr, \uR, \vl$ are as shown on~Fig.~\ref{rotation}. Therefore, if $\om^\delta$ is simply connected, then $\HH^\delta$ is well defined. %
\end{remark}

Recall that $\tilde{\partial}\mathcal{V}_\bullet^\delta$ of hedgehog domain $\om^\delta$ is the set of black vertices of squares of the set $\dom^\delta$, see Fig.~\ref{Fbdr}.
Let us define the discrete {\it leap-frog Laplacian} of $\HH_\circ^\delta$ by 
\begin{equation} \label{lfH}
[\Delta^\delta_\circ\HH^\delta_\circ](z)=\frac{1}{4\delta^2}\sum_{w\sim z}(\HH^\delta_\circ(w)-\HH^\delta_\circ(z)), \quad z\in\mathrm {Int}\mathcal{V}_{\circ}^\delta,
\end{equation} 
where the sum is over the four neighbours $w\in\mathcal{V}_{\circ}^\delta$ of $z$. Similarly, one can define the slightly modified discrete leap-frog Laplacian of $\HH_\bullet^\delta$ by
\begin{equation} \label{lfH}
[\Delta^\delta_\bullet\HH^\delta_\bullet](z)=\frac{1}{c_z\delta^2}\sum_{w\sim z}c_{zw}(\HH^\delta_\bullet(w)-\HH^\delta_\bullet(z)), \quad z\in\mathrm {Int}\mathcal{V}_{\bullet}^\delta,
\end{equation} 
where $c_z=\sum_{w\sim z}c_{zw}$, and $c_{zw}$ equals $1$ for inner edges and $2(\sqrt2 -1)$ for the boundary edges, see Fig.~\ref{Fbdr}. For the reason of this ``boundary modification'' of $\Delta^\delta_\bullet$, see ~\cite[Section 3.6]{CS}.

\begin{figure}
\begin{center}

\begin{tikzpicture}[x={(0.25cm,0.25cm)}, y={(0.25cm,-0.25cm)}]

\draw  [draw, fill=gray!20](0,6) --(0,7)-- (1,7)--(1,6)--cycle;
\draw  [draw, fill=gray!20](0,4)--(0,5)-- (1,5)--(1,4)--cycle;

\draw  [draw, fill=gray!20](2,2)  --(2,3)-- (3,3)--(3,2)--cycle;
\draw  [draw, fill=gray!20](2,4)  --(2,5)-- (3,5)--(3,4)--cycle;
\draw  [draw, fill=gray!20](2,6)  --(2,7)-- (3,7)--(3,6)--cycle;
\draw  [draw, fill=gray!20](2,8)  --(2,9)-- (3,9)--(3,8)--cycle;

\draw  [draw, fill=gray!20](4,0) --(4,1)-- (5,1)--(5,0)--cycle;
\draw  [draw, fill=gray!20](4,2) --(4,3)-- (5,3)--(5,2)--cycle;
\draw  [draw, fill=gray!20](4,4) --(4,5)-- (5,5)--(5,4)--cycle;
\draw  [draw, fill=gray!20](4,6) --(4,7)-- (5,7)--(5,6)--cycle;
\draw  [draw, fill=gray!20](4,8) --(4,9)-- (5,9)--(5,8)--cycle;
\draw  [draw, fill=gray!20](4,10) --(4,11)-- (5,11)--(5,10)--cycle;

\draw  [draw, fill=gray!20](6,0) --(6,1)-- (7,1)--(7,0)--cycle;
\draw  [draw, fill=gray!20](6,2) --(6,3)-- (7,3)--(7,2)--cycle;
\draw  [draw, fill=gray!20](6,4) --(6,5)-- (7,5)--(7,4)--cycle;
\draw  [draw, fill=gray!20](6,6) --(6,7)-- (7,7)--(7,6)--cycle;
\draw  [draw, fill=gray!20](6,8) --(6,9)-- (7,9)--(7,8)--cycle;
\draw  [draw, fill=gray!20](6,10) --(6,11)-- (7,11)--(7,10)--cycle;

\draw  [draw, fill=gray!20](8,0) --(8,1)--  (9,1)--(9,0)--cycle;
\draw  [draw, fill=gray!20](8,2) --(8,3)-- (9,3)--(9,2)--cycle;
\draw  [draw, fill=gray!20](8,4) --(8,5)-- (9,5)--(9,4)--cycle;
\draw  [draw, fill=gray!20](8,6) --(8,7)-- (9,7)--(9,6)--cycle;
\draw  [draw, fill=gray!20](8,8) --(8,9)-- (9,9)--(9,8)--cycle;

\draw  [draw, fill=gray!20](10,8) --(10,9)-- (11,9)--(11,8)--cycle;
\draw  [draw, fill=gray!20](10,6) --(10,7)-- (11,7)--(11,6)--cycle;
\draw  [draw, fill=gray!20](10,4) --(10,5)-- (11,5)--(11,4)--cycle;
\draw  [draw, fill=gray!20](10,2) --(10,3)-- (11,3)--(11,2)--cycle;
\draw  [draw, fill=gray!20](10,0) --(10,1)-- (11,1)--(11,0)--cycle;

\draw  [draw, fill=gray!20](12,10) --(12,11)-- (13,11)--(13,10)--cycle;
\draw  [draw, fill=gray!20](12,8) --(12,9)-- (13,9)--(13,8)--cycle;
\draw  [draw, fill=gray!20](12,6) --(12,7)-- (13,7)--(13,6)--cycle;
\draw  [draw, fill=gray!20](12,4) --(12,5)-- (13,5)--(13,4)--cycle;
\draw  [draw, fill=gray!20](12,2) --(12,3)-- (13,3)--(13,2)--cycle;
\draw  [draw, fill=gray!20](12,0) --(12,1)-- (13,1)--(13,0)--cycle;

\draw  [draw, fill=gray!20](14,12) --(14,13)-- (15,13)--(15,12)--cycle;
\draw  [draw, fill=gray!20](14,10) --(14,11)-- (15,11)--(15,10)--cycle;
\draw  [draw, fill=gray!20](14,8) --(14,9)-- (15,9)--(15,8)--cycle;
\draw  [draw, fill=gray!20](14,6) --(14,7)-- (15,7)--(15,6)--cycle;
\draw  [draw, fill=gray!20](14,4) --(14,5)-- (15,5)--(15,4)--cycle;
\draw  [draw, fill=gray!20](14,2) --(14,3)-- (15,3)--(15,2)--cycle;
\draw  [draw, fill=gray!20](14,0) --(14,1)-- (15,1)--(15,0)--cycle;


\draw  [draw, fill=gray!20](16,12) --(16,13)-- (17,13)--(17,12)--cycle;
\draw  [draw, fill=gray!20](16,10) --(16,11)-- (17,11)--(17,10)--cycle;
\draw  [draw, fill=gray!20](16,8) --(16,9)-- (17,9)--(17,8)--cycle;
\draw  [draw, fill=gray!20](16,6) --(16,7)-- (17,7)--(17,6)--cycle;
\draw  [draw, fill=gray!20](16,4) --(16,5)-- (17,5)--(17,4)--cycle;
\draw  [draw, fill=gray!20](16,2) --(16,3)-- (17,3)--(17,2)--cycle;

\draw  [draw, fill=gray!20](18,6) --(18,7)-- (19,7)--(19,6)--cycle;
\draw  [draw, fill=gray!20](18,4) --(18,5)-- (19,5)--(19,4)--cycle;

\draw  [draw, fill=gray!20](1,7) --(1,8)-- (2,8)--(2,7)--cycle;
\draw  [draw, fill=gray!20](1,5) --(1,6)-- (2,6)--(2,5)--cycle;

\draw  [draw, fill=gray!20](3,3) --(3,4)-- (4,4)--(4,3)--cycle;
\draw  [draw, fill=gray!20](3,5) --(3,6)-- (4,6)--(4,5)--cycle;
\draw  [draw, fill=gray!20](3,7) --(3,8)-- (4,8)--(4,7)--cycle;
\draw  [draw, fill=gray!20](3,9) --(3,10)-- (4,10)--(4,9)--cycle;

\draw  [draw, fill=gray!20](5,1) --(5,2)-- (6,2)--(6,1)--cycle;
\draw  [draw, fill=gray!20](5,3) --(5,4)-- (6,4)--(6,3)--cycle;
\draw  [draw, fill=gray!20](5,5) --(5,6)-- (6,6)--(6,5)--cycle;
\draw  [draw, fill=gray!20](5,7) --(5,8)-- (6,8)--(6,7)--cycle;
\draw  [draw, fill=gray!20](5,9) --(5,10)-- (6,10)--(6,9)--cycle;
\draw  [draw, fill=gray!20](5,11) --(5,12)-- (6,12)--(6,11)--cycle;

\draw  [draw, fill=gray!20](7,1) --(7,2)-- (8,2)--(8,1)--cycle;
\draw  [draw, fill=gray!20](7,3) --(7,4)-- (8,4)--(8,3)--cycle;
\draw  [draw, fill=gray!20](7,5) --(7,6)-- (8,6)--(8,5)--cycle;
\draw  [draw, fill=gray!20](7,7) --(7,8)-- (8,8)--(8,7)--cycle;
\draw  [draw, fill=gray!20](7,9) --(7,10)-- (8,10)--(8,9)--cycle;
\draw  [draw, fill=gray!20](7,11) --(7,12)-- (8,12)--(8,11)--cycle;

\draw  [draw, fill=gray!20](9,1) --(9,2)-- (10,2)--(10,1)--cycle;
\draw  [draw, fill=gray!20](9,3) --(9,4)-- (10,4)--(10,3)--cycle;
\draw  [draw, fill=gray!20](9,5) --(9,6)-- (10,6)--(10,5)--cycle;
\draw  [draw, fill=gray!20](9,7) --(9,8)-- (10,8)--(10,7)--cycle;
\draw  [draw, fill=gray!20](9,9) --(9,10)-- (10,10)--(10,9)--cycle;

\draw  [draw, fill=gray!20](11,1) --(11,2)-- (12,2)--(12,1)--cycle;
\draw  [draw, fill=gray!20](11,3) --(11,4)-- (12,4)--(12,3)--cycle;
\draw  [draw, fill=gray!20](11,5) --(11,6)-- (12,6)--(12,5)--cycle;
\draw  [draw, fill=gray!20](11,7) --(11,8)-- (12,8)--(12,7)--cycle;
\draw  [draw, fill=gray!20](11,9) --(11,10)-- (12,10)--(12,9)--cycle;

\draw  [draw, fill=gray!20](13,1) --(13,2)-- (14,2)--(14,1)--cycle;
\draw  [draw, fill=gray!20](13,3) --(13,4)-- (14,4)--(14,3)--cycle;
\draw  [draw, fill=gray!20](13,5) --(13,6)-- (14,6)--(14,5)--cycle;
\draw  [draw, fill=gray!20](13,7) --(13,8)-- (14,8)--(14,7)--cycle;
\draw  [draw, fill=gray!20](13,9) --(13,10)-- (14,10)--(14,9)--cycle;
\draw  [draw, fill=gray!20](13,11) --(13,12)-- (14,12)--(14,11)--cycle;

\draw  [draw, fill=gray!20](15,1) --(15,2)-- (16,2)--(16,1)--cycle;
\draw  [draw, fill=gray!20](15,3) --(15,4)-- (16,4)--(16,3)--cycle;
\draw  [draw, fill=gray!20](15,5) --(15,6)-- (16,6)--(16,5)--cycle;
\draw  [draw, fill=gray!20](15,7) --(15,8)-- (16,8)--(16,7)--cycle;
\draw  [draw, fill=gray!20](15,9) --(15,10)-- (16,10)--(16,9)--cycle;
\draw  [draw, fill=gray!20](15,11) --(15,12)-- (16,12)--(16,11)--cycle;
\draw  [draw, fill=gray!20](15,13) --(15,14)-- (16,14)--(16,13)--cycle;

\draw  [draw, fill=gray!20](17,3) --(17,4)-- (18,4)--(18,3)--cycle;
\draw  [draw, fill=gray!20](17,5) --(17,6)-- (18,6)--(18,5)--cycle;
\draw  [draw, fill=gray!20](17,7) --(17,8)-- (18,8)--(18,7)--cycle;
\draw  [draw, fill=gray!20](17,9) --(17,10)-- (18,10)--(18,9)--cycle;
\draw  [draw, fill=gray!20](17,11) --(17,12)-- (18,12)--(18,11)--cycle;
\draw  [draw, fill=gray!20](17,13) --(17,14)-- (18,14)--(18,13)--cycle;

\draw  [draw, fill=gray!20](19,5) --(19,6)-- (20,6)--(20,5)--cycle;
\draw  [draw, fill=gray!20](19,7) --(19,8)-- (20,8)--(20,7)--cycle;

\draw[draw,dashed, line width=0.7pt] (0,6)--(6,0)--(8,2)--(10,0)--(12,2)--(14,0)--(20,6)--(16,10)--(18,12)--(16,14)--(10,8)--(6,12)--cycle;

\draw[draw, line width=2pt] (0,4) -- (0,8) -- (2,8) -- (2,10) -- (4,10) -- (4,12) -- (8,12) -- (8,10) -- (10,10) --  (12,10) -- (12,12) -- (14,12) -- (14,14) -- (18,14) -- (18,8) -- (20,8) -- (20,4) -- (18,4) -- (18,2) -- (16,2) -- (16,0) -- (4,0)-- (4,2) -- (2,2) -- (2,4)
-- cycle;
\draw[draw, line width=2pt] (18,10)--(16,10);
\draw[draw, line width=2pt] (10,10)--(10,8);
\draw[draw, line width=2pt] (8,0)--(8,2);
\draw[draw, line width=2pt] (12,0)--(12,2);

\path[draw, fill=black] (1,7) circle[radius=0.05cm];
\path[draw, fill=black] (3,9) circle[radius=0.05cm];
\path[draw, fill=black] (5,11) circle[radius=0.05cm];
\path[draw, fill=black] (7,11) circle[radius=0.05cm];
\path[draw, fill=black] (9,9) circle[radius=0.05cm];
\path[draw, fill=black] (11,9) circle[radius=0.05cm];
\path[draw, fill=black] (13,11) circle[radius=0.05cm];
\path[draw, fill=black] (15,13) circle[radius=0.05cm];
\path[draw, fill=black] (17,13) circle[radius=0.05cm];
\path[draw, fill=black] (17,11) circle[radius=0.05cm];
\path[draw, fill=black] (17,9) circle[radius=0.05cm];
\path[draw, fill=black] (19,7) circle[radius=0.05cm];
\path[draw, fill=black] (19,5) circle[radius=0.05cm];
\path[draw, fill=black] (17,3) circle[radius=0.05cm];
\path[draw, fill=black] (15,1) circle[radius=0.05cm];
\path[draw, fill=black] (13,1) circle[radius=0.05cm];
\path[draw, fill=black] (11,1) circle[radius=0.05cm];
\path[draw, fill=black] (9,1) circle[radius=0.05cm];
\path[draw, fill=black] (7,1) circle[radius=0.05cm];
\path[draw, fill=black] (5,1) circle[radius=0.05cm];
\path[draw, fill=black] (3,3) circle[radius=0.05cm];
\path[draw, fill=black] (1,5) circle[radius=0.05cm];

\path[draw, fill=white] (0,6) circle[radius=0.06cm];
\path[draw, fill=white] (0,8) circle[radius=0.06cm];
\path[draw, fill=white] (2,8) circle[radius=0.06cm];
\path[draw, fill=white] (2,10) circle[radius=0.06cm];
\path[draw, fill=white] (4,10) circle[radius=0.06cm];
\path[draw, fill=white] (4,12) circle[radius=0.06cm];
\path[draw, fill=white] (6,12) circle[radius=0.06cm];
\path[draw, fill=white] (8,12) circle[radius=0.06cm];
\path[draw, fill=white] (8,10) circle[radius=0.06cm];
\path[draw, fill=white] (10,8) circle[radius=0.06cm];
\path[draw, fill=white] (10,10) circle[radius=0.06cm];
\path[draw, fill=white] (12,10) circle[radius=0.06cm];
\path[draw, fill=white] (12,12) circle[radius=0.06cm];
\path[draw, fill=white] (14,12) circle[radius=0.06cm];
\path[draw, fill=white] (14,14) circle[radius=0.06cm];
\path[draw, fill=white] (16,14) circle[radius=0.06cm];
\path[draw, fill=white] (18,14) circle[radius=0.06cm];
\path[draw, fill=white] (18,12) circle[radius=0.06cm];
\path[draw, fill=white] (16,10) circle[radius=0.06cm];
\path[draw, fill=white] (18,10) circle[radius=0.06cm];
\path[draw, fill=white] (18,8) circle[radius=0.06cm];
\path[draw, fill=white] (20,8) circle[radius=0.06cm];
\path[draw, fill=white] (20,6) circle[radius=0.06cm];
\path[draw, fill=white] (20,4) circle[radius=0.06cm];
\path[draw, fill=white] (18,4) circle[radius=0.06cm];
\path[draw, fill=white] (18,2) circle[radius=0.06cm];
\path[draw, fill=white] (16,2) circle[radius=0.06cm];
\path[draw, fill=white] (16,0) circle[radius=0.06cm];
\path[draw, fill=white] (14,0) circle[radius=0.06cm];
\path[draw, fill=white] (12,2) circle[radius=0.06cm];
\path[draw, fill=white] (12,0) circle[radius=0.06cm];
\path[draw, fill=white] (10,0) circle[radius=0.06cm];
\path[draw, fill=white] (8,2) circle[radius=0.06cm];
\path[draw, fill=white] (8,0) circle[radius=0.06cm];
\path[draw, fill=white] (6,0) circle[radius=0.06cm];
\path[draw, fill=white] (4,0) circle[radius=0.06cm];
\path[draw, fill=white] (4,2) circle[radius=0.06cm];
\path[draw, fill=white] (2,2) circle[radius=0.06cm];
\path[draw, fill=white] (2,4) circle[radius=0.06cm];
\path[draw, fill=white] (0,4) circle[radius=0.06cm];

\draw[draw, line width=0.7pt] (6,8)--(7,7)--(6,6)--(5,7)--(6,8);
\path[draw, fill=white] (5.9,6.9)--(6.1,6.9)--(6.1,7.1)--(5.9,7.1)--cycle;
\path[draw, fill=black] (7,7) circle[radius=0.05cm];
\path[draw, fill=black] (5,7) circle[radius=0.05cm];
\path[draw, fill=white] (6,6) circle[radius=0.06cm];
\path[draw, fill=white] (6,8) circle[radius=0.06cm];

\path (6.1,6.9) node[anchor=north, font=\tiny]{$\bf z$};

\path (12.5,5.5) node[font=\tiny]{$\www$};
\draw[draw,dotted, line width=0.9pt] (11,5)--(12,4)--(14,6)--(13,7)--(11,5);

\path[draw, fill=black] (13,5) circle[radius=0.05cm];
\path[draw, fill=black] (11,5) circle[radius=0.05cm];
\path[draw, fill=black] (13,7) circle[radius=0.05cm];
\path[draw, fill=white] (12,6) circle[radius=0.06cm];
\path[draw, fill=white] (14,6) circle[radius=0.06cm];
\path[draw, fill=white] (12,4) circle[radius=0.06cm];

\path (12,4) node[font=\tiny, anchor=south]{$z'_\circ$};
\path (14,6) node[font=\tiny, anchor=south]{$z''_\circ$};
\path (12,6) node[font=\tiny, anchor=north]{$z_\circ$};

\path (13,5) node[font=\tiny, anchor=south]{$z_\bullet$};
\path  (13,7) node[font=\tiny, anchor=north]{$z''_\bullet$};
\path (11,5) node[font=\tiny, anchor=north]{$z'_\bullet$};

\end{tikzpicture}

\end{center}
\caption{Boundary contour around $\mathrm{Int}\mathcal{V}_\diamond^\delta$ (dashed): s-holomorphicity together with Riemann boundary conditions imply that at two consecutive white vertices on the boundary contour the values of $\HH^\delta$ coincide; contour around $z_{\diamond} \in\mathrm{Int}\mathcal{V}_\diamond^\delta$ (full); contour around $\www$ (dotted).
} \label{H_monodr}
\end{figure}
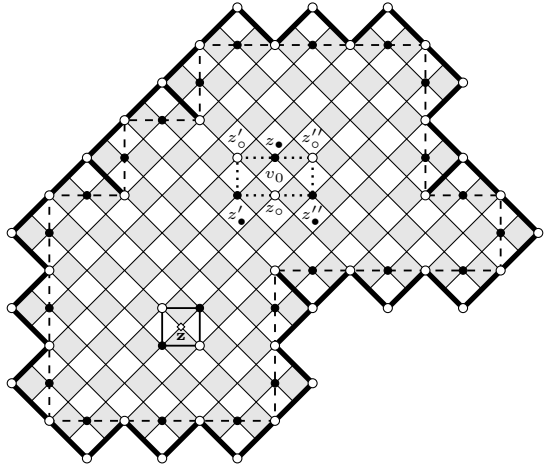

Let  an s-holomorphic function $\F_{\operatorname{s-hol}}^\delta$ solve a discrete boundary value problem ${\operatorname{RBVP}(\om^\delta,\www)}$. 
Let $H^\delta$ be the primitive of the square of $\F_{\operatorname{s-hol}}^\delta$ defined by~(\ref{def H}). 
 Note that $\F_{\operatorname{s-hol}}^\delta$ is not defined at $\www$. Hence, we need to check that the definition of $\HH^\delta$ is consistent around $\www$.
Given a discrete path $\gamma=(z_1,...,z_n)$ with vertices in $\mathcal{V}_{\bullet}^\delta\cup\mathcal{V}_{\circ}^\delta$ and a function $F^\delta_{\operatorname{s-hol}}$ on faces, we define  
\[\int_\gamma (F^\delta_{\operatorname{s-hol}})^2\,dz=\sum_{i=1}^{n-1} (F^\delta_{\operatorname{s-hol}}(a_{i,i+1}))^2(z_{i+1}-z_{i}),\] 
where $z_i, z_{i+1} \in \mathcal{V}_{\bullet}^\delta\cup\mathcal{V}_{\circ}^\delta$ are two vertices of the same square $a_{i,i+1}.$
We similarly define contour integrals around sets of vertices of type $\mathcal{V}_{\diamond}^\delta$. Remark~\ref{int} implies that for each $z_{\diamond} \in\mathrm{Int}\mathcal{V}_\diamond^\delta$ non-adjacent to the face $\www$  one has
$\oint_{z_\diamond} (F^\delta_{\operatorname{s-hol}})^2\,dz=0$.  Moreover, since $F^\delta_{\operatorname{s-hol}}$ satisfies Riemann boundary conditions, the contour integral along the boundary contour is zero, see Fig.~\ref{H_monodr}. Note that
\[\oint_{\www} (F^\delta_{\operatorname{s-hol}})^2\,dz=\oint_{\mathrm{Int}\mathcal{V}_\diamond^\delta} (F^\delta_{\operatorname{s-hol}})^2\,dz,\] where the contour around $\www$ is as shown on Fig.~\ref{H_monodr}. 
Therefore, $\HH^\delta$ is well defined. 
Then due to~\cite{CS} we have the following proposition.

\begin{prop}\label{proportiesH}  Let $z_\circ\in\mathcal{V}_{\circ}^\delta$ and $z_\bullet\in\mathcal{V}_{\bullet}^\delta$ be vertices of the square $\www$. The function $H^\delta$ satisfies the following properties:

\begin{enumerate}
\item[$\rhd$] if $z$ and $z'$ are two vertices of the same square, then $H^\delta_\circ(z)\geq H^\delta_\bullet(z')$;
\item[$\rhd$] $\HH^\delta$ satisfies Dirichlet boundary conditions: $\HH^\delta_\circ(z)=0$ for any $z\in\partial\mathcal{V}_\circ^\delta$, and $\HH^\delta_\bullet(z')=0$ for any $z'\in\tilde{\partial}\mathcal{V}_\bullet^\delta$;
\item[$\rhd$] $\HH^\delta_\bullet$ has a "nonpositive inner normal derivative", i.e. $\HH^\delta_\bullet(w)\leq0$ for any vertex $w\in\mathcal{V}_\bullet^\delta$ adjacent to a boundary vertex;
\item[$\rhd$] $\HH^\delta_\circ$ is leap-frog subharmonic on $\mathcal{V}_{\circ}^\delta\smallsetminus\{z_\circ\}$, while $\HH^\delta_\bullet$ is leap-frog superharmonic on $\mathcal{V}_{\bullet}^\delta\smallsetminus\{z_\bullet\}$.
\end{enumerate}
\end{prop}

\begin{proof}
All the statements follow directly from~\cite[Section 3.3]{CS}.
\end{proof}

\begin{remark}
The function $\HH^\delta$ satisfies the maximum principle: if $\widetilde{\mathcal{V}}^\delta\subset\mathcal{V}^\delta$ does not contain $z_\circ$ (respectively  $z_\bullet$), then
\[
\max_{z\in\widetilde{\V{}}^\delta}\HH^\delta(z)=\max_{z\in\partial\widetilde{\V{}}^\delta}\HH_\circ^\delta(z) \quad (respectively \,\min_{z\in\widetilde{\V{}}^\delta}\HH^\delta(z)=\min_{z\in\partial\widetilde{\V{}}^\delta}\HH_\bullet^\delta(z) ).
\]
\end{remark}

\begin{prop}
A discrete Riemann boundary value problem ${\operatorname{RBVP}(\om^\delta,\www)}$ has a unique solution.
\end{prop}

\begin{proof}
The existence of such a function will be shown in Section~\ref{coupl}. Let us prove that the solution is unique. Let $F_1^\delta$ and $F_2^\delta$ be two different solutions. Then $F_1^\delta-F_2^\delta$ is s-holomorphic on $\clom^\delta\sqcup\mathcal{V}_{\diamond}^\delta$ and $\im[(F_1^\delta-F_2^\delta)(z)\cdot\sqrt{(n(z))}]=0, \quad z\in\partial\mathcal{V}_{\diamond}^\delta$. Hence one can define the primitive of the square of the difference. The maximum principle for the primitive tells us that such a primitive is identically zero. Therefore $F_1^\delta-F_2^\delta$ is identically zero.
\end{proof}

\subsection{The continuous analogue of the functions $\F_{\operatorname{s-hol}}^\delta$ and $\HH^\delta$}\label{cont}
In this section we describe the continuous analogue of the functions $\F_{\operatorname{s-hol}}^\delta$~and~$\HH^\delta$.
Also, we give a characterisation of the holomorphic solution of the Riemann-type boundary value problem in terms of the primitive of its square.

\begin{prop}\label{f} Let $\om$ be a bounded simply connected domain with smooth boundary, and 
 $v$ be a point in the interior of $\om$. Then for any $\lambda\in\mathbb{C}$ 
there exists a unique holomorphic function $f_{{\om}}^v$ such that:
\begin{enumerate}
\item[(f1)] $f_\om^v(z)=\frac{1}{2\pi}\cdot\frac{\lambda}{z-v}+O(1)$ in a vicinity of the point $v$;
\item[(f2)] $\im[f_\om^v(z){\sqrt{(n(z))}}]=0, \quad z\in\dom$.
\end{enumerate}
\end{prop}

\begin{proof} Let $\phi$ be a conformal mapping of the domain $\om$ onto the unit disk $\mathbb{D}$ such that $v$ is mapped onto $0$ and $\phi'(v)>0$. Note that if $f_{\mathbb{D}}^0$ is a solution in the unit disk with singularity at zero, then 
\begin{equation} \label{notSmooth}
f_\om^v(z):= f_{\mathbb{D}}^0(\phi(z))\cdot(\phi'(z))^\frac12\cdot(\phi'(v))^\frac12
\end{equation}
solves our boundary value problem. It is easy to check that $f_{\mathbb{D}}^0(z)=\frac{1}{2\pi}(\frac{\lambda}{z}+\bar{\lambda}).$

Let $f_1$ and $f_2$ be two different solutions. 
Let $z(t)$ be the natural parametrization of $\dom$. Note that the difference $f_1-f_2$ is holomorphic on $\om$, then
\[
0=\int_{\dom} (f_1-f_2)^2(z)\,dz=\int_{\dom} (f_1-f_2)^2(z(t))\cdot in(z(t))\,dt.
\] The boundary conditions imply that $(f_1-f_2)^2(z(t))\cdot n(z(t))\geq 0$. So, $f_1-f_2=0$ on the boundary of $\om$.
Thus, $f_1= f_2$ in $\om$. 
\end{proof}

The previous proposition also holds if $v$ is a boundary point of $\om$.

\begin{lemma}\label{fboundary} Let $\om$ be a bounded simply connected domain with smooth boundary, and $v\in\dom$. 
Then there exists a unique holomorphic function $f_{{\om}}^v$ such that:
\begin{enumerate}
\item[(F1)] $f_\om^v(z)=\frac{1}{2\pi}\cdot\frac{i\sqrt{n(v)}}{z-v}+O(1)$ in a vicinity of the point $v$;
\item[(F2)] $\im[f_\om^v(z){\sqrt{(n(z))}}]=0, \quad z\in\dom\smallsetminus\{v\}$.
\end{enumerate}
\end{lemma}

\begin{proof}
The uniqueness of the solution can be proven using the same arguments as in Proposition~\ref{f}. To construct $f_\om^v$, consider a holomorphic map $\phi$ of $\om$ onto $\mathbb{D}$ such that $\phi(v)=n(v)$ and hence $\phi'(v)>0$. 
As in Proposition~\ref{f}, we can define 
\begin{equation} \label{notSmooth_b}
f_\om^v(z):= f_{\mathbb{D}}^{\phi(v)}(\phi(z))\cdot(\phi'(z))^\frac12\cdot(\phi'(v))^\frac12,
\end{equation}
where the function 
$f_{\mathbb{D}}^w(z)=\frac{i\sqrt{w}}{2\pi\cdot(z-w)}$ solves a similar boundary value problem in $\mathbb{D}$.
\end{proof}

Due to~\cite[Section 3.3.2]{CHI} one can give a characterisation of the holomorphic solution of the boundary value problem (f1)--(f2) in terms of the primitive of its square. This characterisation will be used 
in the proof of Theorem~\ref{convF} dedicated to one of the main convergence results.

\begin{prop}
Let $\om$ be a simply connected domain, and $v$ be a point in the interior of $\om$. 
Let a holomorphic function $f$ solve the boundary value problem described in Proposition~\ref{f} (or $f$ is defined by~(\ref{notSmooth}), if $\om$ is not smooth). Define two harmonic functions
\[
h:= \int Re[f^2(z)\,dz] \quad and \quad h_\star:= \int Re\left[ \left( f(z)-\frac{1}{2\pi}\cdot\frac{\lambda}{z-v} \right) ^2\,dz\right].
\]
Then the following holds:
\begin{enumerate}
\item[(h1)] $h$ satisfies Dirichlet boundary conditions, since $h$ is defined up to an additive constant, we can assume that $h\equiv 0$ on $\dom$;
\item[(h2)] $\partial_n h\geq 0$ (outer normal derivative is nonnegative);
\item[(h3)] $h_\star$ is bounded in a vicinity of $v$.
\end{enumerate}
Moreover, if $h$ and $h_\star$ satisfy all these conditions, then $f$ coincides with the solution $f_\om^v$ defined in Proposition~\ref{f}.
\end{prop}

\begin{proof}
The property (f2) is equivalent to (h1) and (h2). Property (f1) is equivalent to (h3).
\end{proof}

\section{Coupling function on hedgehog domain}\label{3}
In this section we show that a slightly modified s-holomorphic version of the coupling function satisfies Riemann-type boundary conditions on hedgehog domains. We then prove the convergence of the coupling function.

\subsection{Coupling function as s-holomorphic function}\label{coupl}
We can think of the inverse Kasteleyn matrix $\frac{1}{\delta}C_{\om^\delta}(u,v)$ as a function of two variables $u \in \bb{}^\delta$ and $v \in \ww{}^\delta$. If $v\in\ww{0}^\delta$, then $\frac{1}{\delta}C_{\om^\delta}(u,v)$ is a discrete holomorphic function of $u$, with a simple pole at $v$: 
\[4\delta\bar{\lambda}\bar{\partial}[C_{\om^\delta}(\,\cdot\,,v)](v) =C_{\om^\delta} (v+\delta\lambda,v)-C_{\om^\delta} (v-\delta\lambda,v)+
 iC_{\om^\delta} (v-\delta\bar{\lambda},v)-
 iC_{\om^\delta} (v+\delta\bar{\lambda},v)=1,\]
since the product of the Kasteleyn matrix and the inverse Kasteleyn matrix is equal to the identity matrix. For more details see~\cite{Kdom, rus}. Note that the coupling function as a function of $u$ can be extended to be zero on $\db^\delta$, so that the above equation makes sense.

Let $\om^\delta$ be a hedgehog domain. Fix a white square $\www\in\mathrm{Int}\ww{0}^\delta$. 
Let us define a function $\check{\F}^\delta\colon\clb{}^\delta~\to~\mathbb{C}$ by $\check{\F}^\delta(u):=\frac{1}{\delta}C_{\om^\delta}(u,\www)$. 
Note that $\check{\F}^\delta$ is a discrete holomorphic everywhere in~$\ww{}^\delta$ except at the face $\www$ where one has $[\bar{\partial}^\delta\check{\F}^\delta](\www) = \frac{\lambda}{\delta^2}.$ Therefore one can define an s-holomorphic function $\check{\F}^\delta_{\operatorname{s-hol}}$ on the set
$
(\clom^\delta\smallsetminus\{\www\})\sqcup\mathcal{V}_{\diamond}^\delta$ as described in Remark~\ref{bij:hol_s-hol}. 

Let us divide the set $\partial\mathcal{V}^\delta_\circ$ into two sets $\partial\check{\mathcal{V}}^\delta_\circ$ and $\partial\hat{\mathcal{V}}^\delta_\circ$, where $\partial\check{\mathcal{V}}^\delta_\circ$ are vertices of the dashed lattice and $\partial\mathcal{V}^\delta_\circ=\partial\check{\mathcal{V}}^\delta_\circ\sqcup\partial\hat{\mathcal{V}}^\delta_\circ.$

\begin{prop}\label{half_dirichlet} Let $\check{H}^\delta$ be the primitive of the square of $\check{\F}_{\operatorname{s-hol}}^\delta$ defined by~(\ref{def H}).
Then $\check{\HH}^\delta$ satisfies Dirichlet boundary conditions on the set $\partial\check{\mathcal{V}}_\circ^\delta$: $\check{\HH}^\delta(z)=0$ for any $z\in\partial\check{\mathcal{V}}_\circ^\delta$.
\end{prop}

\begin{proof}
Let $\uI, \uR, \vl, \vlbr, z, z_1, z_2, \check{z}, \check{z}'$ be as shown on Fig.~\ref{Fbdr} and $\vl\neq\www$.
Note that on the upper boundary 
\[-C_{\om^\delta}(\uI,\www)+iC_{\om^\delta}(\uR,\www)=0,\]  since the product of the Kasteleyn matrix and the inverse Kasteleyn matrix is equal to the identity matrix. Therefore $i\check{\F}_{\operatorname{s-hol}}^\delta(\uI)=-\check{\F}_{\operatorname{s-hol}}^\delta(\uR)$, then by~(\ref{def H}) we obtain that 
\[\check{\HH}^\delta(\check{z})-\check{\HH}^\delta(\check{z}')=
(\check{\HH}^\delta(\check{z})-\check{\HH}^\delta({z}))+(\check{\HH}^\delta({z})-\check{\HH}^\delta(\check{z}'))=-(\check{\F}_{\operatorname{s-hol}}^\delta(\uI))^2-(\check{\F}_{\operatorname{s-hol}}^\delta(\uR))^2=0,
\]  on the upper boundary.  
Similarly, one can check that $i\check{\F}_{\operatorname{s-hol}}^\delta(\uI)=\check{\F}_{\operatorname{s-hol}}^\delta(\uR)$ and $\check{\HH}^\delta(\check{z})-\check{\HH}^\delta(\check{z}')=0$ on the lower boundary. 

Recall that $\check{\HH}^\delta$ is defined up to an additive constant, which can be chosen so that $\check{\HH}^\delta(z)=0$ for any $z\in\partial\check{\mathcal{V}}_\circ^\delta$, therefore to finish the proof it is enough to check that $\check{\HH}^\delta(\check{z})=\check{\HH}^\delta(\check{z}')$ on the right and left boundaries.

On the right boundary $\check{\F}^\delta_{\operatorname{s-hol}}(z_1)=\check{\F}_{\operatorname{s-hol}}^\delta(z_2)=\check{\F}_{\operatorname{s-hol}}^\delta(\uR)$. Therefore, on the right boundary,
\[
i\check{\F}_{\operatorname{s-hol}}^\delta(\vlbr)=  i{\mathrm{Proj}}_{\bar\lambda}[F_{\operatorname{s-hol}}^\delta(z_1)]= 
{\mathrm{Proj}}_{\lambda}[F_{\operatorname{s-hol}}^\delta(z_2)]=\check{\F}_{\operatorname{s-hol}}^\delta(\vl).
\] 
Analogously, one can check that $i\check{\F}_{\operatorname{s-hol}}^\delta(\vlbr)=-\check{\F}_{\operatorname{s-hol}}^\delta(\vl)$ on the left boundary. Hence $\check{\HH}^\delta(\,\cdot\,)$ is a constant on $\partial\check{\mathcal{V}}_\circ^\delta$.
\old{Note that $-C_{\om^\delta}(\uI,\www)+iC_{\om^\delta}(\uR,\www)=0$ on the upper boundary, and $\check{\F}^\delta(z_1)=\check{\F}^\delta(z_2)=\check{\F}^\delta(\uI)$ on the left boundary, where $\uI, \uR, z_1, z_2$ are as shown on Fig.~\ref{Fbdr}. Therefore $i\check{\F}_{\operatorname{s-hol}}^\delta(\uI)=\pm\check{\F}_{\operatorname{s-hol}}^\delta(\uR)$ on the upper or lower boundary, and $i\check{\F}_{\operatorname{s-hol}}^\delta(\vlbr)=\pm\check{\F}_{\operatorname{s-hol}}^\delta(\vl)$ on the left or right boundary, see Fig.~\ref{Fbdr}. Therefore $\check{\HH}^\delta(\check{z})=\check{\HH}^\delta(\check{z}')$, so $\check{\HH}^\delta(z)$ is constant on $\partial\check{\mathcal{V}}_\circ^\delta$.
Recall that $\check{\HH}^\delta$ is defined up to an additive constant, which can be chosen so that $\check{\HH}^\delta(z)=0$ for any $z\in\partial\check{\mathcal{V}}_\circ^\delta$.}
\end{proof}

The above proposition does not hold for the whole $\partial\mathcal{V}^\delta_\circ$. Let $ \uR, \vlbr, z, \hat{z}, \check{z}'$ be as shown on Fig.~\ref{Fbdr}. Then on the upper boundary one has
\[
\check{\HH}^\delta(\hat{z})-\check{\HH}^\delta(\check{z}')=
(\check{\HH}^\delta(\hat{z})-\check{\HH}^\delta({z}))+(\check{\HH}^\delta({z})-\check{\HH}^\delta(\check{z}'))=i(\check{\F}_{\operatorname{s-hol}}^\delta(\vlbr))^2-(\check{\F}_{\operatorname{s-hol}}^\delta(\uR))^2.
\] 
Note that in this case $\check{\F}_{\operatorname{s-hol}}^\delta(\vlbr)=\frac{\lambda}{\sqrt{2}}\check{\F}_{\operatorname{s-hol}}^\delta(\uR)$. Therefore, $\check{\HH}^\delta(\hat{z})-\check{\HH}^\delta(\check{z}')\neq 0.$

One can modify the s-holomorphic version $\check{\F}_{\operatorname{s-hol}}^\delta$ of the normalised coupling function on hedgehog domain, in such a way that the primitive of its square vanishes \emph{everywhere} on $\partial{\mathcal{V}}^\delta_\circ$. In other words, one can define an s-holomorphic function $F^\delta_{\operatorname{s-hol}}(\cdot)$, which satisfies Riemann-type boundary conditions everywhere on $\partial\mathcal{V}_{\diamond}^\delta$ and coincides with  $\frac{1}{\delta}\Cmd(\cdot,\www)$ on the set $\intbb{}^\delta$. Such a modification is possible since the values $\check{\F}_{\operatorname{s-hol}}^\delta(z_1), \check{\F}_{\operatorname{s-hol}}^\delta(z_2)$ near a given vertex $z_\circ\in\partial\hat{\mathcal{V}}^\delta_\circ$ are subject to only three real equations from Definition~\ref{proj_s_hol}, which leave one degree of freedom to adjust the value $\HH^\delta(z_\circ).$

Let $F^\delta_{\operatorname{s-hol}}$ be an s-holomorphic function defined on $(\clom^\delta\smallsetminus\{\www\})\sqcup\mathcal{V}^\delta$ coinciding with $\check{\F}^\delta_{\operatorname{s-hol}}$ on the set  
$(\om^\delta\smallsetminus(\{\www\}\cup\partial^{+}_{\mathrm {int}}\bb{0}^\delta\cup\partial^{-}_{\mathrm {int}}\bb{1}^\delta\cup\partial^{\flat}_{\mathrm {int}}\ww{0}^\delta\cup
\partial^{\sharp}_{\mathrm {int}}\ww{1}^\delta
))\sqcup(\mathcal{V}_{\diamond}^\delta\smallsetminus\partial\mathcal{V}_\diamond^\delta)$, i.e.:
\[
\begin{cases} 
\begin{array}{llll}
F^\delta_{\operatorname{s-hol}}(u)=\check{\F}^\delta_{\operatorname{s-hol}}(u)  \quad & \operatorname{if}\, u\in \bb{}^\delta\smallsetminus
(\partial^{+}_{\mathrm {int}}\bb{0}^\delta\cup\partial^{-}_{\mathrm {int}}\bb{1}^\delta); \\
F^\delta_{\operatorname{s-hol}}(v)=\check{\F}^\delta_{\operatorname{s-hol}}(v)  \quad & \operatorname{if}\, v\in \ww{}^\delta\smallsetminus(\{\www\}\cup\partial^{\flat}_{\mathrm {int}}\ww{0}^\delta\cup
\partial^{\sharp}_{\mathrm {int}}\ww{1}^\delta); \\
F^\delta_{\operatorname{s-hol}}(z)=\check{\F}^\delta_{\operatorname{s-hol}}(z)  \quad & \operatorname{if}\, z\in \mathcal{V}_\diamond^\delta\smallsetminus\partial\mathcal{V}_\diamond^\delta. \\
\end{array}
\end{cases}
\]

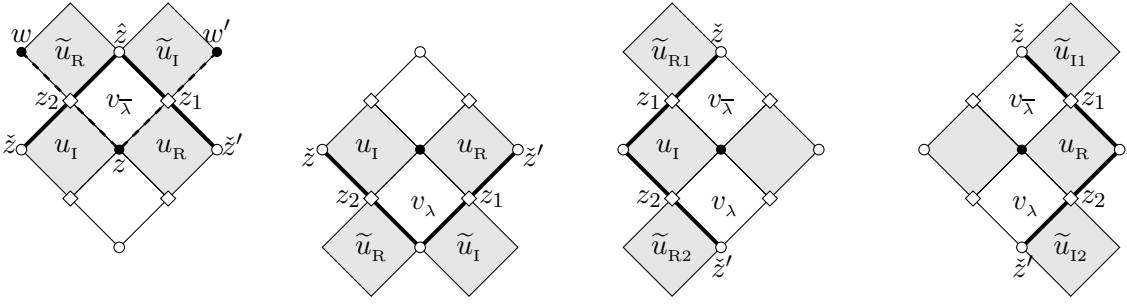
\begin{figure}
\begin{center}

\begin{tikzpicture}[x={(0.5cm,0.5cm)}, y={(-0.5cm,0.5cm)}]
\begin{scope}

\path (-1.3,0) node[name=l1, shape=coordinate]{};
\path (0,0) node[name=l2, shape=coordinate]{};
\path (0,1.3) node[name=l3, shape=coordinate]{};
\path (-1.3,1.3) node[name=l4, shape=coordinate]{};
\path[draw, fill=gray!20] (l1)--(l2)--(l3)--(l4)--cycle;
\path (-0.7,0.7) node[]{$\uI$};

\path[draw, fill=gray!20] (1.3,0)--(2.6,0)--(2.6,1.3)--(1.3,1.3)--cycle;
\path (1.95,0.65) node[]{$\widetilde{u}_{\raisebox{-1pt}{\tiny I}}$};

\path (0,0) node[name=u1, shape=coordinate]{};
\path (1.3,0) node[name=u2, shape=coordinate]{};
\path (1.3,1.3) node[name=u3, shape=coordinate]{};
\path (0,1.3) node[name=u4, shape=coordinate]{};
\path[draw] (u1)--(u2)--(u3)--(u4)--cycle;
\path (0.6,0.6) node[]{$\vlbr$};

\path (0,-1.3) node[name=r1, shape=coordinate]{};
\path (1.3,-1.3) node[name=r2, shape=coordinate]{};
\path (1.3,0) node[name=r3, shape=coordinate]{};
\path (0,0) node[name=r4, shape=coordinate]{};
\path[draw, fill=gray!20] (r1)--(r2)--(r3)--(r4)--cycle;
\path (0.7,-0.7) node[]{$\uR$};

\path[draw, fill=gray!20] (0,1.3)--(1.3,1.3)--(1.3,2.6)--(0,2.6)--cycle;
\path (0.65,1.95) node[]{$\widetilde{u}_{\raisebox{-1pt}{\tiny R}}$};

\path (-1.3,-1.3) node[name=d1, shape=coordinate]{};
\path (0,-1.3) node[name=d2, shape=coordinate]{};
\path (0,0) node[name=d3, shape=coordinate]{};
\path (-1.3,0) node[name=d4, shape=coordinate]{};
\path[draw] (d1)--(d2)--(d3)--(d4)--cycle;

\draw[draw, line width=1.5pt] (1.3,-1.3)--(1.3,1.3);
\draw[draw, line width=1.5pt] (1.3,1.3)--(-1.3,1.3);

\draw[draw, line width=1pt, dashed] (0,2.6)--(0,0);
\draw[draw, line width=1pt, dashed] (2.6,0)--(0,0);

\path[draw, fill=white] (1.2,-0.1)--(1.4,-0.1)--(1.4,0.1)--(1.2,0.1)--cycle;
\path[draw, fill=white] (-1.4,-0.1)--(-1.2,-0.1)--(-1.2,0.1)--(-1.4,0.1)--cycle;
\path[draw, fill=white] (-0.1,1.2)--(0.1,1.2)--(0.1,1.4)--(-0.1,1.4)--cycle;
\path[draw, fill=white] (-0.1,-1.4)--(0.1,-1.4)--(0.1,-1.2)--(-0.1,-1.2)--cycle;
\path[draw, fill=white] (1.3,1.3) circle[radius=0.07cm];
\path[draw, fill=white] (-1.3,-1.3) circle[radius=0.07cm];
\path[draw, fill=white] (1.3,-1.3) circle[radius=0.07cm];
\path[draw, fill=white] (-1.3,1.3) circle[radius=0.07cm];
\path[draw, fill=black] (0,0) circle[radius=0.06cm];
\path[draw, fill=black] (0,2.6) circle[radius=0.06cm];
\path[draw, fill=black] (2.6,0) circle[radius=0.06cm];
\path (0,2.6) node[anchor=south]{$w$};
\path (2.6,0) node[anchor=south]{$w'$};
\path (0,0) node[anchor=north]{$z$};

\path (1.3,0) node[anchor=west]{$z_1$};
\path (0,1.3) node[anchor=east]{$z_2$};
\path (1.8,1.3) node[anchor=east]{$\hat z$};

\path (1.3,-1.15) node[anchor=west]{$\check{z}'$};
\path (-1.15,1.3) node[anchor=east]{$\check{z}$};

\end{scope}
\begin{scope}[xshift=4cm]

\path (-1.3,0) node[name=l1, shape=coordinate]{};
\path (0,0) node[name=l2, shape=coordinate]{};
\path (0,1.3) node[name=l3, shape=coordinate]{};
\path (-1.3,1.3) node[name=l4, shape=coordinate]{};
\path[draw, fill=gray!20] (l1)--(l2)--(l3)--(l4)--cycle;
\path (-0.7,0.7) node[]{$\uI$};

\path[draw, fill=gray!20] (-1.3,-2.6)--(0,-2.6)--(0,-1.3)--(-1.3,-1.3)--cycle;
\path (-0.65,-1.95) node[]{$\widetilde{u}_{\raisebox{-1pt}{\tiny I}}$};

\path (0,0) node[name=u1, shape=coordinate]{};
\path (1.3,0) node[name=u2, shape=coordinate]{};
\path (1.3,1.3) node[name=u3, shape=coordinate]{};
\path (0,1.3) node[name=u4, shape=coordinate]{};
\path[draw] (u1)--(u2)--(u3)--(u4)--cycle;

\path (0,-1.3) node[name=r1, shape=coordinate]{};
\path (1.3,-1.3) node[name=r2, shape=coordinate]{};
\path (1.3,0) node[name=r3, shape=coordinate]{};
\path (0,0) node[name=r4, shape=coordinate]{};
\path[draw, fill=gray!20] (r1)--(r2)--(r3)--(r4)--cycle;
\path (0.7,-0.7) node[]{$\uR$};

\path (-1.3,-1.3) node[name=d1, shape=coordinate]{};
\path (0,-1.3) node[name=d2, shape=coordinate]{};
\path (0,0) node[name=d3, shape=coordinate]{};
\path (-1.3,0) node[name=d4, shape=coordinate]{};
\path[draw] (d1)--(d2)--(d3)--(d4)--cycle;
\path (-0.7,-0.8) node[]{$\vl$};

\path[draw, fill=gray!20] (-2.6,-1.3)--(-1.3,-1.3)--(-1.3,0)--(-2.6,0)--cycle;
\path (-1.95,-0.65) node[]{$\widetilde{u}_{\raisebox{-1pt}{\tiny R}}$};

\draw[draw, line width=1.5pt] (-1.3,1.3)--(-1.3,-1.3);
\draw[draw, line width=1.5pt] (-1.3,-1.3)--(1.3,-1.3);

\path[draw, fill=white] (1.2,-0.1)--(1.4,-0.1)--(1.4,0.1)--(1.2,0.1)--cycle;
\path[draw, fill=white] (-1.4,-0.1)--(-1.2,-0.1)--(-1.2,0.1)--(-1.4,0.1)--cycle;
\path[draw, fill=white] (-0.1,1.2)--(0.1,1.2)--(0.1,1.4)--(-0.1,1.4)--cycle;
\path[draw, fill=white] (-0.1,-1.4)--(0.1,-1.4)--(0.1,-1.2)--(-0.1,-1.2)--cycle;
\path[draw, fill=white] (1.3,1.3) circle[radius=0.07cm];
\path[draw, fill=white] (-1.3,-1.3) circle[radius=0.07cm];
\path[draw, fill=white] (1.3,-1.3) circle[radius=0.07cm];
\path[draw, fill=white] (-1.3,1.3) circle[radius=0.07cm];
\path[draw, fill=black] (0,0) circle[radius=0.06cm];

\path (0,-1.3) node[anchor=west]{$z_1$};
\path (-1.3,0) node[anchor=east]{$z_2$};

\path (1.15,-1.3) node[anchor=west]{$\check{z}'$};
\path (-1.35,1.15) node[anchor=east]{$\check{z}$};

\end{scope}

\begin{scope}[xshift=8cm]

\path (-1.3,0) node[name=l1, shape=coordinate]{};
\path (0,0) node[name=l2, shape=coordinate]{};
\path (0,1.3) node[name=l3, shape=coordinate]{};
\path (-1.3,1.3) node[name=l4, shape=coordinate]{};
\path[draw, fill=gray!20] (l1)--(l2)--(l3)--(l4)--cycle;
\path (-0.7,0.7) node[]{$\uI$};

\path (0,0) node[name=u1, shape=coordinate]{};
\path (1.3,0) node[name=u2, shape=coordinate]{};
\path (1.3,1.3) node[name=u3, shape=coordinate]{};
\path (0,1.3) node[name=u4, shape=coordinate]{};
\path[draw] (u1)--(u2)--(u3)--(u4)--cycle;
\path (0.6,0.6) node[]{$\vlbr$};

\path (0,-1.3) node[name=r1, shape=coordinate]{};
\path (1.3,-1.3) node[name=r2, shape=coordinate]{};
\path (1.3,0) node[name=r3, shape=coordinate]{};
\path (0,0) node[name=r4, shape=coordinate]{};
\path[draw, fill=gray!20] (r1)--(r2)--(r3)--(r4)--cycle;

\path (-1.3,-1.3) node[name=d1, shape=coordinate]{};
\path (0,-1.3) node[name=d2, shape=coordinate]{};
\path (0,0) node[name=d3, shape=coordinate]{};
\path (-1.3,0) node[name=d4, shape=coordinate]{};
\path[draw] (d1)--(d2)--(d3)--(d4)--cycle;
\path (-0.7,-0.8) node[]{$\vl$};

\path[draw, fill=gray!20] (0,1.3)--(1.3,1.3)--(1.3,2.6)--(0,2.6)--cycle;
\path (0.65,1.95) node[]{$\widetilde{u}_{\raisebox{-1pt}{\tiny R1}}$};

\path[draw, fill=gray!20] (-2.6,-1.3)--(-1.3,-1.3)--(-1.3,0)--(-2.6,0)--cycle;
\path (-1.95,-0.65) node[]{$\widetilde{u}_{\raisebox{-1pt}{\tiny R2}}$};

\draw[draw, line width=1.5pt] (1.3,1.3)--(-1.3,1.3);
\draw[draw, line width=1.5pt] (-1.3,1.3)--(-1.3,-1.3);

\path[draw, fill=white] (1.2,-0.1)--(1.4,-0.1)--(1.4,0.1)--(1.2,0.1)--cycle;
\path[draw, fill=white] (-1.4,-0.1)--(-1.2,-0.1)--(-1.2,0.1)--(-1.4,0.1)--cycle;
\path[draw, fill=white] (-0.1,1.2)--(0.1,1.2)--(0.1,1.4)--(-0.1,1.4)--cycle;
\path[draw, fill=white] (-0.1,-1.4)--(0.1,-1.4)--(0.1,-1.2)--(-0.1,-1.2)--cycle;
\path[draw, fill=white] (1.3,1.3) circle[radius=0.07cm];
\path[draw, fill=white] (-1.3,-1.3) circle[radius=0.07cm];
\path[draw, fill=white] (1.3,-1.3) circle[radius=0.07cm];
\path[draw, fill=white] (-1.3,1.3) circle[radius=0.07cm];
\path[draw, fill=black] (0,0) circle[radius=0.06cm];

\path (0,1.3) node[anchor=east]{$z_1$};
\path (-1.3,0) node[anchor=east]{$z_2$};

\path (1.73,1.37) node[anchor=east]{$\check{z}$};
\path (-1.25,-1.85) node[anchor=east]{$\check{z}'$};
\end{scope}

\begin{scope}[xshift=12cm]

\path (-1.3,0) node[name=l1, shape=coordinate]{};
\path (0,0) node[name=l2, shape=coordinate]{};
\path (0,1.3) node[name=l3, shape=coordinate]{};
\path (-1.3,1.3) node[name=l4, shape=coordinate]{};
\path[draw, fill=gray!20] (l1)--(l2)--(l3)--(l4)--cycle;

\path (0,0) node[name=u1, shape=coordinate]{};
\path (1.3,0) node[name=u2, shape=coordinate]{};
\path (1.3,1.3) node[name=u3, shape=coordinate]{};
\path (0,1.3) node[name=u4, shape=coordinate]{};
\path[draw] (u1)--(u2)--(u3)--(u4)--cycle;
\path (0.6,0.6) node[]{$\vlbr$};

\path (0,-1.3) node[name=r1, shape=coordinate]{};
\path (1.3,-1.3) node[name=r2, shape=coordinate]{};
\path (1.3,0) node[name=r3, shape=coordinate]{};
\path (0,0) node[name=r4, shape=coordinate]{};
\path[draw, fill=gray!20] (r1)--(r2)--(r3)--(r4)--cycle;
\path (0.7,-0.7) node[]{$\uR$};

\path (-1.3,-1.3) node[name=d1, shape=coordinate]{};
\path (0,-1.3) node[name=d2, shape=coordinate]{};
\path (0,0) node[name=d3, shape=coordinate]{};
\path (-1.3,0) node[name=d4, shape=coordinate]{};
\path[draw] (d1)--(d2)--(d3)--(d4)--cycle;
\path (-0.7,-0.8) node[]{$\vl$};

\path[draw, fill=gray!20] (1.3,0)--(2.6,0)--(2.6,1.3)--(1.3,1.3)--cycle;
\path (1.95,0.65) node[]{$\widetilde{u}_{\raisebox{-1pt}{\tiny I1}}$};

\path[draw, fill=gray!20] (-1.3,-2.6)--(0,-2.6)--(0,-1.3)--(-1.3,-1.3)--cycle;
\path (-0.65,-1.95) node[]{$\widetilde{u}_{\raisebox{-1pt}{\tiny I2}}$};

\draw[draw, line width=1.5pt] (1.3,-1.3)--(1.3,1.3);
\draw[draw, line width=1.5pt] (-1.3,-1.3)--(1.3,-1.3);

\path[draw, fill=white] (1.2,-0.1)--(1.4,-0.1)--(1.4,0.1)--(1.2,0.1)--cycle;
\path[draw, fill=white] (-1.4,-0.1)--(-1.2,-0.1)--(-1.2,0.1)--(-1.4,0.1)--cycle;
\path[draw, fill=white] (-0.1,1.2)--(0.1,1.2)--(0.1,1.4)--(-0.1,1.4)--cycle;
\path[draw, fill=white] (-0.1,-1.4)--(0.1,-1.4)--(0.1,-1.2)--(-0.1,-1.2)--cycle;
\path[draw, fill=white] (1.3,1.3) circle[radius=0.07cm];
\path[draw, fill=white] (-1.3,-1.3) circle[radius=0.07cm];
\path[draw, fill=white] (1.3,-1.3) circle[radius=0.07cm];
\path[draw, fill=white] (-1.3,1.3) circle[radius=0.07cm];
\path[draw, fill=black] (0,0) circle[radius=0.06cm];

\path (0,-1.3) node[anchor=west]{$z_2$};
\path (1.3,0) node[anchor=west]{$z_1$};

\path (1.73,1.37) node[anchor=east]{$\check{z}$};
\path (-1.25,-1.85) node[anchor=east]{$\check{z}'$};

\end{scope}

\end{tikzpicture}\end{center}
\caption{
{\bf First}: upper boundary; black vertices $w,w'\in\tilde{\partial}\mathcal{V}_\bullet^\delta$; boundary weights of the leap-frog Laplacian $\Delta^\delta_\bullet$: $c_{zw}=c_{zw'}=2(\sqrt2 -1)$; 
{\bf Second}: lower boundary. 
{\bf Third}: left boundary. 
{\bf Fourth}: right boundary.
} \label{Fbdr}
\end{figure}

We define  the function $F^\delta_{\operatorname{s-hol}}$ on 
$(\dom^\delta\cup\partial^{+}_{\mathrm {int}}\bb{0}^\delta\cup\partial^{-}_{\mathrm {int}}\bb{1}^\delta\cup\partial^{\flat}_{\mathrm {int}}\ww{0}^\delta\cup
\partial^{\sharp}_{\mathrm {int}}\ww{1}^\delta)\sqcup\partial\mathcal{V}_\diamond^\delta$ as follows, see~Fig.~\ref{Fbdr} for the notation:
\[ \operatorname{Upper} \operatorname{boundary} 
\begin{cases} 
\begin{array}{llll}
F^\delta_{\operatorname{s-hol}}(\widetilde{u}_{\raisebox{-1pt}{\tiny I}})&:=i(1-\sqrt{2})\check{\F}^\delta_{\operatorname{s-hol}}(\uR); \\
F^\delta_{\operatorname{s-hol}}(\widetilde{u}_{\raisebox{-1pt}{\tiny R}})&:=-(1-\sqrt{2})\check{\F}^\delta_{\operatorname{s-hol}}(\uR); \\
F^\delta_{\operatorname{s-hol}}(\vlbr)&:=\bar{\lambda}\check{\F}^\delta_{\operatorname{s-hol}}(\uR); \\
F^\delta_{\operatorname{s-hol}}(z_1)&:=\check{\F}^\delta_{\operatorname{s-hol}}(\uR)+\check{\F}^\delta_{\operatorname{s-hol}}(\widetilde{u}_{\raisebox{-1pt}{\tiny I}});\\
F^\delta_{\operatorname{s-hol}}(z_2)&:=\check{\F}^\delta_{\operatorname{s-hol}}(\widetilde{u}_{\raisebox{-1pt}{\tiny R}})+\check{\F}^\delta_{\operatorname{s-hol}}({\uI});\\
\end{array}
\end{cases}
\]
\[ \operatorname{Lower} \operatorname{boundary} 
\begin{cases} 
\begin{array}{llll}
F^\delta_{\operatorname{s-hol}}(\widetilde{u}_{\raisebox{-1pt}{\tiny I}})&:=i(\sqrt{2}-1)\check{\F}^\delta_{\operatorname{s-hol}}(\uR); \\
F^\delta_{\operatorname{s-hol}}(\widetilde{u}_{\raisebox{-1pt}{\tiny R}})&:=(\sqrt{2}-1)\check{\F}^\delta_{\operatorname{s-hol}}(\uR); \\
F^\delta_{\operatorname{s-hol}}(\vl)&:={\lambda}\check{\F}^\delta_{\operatorname{s-hol}}(\uR); \\
F^\delta_{\operatorname{s-hol}}(z_1)&:=\check{\F}^\delta_{\operatorname{s-hol}}(\uR)+\check{\F}^\delta_{\operatorname{s-hol}}(\widetilde{u}_{\raisebox{-1pt}{\tiny I}});\\
F^\delta_{\operatorname{s-hol}}(z_2)&:=\check{\F}^\delta_{\operatorname{s-hol}}(\widetilde{u}_{\raisebox{-1pt}{\tiny R}})+\check{\F}^\delta_{\operatorname{s-hol}}({\uI});\\
\end{array}
\end{cases}
\]
\[ \operatorname{Left} \operatorname{boundary} 
\begin{cases} 
\begin{array}{llll}
F^\delta_{\operatorname{s-hol}}(\widetilde{u}_{\raisebox{-1pt}{\tiny R1}})&:=(\sqrt{2}-1){\lambda}\check{\F}^\delta_{\operatorname{s-hol}}(\vlbr); \\
F^\delta_{\operatorname{s-hol}}(\widetilde{u}_{\raisebox{-1pt}{\tiny R2}})&:=(1-\sqrt{2}){\lambda}\check{\F}^\delta_{\operatorname{s-hol}}(\vlbr); \\
F^\delta_{\operatorname{s-hol}}(\uI)&:=-i{\lambda}\check{\F}^\delta_{\operatorname{s-hol}}(\vlbr); \\
F^\delta_{\operatorname{s-hol}}(z_1)&:=\check{\F}^\delta_{\operatorname{s-hol}}(\widetilde{u}_{\raisebox{-1pt}{\tiny R1}})+\check{\F}^\delta_{\operatorname{s-hol}}({\uI});\\
F^\delta_{\operatorname{s-hol}}(z_2)&:=\check{\F}^\delta_{\operatorname{s-hol}}(\widetilde{u}_{\raisebox{-1pt}{\tiny R2}})+\check{\F}^\delta_{\operatorname{s-hol}}({\uI});\\
\end{array}
\end{cases}
\]
\[ \operatorname{Right} \operatorname{boundary} 
\begin{cases} 
\begin{array}{llll}
F^\delta_{\operatorname{s-hol}}(\widetilde{u}_{\raisebox{-1pt}{\tiny I1}})&:=i(1-\sqrt{2}){\lambda}\check{\F}^\delta_{\operatorname{s-hol}}(\vlbr); \\
F^\delta_{\operatorname{s-hol}}(\widetilde{u}_{\raisebox{-1pt}{\tiny I2}})&:=i(\sqrt{2}-1){\lambda}\check{\F}^\delta_{\operatorname{s-hol}}(\vlbr); \\
F^\delta_{\operatorname{s-hol}}(\uR)&:={\lambda}\check{\F}^\delta_{\operatorname{s-hol}}(\vlbr); \\
F^\delta_{\operatorname{s-hol}}(z_1)&:=\check{\F}^\delta_{\operatorname{s-hol}}(\widetilde{u}_{\raisebox{-1pt}{\tiny I1}})+\check{\F}^\delta_{\operatorname{s-hol}}({\uR});\\
F^\delta_{\operatorname{s-hol}}(z_2)&:=\check{\F}^\delta_{\operatorname{s-hol}}(\widetilde{u}_{\raisebox{-1pt}{\tiny I2}})+\check{\F}^\delta_{\operatorname{s-hol}}({\uR}).\\
\end{array}
\end{cases}
\]

\begin{prop}\label{second_half_dirichlet}
The function $F^\delta_{\operatorname{s-hol}}$ is s-holomorphic on $(\clom^\delta\smallsetminus\{\www\})\sqcup\mathcal{V}_{\diamond}^\delta$ and
$[\bar{\partial}F^\delta_{\operatorname{s-hol}}](\www)=\frac{\lambda}{4\delta^2}$.
Moreover, 
it satisfies Riemann-type boundary conditions, i.e. for any $z\in\partial\mathcal{V}_{\diamond}^\delta$ one has
\[
\im[F^\delta_{\operatorname{s-hol}}(z)\cdot{\sqrt{(n(z))}}]=0.\]
\end{prop}

\begin{proof}
Let us check that $F^\delta_{\operatorname{s-hol}}$ is s-holomorphic on the upper boundary.
Note that, for the above definition of the function $F^\delta_{\operatorname{s-hol}}$ the following holds:
\[{\mathrm{Proj}}_{\bar\lambda}[F_{\operatorname{s-hol}}^\delta(z_1)]=
{\mathrm{Proj}}_{\bar\lambda}[F_{\operatorname{s-hol}}^\delta(z_2)]=
{\F}_{\operatorname{s-hol}}^\delta(\vlbr),\]
\[{\mathrm{Proj}}_{i}[F_{\operatorname{s-hol}}^\delta(z_1)]=
{\F}_{\operatorname{s-hol}}^\delta(\uI),\]
\[{\mathrm{Proj}}_{1}[F_{\operatorname{s-hol}}^\delta(z_1)]=
{\F}_{\operatorname{s-hol}}^\delta(\widetilde{u}_{\raisebox{-1pt}{\tiny R}}),\]
\[{\mathrm{Proj}}_{i}[F_{\operatorname{s-hol}}^\delta(z_2)]=
{\F}_{\operatorname{s-hol}}^\delta(\widetilde{u}_{\raisebox{-1pt}{\tiny I}}),\]
\[{\mathrm{Proj}}_{1}[F_{\operatorname{s-hol}}^\delta(z_2)]=
{\F}_{\operatorname{s-hol}}^\delta(\uR),\]
where $z_1, z_2, \vlbr, \uR, \uI, \widetilde{u}_{\raisebox{-1pt}{\tiny R}}, \widetilde{u}_{\raisebox{-1pt}{\tiny R}}$ are on the upper boundary as shown on Fig.~\ref{Fbdr}. The cases of lower, right and left boundaries can be checked similarly. Therefore, $F^\delta_{\operatorname{s-hol}}$ is s-holomorphic on $(\clom^\delta\smallsetminus\{\www\})\sqcup\mathcal{V}_{\diamond}^\delta$.

We now check that $F^\delta_{\operatorname{s-hol}}$ satisfies Riemann-type boundary conditions on the left boundary. Note that by the above definition
\[ 
F^\delta_{\operatorname{s-hol}}(z_1)={\F}^\delta_{\operatorname{s-hol}}(\widetilde{u}_{\raisebox{-1pt}{\tiny R1}})+{\F}^\delta_{\operatorname{s-hol}}({\uI})=(\sqrt{2}-1){\lambda}{\F}^\delta_{\operatorname{s-hol}}(\vlbr)+{\F}^\delta_{\operatorname{s-hol}}({\uI}),
\]
\[ 
F^\delta_{\operatorname{s-hol}}(z_2)={\F}^\delta_{\operatorname{s-hol}}(\widetilde{u}_{\raisebox{-1pt}{\tiny R2}})+{\F}^\delta_{\operatorname{s-hol}}({\uI})=(1-\sqrt{2}){\lambda}{\F}^\delta_{\operatorname{s-hol}}(\vlbr)+{\F}^\delta_{\operatorname{s-hol}}({\uI}),
\]
where $z_1, z_2, \vlbr, \uI, \widetilde{u}_{\raisebox{-1pt}{\tiny R1}},  \widetilde{u}_{\raisebox{-1pt}{\tiny R2}}$ are on the left boundary as shown on Fig.~\ref{Fbdr}. 
Therefore, 
\[ 
F^\delta_{\operatorname{s-hol}}(z_1)=(\sqrt{2}-1){\lambda}{\F}^\delta_{\operatorname{s-hol}}(\vlbr)-i{\lambda}{\F}^\delta_{\operatorname{s-hol}}(\vlbr)=\sqrt{2}{\lambda}{\F}^\delta_{\operatorname{s-hol}}(\vlbr)(1-\lambda),
\]
\[ 
F^\delta_{\operatorname{s-hol}}(z_2)=(1-\sqrt{2}){\lambda}{\F}^\delta_{\operatorname{s-hol}}(\vlbr)-i{\lambda}{\F}^\delta_{\operatorname{s-hol}}(\vlbr)=\sqrt{2}{\lambda}{\F}^\delta_{\operatorname{s-hol}}(\vlbr)(\bar\lambda-1),
\]
since $F^\delta_{\operatorname{s-hol}}(\uI)=-i{\lambda}{\F}^\delta_{\operatorname{s-hol}}(\vlbr)$. 
Note that $(1-\lambda)\cdot{\sqrt{(n(z_1))}}\in\mathbb{R}$, $(\bar\lambda-1)\cdot{\sqrt{(n(z_2))}}\in\mathbb{R}$  on the left boundary and ${\lambda}{\F}^\delta_{\operatorname{s-hol}}(\vlbr)\in\mathbb{R}$. Therefore, for any $z\in\partial\mathcal{V}_{\diamond}^\delta$ on the left boundary we obtain $\im[F^\delta_{\operatorname{s-hol}}(z)\cdot{\sqrt{(n(z))}}]=0.$ The cases of upper, lower and right boundaries can be checked similarly.

To finish the proof note that $[\bar{\partial}F^\delta_{\operatorname{s-hol}}](\www)=[\bar{\partial}F^\delta](\www)=\frac{\lambda}{4\delta^2}.$
\end{proof}

Riemann boundary conditions of the coupling function imply the following local relations for the domino probabilities in hedgehog domains. These relations are not satisfied for general even domains. 

\begin{cor} Let $\mathbb{P}[u,v]$ be the probability that the domino  $[uv]$ is contained in a random domino tiling of~$\om^\delta$, where $u\in\bb{}$ and $v\in\ww{}$. Then for a hedgehog domain $\om^\delta$ the following holds\\
for $a\in\intw{}$
\[\mathbb{P}[a-\delta\lambda,a]+\mathbb{P}[a-\delta\bar\lambda,a]=\mathbb{P}[a+\delta\lambda,a]+\mathbb{P}[a+\delta\bar\lambda,a]=\frac{1}{2},\]
for $a\in\intb{}$
\[\mathbb{P}[a,a+\delta\lambda]+\mathbb{P}[a,a-\delta\bar\lambda]=\mathbb{P}[a,a-\delta\lambda]+\mathbb{P}[a,a+\delta\bar\lambda]=\frac{1}{2}.\] 
\end{cor}

\begin{proof} We check the statement for $a\in\intw{0}$. Let $v_0=a$. Recall that $|C_{\om^\delta}(u,v)|=\mathbb{P}[u,v]$ for adjacent $u\in\bb{}$ and $v\in\ww{}$, therefore, 
\[|C_{\om^\delta} (v+\delta\lambda,v)|+|C_{\om^\delta} (v-\delta\lambda,v)|+
 |C_{\om^\delta} (v-\delta\bar{\lambda},v)|+
 |C_{\om^\delta} (v+\delta\bar{\lambda},v)|=1.\]
On the other hand, \[C_{\om^\delta} (v+\delta\lambda,v)-C_{\om^\delta} (v-\delta\lambda,v)+
 iC_{\om^\delta} (v-\delta\bar{\lambda},v)-
 iC_{\om^\delta} (v+\delta\bar{\lambda},v)=1.\]
 Hence, 
 \begin{align*}
 &\mathbb{P}[\www+\delta\lambda,\www]=C_{\om^\delta}(\www+\delta\lambda,\www)=\delta F^\delta_{\operatorname{s-hol}}(\www+\delta\lambda),\\
 &\mathbb{P}[\www-\delta\lambda,\www]=-C_{\om^\delta}(\www-\delta\lambda,\www)=-\delta F^\delta_{\operatorname{s-hol}}(\www-\delta\lambda),\\
  &\mathbb{P}[\www+\delta\bar\lambda,\www]=-iC_{\om^\delta}(\www+\delta\bar\lambda,\www)=-i\delta F^\delta_{\operatorname{s-hol}}(\www+\delta\bar\lambda),\\
   &\mathbb{P}[\www-\delta\bar\lambda,\www]=iC_{\om^\delta}(\www-\delta\bar\lambda,\www)=i\delta F^\delta_{\operatorname{s-hol}}(\www-\delta\bar\lambda).
 \end{align*}
 
In Section~\ref{h-h=ff} we showed that the primitive $H^\delta$ of the square of the solution $\F_{\operatorname{s-hol}}^\delta$ of a discrete boundary value problem ${\operatorname{RBVP}(\om^\delta,\www)}$ defined by~(\ref{def H}) is well-defined, 
therefore, for $z_\bullet$, $z'_\bullet$, $z''_\bullet$, $z_\circ$, $z'_\circ$, $z''_\circ$ and $\www$ as shown on~Fig.~\ref{H_monodr}
\begin{align*}
0=&(\HH(z_\bullet)-\HH(z'_\circ)) + (\HH(z'_\circ)-\HH(z'_\bullet)) + (\HH(z'_\bullet)-\HH(z_\circ)) +\\
&(\HH(z_\circ)-\HH(z''_\bullet))+ (\HH(z''_\bullet)-\HH(z''_\circ)) + (\HH(z''_\circ)-\HH(z_\bullet))=\\
&i\delta \left( \frac{\lambda}{\sqrt2} \left( F^\delta_{\operatorname{s-hol}}(\www-\delta\lambda)-iF^\delta_{\operatorname{s-hol}}(\www-\delta\bar\lambda) \right)\right)^2
-i\delta\left(\frac{\lambda}{\sqrt2}\left(F^\delta_{\operatorname{s-hol}}(\www+\delta\lambda)-iF^\delta_{\operatorname{s-hol}}(\www+\delta\bar\lambda)\right)\right)^2=\\
&-\frac{1}{2}\left(
F^\delta_{\operatorname{s-hol}}(\www+\delta\lambda)-iF^\delta_{\operatorname{s-hol}}(\www+\delta\bar\lambda)
+F^\delta_{\operatorname{s-hol}}(\www-\delta\lambda)-iF^\delta_{\operatorname{s-hol}}(\www-\delta\bar\lambda)\right)=\\
&\frac{1}{2\delta}\left(\left(\mathbb{P}[\www-\delta\lambda,\www]+\mathbb{P}[\www-\delta\bar\lambda,\www]\right)-\left(\mathbb{P}[\www+\delta\lambda,\www]+\mathbb{P}[\www+\delta\bar\lambda,\www]\right)\right).
\end{align*}
To obtain the result for $a\in\intw{1}$ rotate $\om^\delta$ by $\pi$ and note that $\ww{0}$ and $\ww{1}$ change the roles. For $a\in\intb{1,2}$ rotate $\om^\delta$ by $\frac{\pi}{2}$ and $\frac{3\pi}{2}$.
\end{proof}

\subsection{Proof of the convergence}\label{poc} Let $F^\delta_{\operatorname{s-hol}}$ solve the discrete Riemann boundary value problem ${\operatorname{RBVP}(\om^\delta,\www)}$.
In this section we prove the convergence of $F^\delta_{\operatorname{s-hol}}$ to its continuous counterpart.

Let $\F^\delta_{\mathbb{C}, \www^\delta}$ be the unique discrete s-holomorphic function on the whole plane $\mathbb{C}^\delta \smallsetminus \{\www^\delta\}$ tending to zero at infinity and such that $[\bar{\partial}^\delta\F^\delta_{\mathbb{C}, \www^\delta}](\www^\delta) = \frac{\lambda}{4\delta^2}.$ 
The function $\F^\delta_{\mathbb{C}(z), \www^\delta}$ is asymptotically equal to $\frac{1}{2\pi}\cdot\frac{\lambda}{z-\www}$ as~$\delta\downarrow 0$, see~\cite[Theorem 2.21]{C+S}. Let us define a discrete primitive $H_\star^\delta\colon \mathcal{V}_{\bullet}^\delta\sqcup\mathcal{V}_{\circ}^\delta 
\to \mathbb{R}$ of the difference $F^\delta_{\operatorname{s-hol}}-\F^\delta_{\mathbb{C}, \www^\delta}$ in the same way as above: 
\begin{equation} \label{H_star}
H_\star^\delta(z_2)-H_\star^\delta(z_1)=[F_{\operatorname{s-hol}}^\delta-\F^\delta_{\mathbb{C}, \www^\delta}]^2(a)\cdot(z_2 - z_1),
\end{equation}
where $z_1\in\mathcal{V}_{\circ}^\delta$, $z_2\in\mathcal{V}_{\bullet}^\delta$ are two vertices of the same square $a$.

\begin{remark}\label{harmH*}
The difference $F_{\operatorname{s-hol}}^\delta-\F^\delta_{\mathbb{C}, \www^\delta}$ is s-holomorphic everywhere on $\mathcal{V}_\diamond^\delta\sqcup\clom^\delta$, therefore the function $\HH^\delta_\star$ is leap-frog subharmonic on $\mathcal{V}_{\circ}^\delta$ and it is leap-frog superharmonic on $\mathcal{V}_{\bullet}^\delta$.
\end{remark}

The following convergence theorem for s-holomorphic functions is a straightforward analogue of~\cite[Theorem~$1.8$]{HS}. Alternatively, one can use ideas described in the proof of~\cite[Theorem~$5.5$]{CGS} or ideas from the proof of~\cite[Theorem $2.16$]{CHI}. 

\begin{Th}\label{convF}
Let $\om^\delta$ be a sequence of discrete hedgehog domains of mesh size $\delta$ approximating a simply connected domain $\om$.  Let $\www^{\delta}$ approximate an inner point $v\in\om$. Then $F^\delta_{\operatorname{s-hol}}$ converges uniformly on compact subsets of  $\om\smallsetminus\{v\}$ to a continuous holomorphic function $f^v_\om$, where $f^v_\om$ is defined as in Proposition~$\ref{f}$ (or is defined by~(\ref{notSmooth}), if $\om$ is not smooth).
\end{Th}

The proof is done following the ideas described in ~\cite{CHI} 
and using results described in~\cite{CS}.

\begin{proof}
Let $u^\delta$ be a square on the square lattice with mesh size $\delta$. We denote by $B_r^\delta(u^\delta)$ the set of squares and vertices on this lattice such that the distance from them to $u^\delta$ is less than or equal to~$r$. Let $\partial B_r^\delta(u^\delta)$ be the set of boundary squares and vertices of the set $B_r^\delta(u^\delta)$.

Let $\om_{r}^\delta=(\om^\delta\sqcup\mathcal{V}^\delta)\smallsetminus B_r^\delta(\www^\delta).$ 
Let $\mdel(r)=\max\limits_{z^\delta\in\om_{r}^\delta}|\HH^\delta(z^\delta)|.$

\bigskip

\noindent {\it 1. Assume that for each fixed positive $r$ the function $\mdel(r)$ is bounded, as $\delta\to~0.$}

\bigskip

Theorem 3.12 in~\cite{CS} implies that the functions $F^\delta_{\operatorname{s-hol}}$ are uniformly bounded and therefore equicontinuous on $\om_{r}^\delta$. Thus, due to the Arzelà–Ascoli theorem, the family $F^\delta_{\operatorname{s-hol}}$ is precompact and hence converges along a subsequence to some holomorphic function~$\widetilde{f}$ and $\HH^\delta$ converges to $\widetilde{h}:=\re\int\widetilde{f}^2$ uniformly on compact subsets of $\om\smallsetminus\{v\}$. Let us show that $\widetilde{f}=f^v_\om.$ It is enough to check that $\widetilde{f}$ satisfies properties (h1) -- (h3). Then the uniqueness of a solution of the boundary value problem (f1) -- (f2) implies that $\widetilde{f}$ coincides with the function $f^v_\om.$

Discrete Dirichlet boundary conditions together with the maximum principle for $\HH^\delta$ implies $\widetilde{h}\equiv 0$ on $\dom$, which gives us~(h1). It follows from \cite[Remark 6.3]{CS} that we also have~(h2). 
The fact that $\F^\delta_{\mathbb{C}(z), \www^\delta}$ is asymptotically equal to $\frac{1}{2\pi}\cdot\frac{\lambda}{z-\www}$ implies that $\HH_\star^\delta$ converges to a harmonic function
$\widetilde{h}_\star:= \re\int \left( \widetilde{f}(z)-\frac{1}{2\pi}\cdot\frac{\lambda}{z-v} \right) ^2\,dz$. Remark~\ref{harmH*} gives us that $\widetilde{h}_\star$ is bounded in a vicinity of $v$. Hence $\widetilde{f}$ satisfies properties (h1) -- (h3), so $\widetilde{f}=f^v_\om.$

\bigskip

\noindent {\it 2. Now, suppose that, for some $r>0$, $\mdel(r)$ tends to infinity along a subsequence as $\delta\to 0$.}

\bigskip

Let us show that this is impossible.
Consider renormalized functions $\Ftilda^\delta:=\frac{F_{\operatorname{s-hol}}^\delta}{\sqrt{\mdel(r)}}$ and $\widetilde{H}^\delta:=\frac{H^\delta}{\mdel(r)}.$ 
Using the same arguments as above, we can show that the family $\Ftilda^\delta$ converges to some holomorphic function $\widetilde{f}$ and $\widetilde{H}^\delta$ converges to the harmonic function $\widetilde{h}=\re\int\widetilde{f}^2$ on compact subsets of $\om\smallsetminus B_r(v)$. 

Suppose that $\widetilde{h}$ cannot be identically zero. Then for any $0 < r' < r$ there exists $C(r', r)$ independent of $\delta$, such that $M^\delta(r')\leq C(r', r)\cdot M^\delta(r).$ Therefore we may assume that  
$\widetilde{H}^\delta$ converges to $\widetilde{h}$ uniformly on each $\om_{r'}=\om\smallsetminus B_{r'}(v)$.
Arguing as above, we see that $\widetilde{h}$ is harmonic and satisfies properties (h1) -- (h2). Moreover, since 
$\frac{\F^\delta_{\mathbb{C}, \www^\delta}}{\sqrt{\mdel(r)}}$ tends to zero (as $\delta\to 0$), the limit of $\widetilde{H}_\star^\delta$ coincides with $\widetilde{h}$. Therefore $\widetilde{h}$ is bounded in a vicinity of $v$, satisfies Dirichlet boundary conditions $h\equiv 0$ on $\dom$ and has a nonnegative outer normal derivative. This contradicts  the maximum principle, if it is not identically zero.

{\it 3. To complete the proof it remains to show that none of the subsequential  limits of $\widetilde{\HH}^\delta$ is identically zero.}

Suppose that $\widetilde{H}^\delta$ converges to zero uniformly on compact subsets of $\om\smallsetminus B_r(v)$. 
Let $z_{\operatorname{max}}^\delta$ be chosen so that 
$ 1=\sup\limits_{z^\delta\in\om_{r}^\delta}|\widetilde{H}^\delta(z^\delta)|=|\widetilde{H}^\delta(z_{\operatorname{max}}^\delta)|.$
Since $\widetilde{H}^\delta$ vanishes on the boundary, the discrete maximum principle implies that  $z_{\operatorname{max}}\in \partial B^\delta_r(\www^\delta)$. 

Consider the function $\frac{H_\star^\delta}{\mdel(r)}$. Note that it tends to zero on compact subsets of $\om\smallsetminus B_r(v)$. Therefore the maximum principle together with Remark~\ref{harmH*} implies that $\frac{H_\star^\delta}{\mdel(r)}$ tends to zero in the neighbourhood of $\partial B_r(v)$. Hence, each of the functions 
$
\frac{F_{\operatorname{s-hol}}^\delta-\F^\delta_{\mathbb{C}, \www^\delta}}{\sqrt{\mdel(r)}}, \quad \frac{F_{\operatorname{s-hol}}^\delta}{\sqrt{\mdel(r)}} \quad \text{and} \quad
\frac{H^\delta}{{\mdel(r)}} 
$
tends to zero uniformly in the neighbourhood of $\partial B_r(v)$. In particular, we have $1=|\widetilde{H}^\delta(z_{\operatorname{max}}^\delta)|\to 0$, which is a contradiction.
\end{proof}

\old{
Now to complete the proof of Theorem~\ref{convF} it remains to prove the following:

\begin{lemma} In the notations of the proof above, none of the subsequential  limits of $\widetilde{\HH}^\delta$ is identically zero.
\end{lemma}

\begin{proof} Suppose that $\widetilde{H}^\delta$ converges to zero uniformly on compact subsets of $\om\smallsetminus B_r(v)$. 
Let $z_{\operatorname{max}}^\delta$ be chosen so that 
$ 1=\sup\limits_{z^\delta\in\om_{r}^\delta}|\widetilde{H}^\delta(z^\delta)|=|\widetilde{H}^\delta(z_{\operatorname{max}}^\delta)|.$
Since $\widetilde{H}^\delta$ vanishes on the boundary, the discrete maximum principle implies that  $z_{\operatorname{max}}\in \partial B^\delta_r(\www^\delta)$. 

Consider the function $\frac{H_\star^\delta}{\mdel(r)}$. Note that it tends to zero on compact subsets of $\om\smallsetminus B_r(v)$. Therefore the maximum principle together with Remark~\ref{harmH*} implies that $\frac{H_\star^\delta}{\mdel(r)}$ tends to zero in the neighbourhood of $\partial B_r(v)$. And hence, each of the functions 
$
\frac{F_{\operatorname{s-hol}}^\delta-\F^\delta_{\mathbb{C}, \www^\delta}}{\sqrt{\mdel(r)}}, \quad \frac{F_{\operatorname{s-hol}}^\delta}{\sqrt{\mdel(r)}} \quad \text{and} \quad
\frac{H^\delta}{{\mdel(r)}} 
$
tends to zero uniformly in the neighbourhood of $\partial B_r(v)$. In particular, we have $1=|\widetilde{H}^\delta(z_{\operatorname{max}}^\delta)|\to 0$, which is a contradiction.
\end{proof}}

\begin{figure}
\begin{center}

\begin{tikzpicture}[x={(0.5cm,0.5cm)}, y={(-0.5cm,0.5cm)}]

\begin{scope}

\path (-1,0) node[name=l1, shape=coordinate]{};
\path (0,0) node[name=l2, shape=coordinate]{};
\path (0,1) node[name=l3, shape=coordinate]{};
\path (-1,1) node[name=l4, shape=coordinate]{};
\path[draw, fill=gray!20] (l1)--(l2)--(l3)--(l4)--cycle;
\path (-0.5,0.5) node[]{$\uI$};

\path[draw, fill=gray!20] (-1,-2)--(0,-2)--(0,-1)--(-1,-1)--cycle;

\path (0,0) node[name=u1, shape=coordinate]{};
\path (1,0) node[name=u2, shape=coordinate]{};
\path (1,1) node[name=u3, shape=coordinate]{};
\path (0,1) node[name=u4, shape=coordinate]{};
\path[draw] (u1)--(u2)--(u3)--(u4)--cycle;

\path (0,-1) node[name=r1, shape=coordinate]{};
\path (1,-1) node[name=r2, shape=coordinate]{};
\path (1,0) node[name=r3, shape=coordinate]{};
\path (0,0) node[name=r4, shape=coordinate]{};
\path[draw, fill=gray!20] (r1)--(r2)--(r3)--(r4)--cycle;
\path (0.5,-0.5) node[]{$\uR$};






\path (-1,-1) node[name=d1, shape=coordinate]{};
\path (0,-1) node[name=d2, shape=coordinate]{};
\path (0,0) node[name=d3, shape=coordinate]{};
\path (-1,0) node[name=d4, shape=coordinate]{};
\path[draw] (d1)--(d2)--(d3)--(d4)--cycle;

\path[draw, fill=gray!20] (-2,-1)--(-1,-1)--(-1,0)--(-2,0)--cycle;

\draw[draw, line width=1.5pt] (-1,1)--(-1,-1);
\draw[draw, line width=1.5pt] (-1,-1)--(1,-1);

\path (1,-2) node[name=l1, shape=coordinate]{};
\path (2,-2) node[name=l2, shape=coordinate]{};
\path (2,-1) node[name=l3, shape=coordinate]{};
\path (1,-1) node[name=l4, shape=coordinate]{};
\path[draw, fill=gray!20] (l1)--(l2)--(l3)--(l4)--cycle;
\path (1.5,-1.5) node[]{$\uI$};

\path[draw, fill=gray!20] (1,-4)--(2,-4)--(2,-3)--(1,-3)--cycle;
\path (2,-2) node[name=u1, shape=coordinate]{};
\path (3,-2) node[name=u2, shape=coordinate]{};
\path (3,-1) node[name=u3, shape=coordinate]{};
\path (2,-1) node[name=u4, shape=coordinate]{};
\path[draw] (u1)--(u2)--(u3)--(u4)--cycle;

\path (2,-3) node[name=r1, shape=coordinate]{};
\path (3,-3) node[name=r2, shape=coordinate]{};
\path (3,-2) node[name=r3, shape=coordinate]{};
\path (2,-2) node[name=r4, shape=coordinate]{};
\path[draw, fill=gray!20] (r1)--(r2)--(r3)--(r4)--cycle;
\path (2.5,-2.5) node[]{$\uR$};

\path (1,-3) node[name=d1, shape=coordinate]{};
\path (2,-3) node[name=d2, shape=coordinate]{};
\path (2,-2) node[name=d3, shape=coordinate]{};
\path (1,-2) node[name=d4, shape=coordinate]{};
\path[draw] (d1)--(d2)--(d3)--(d4)--cycle;

\path[draw, fill=gray!20] (0,-3)--(1,-3)--(1,-2)--(0,-2)--cycle;

\draw[draw, line width=1.5pt] (1,-1)--(1,-3);
\draw[draw, line width=1.5pt] (1,-3)--(3,-3);

\draw[draw,gray!60, dashed, line width=1.5pt] (-2.5,0.5)--(2.5,-4.5);

\path (0.5,-2.5) node[]{$\bf \widetilde{u}_{\raisebox{-1pt}{\tiny R}}$};
\path (1.5,-3.5) node[]{$\bf \widetilde{u}_{\raisebox{-1pt}{\tiny I}}$};
\path (-0.5,-1.5) node[]{$\bf \widetilde{u}_{\raisebox{-1pt}{\tiny I}}$};
\path (-1.5,-0.5) node[]{$\bf \widetilde{u}_{\raisebox{-1pt}{\tiny R}}$};
\path (-1.5,-1.5) node[]{$\bf {\bar{v}}^\delta_{\raisebox{-1pt}{\tiny 0}}$};
\path (-0.5,-0.5) node[]{$\bf {v}^\delta_{\raisebox{-1pt}{\tiny 0}}$};

\end{scope}

\end{tikzpicture}\end{center}
\caption{The upper half plane $\mathbb{H}^\delta$ and its boundary (solid line). The real axis (dashed). On the real axis $\im\F^\delta_\mathbb{H}(u)=F^\delta_\mathbb{H}(\widetilde{u}_{\tiny I})=0$. The function $\F^\delta_{\mathbb{H},\www^\delta}$ equals zero on the boundary, and $[\bar{\partial}^\delta\F^\delta_{\mathbb{H}, \www^\delta}](\www^\delta) = \frac{\lambda}{\delta^2},$ $\bar{v}^\delta_{0}=\www^\delta-i\sqrt{2}\delta$.} \label{figSch}
\end{figure}
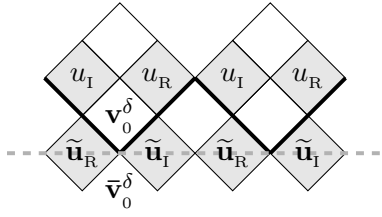

To show the convergence of the fluctuations of the hight function we need to show the convergence of the coupling function up to a straight horizontal boundary segment. To prove Theorem~\ref{double} we also need a version of Theorem~\ref{convF} for the boundary vertex $v$. In order to obtain these results we need to introduce discrete Schwarz reflection principle for hedgehog-type straight boundary. We also need to introduce an analog of the function 
$\F^\delta_{\mathbb{C}, \www^\delta}(u)$ in the upper half-plane.

\begin{lemma}[discrete Schwarz reflection principle]\label{Schwarz}
Let $\F^\delta_{\mathbb{H}}: \bb{} \to\mathbb{C}$ be a discrete holomorphic on the upper half plane $\mathbb{H}^\delta$, such that 
$\im\F^\delta_\mathbb{H}(u)=0$ on the real axis, see Fig.~\ref{figSch}. Then the function 
$\F^\delta_\mathbb{C}$ defined by 
\[
\begin{cases} 
\begin{array}{llll}
\F^\delta_\mathbb{C}(u):={\F^\delta_\mathbb{H}({u})}  \quad 
& \operatorname{if}\, u\in \mathbb{H}^\delta\\
\F^\delta_\mathbb{C}(u):=\overline{\F^\delta_\mathbb{H}(\bar{u})} \quad & \operatorname{if}\, u\in \mathbb{C}^\delta\smallsetminus\mathbb{H}^\delta \\
\end{array}
\end{cases}
\] 
is discrete holomorphic on the whole plane $\mathbb{C}^\delta$, where squares $u$ and $\bar{u}$ are symmetric with respect to the real axis.
\end{lemma}

\begin{proof} Note that on the upper half plane $\bar{\partial}^\delta \F^\delta_\mathbb{C}=\bar{\partial}^\delta\F^\delta_\mathbb{H}=0$. 
Therefore we need to check that $ [\bar{\partial}^\delta \F^\delta_\mathbb{C}](\bar{v})=0$ for all $v\in\mathbb{H}^\delta$: 
\begin{align*} [\bar{\partial}^\delta \F^\delta_\mathbb{C}](\bar{v})&=\frac12\left(\frac{\F^\delta_\mathbb{C}(\bar{v}+\delta\lambda)-
 \F^\delta_\mathbb{C}(\bar{v}-\delta\lambda)}{2\delta\bar{\lambda}}+\frac{\F^\delta_\mathbb{C}(\bar{v}+\delta\bar{\lambda})-
\F^\delta_\mathbb{C}(\bar{v}-\delta\bar{\lambda})}{2\delta\lambda}\right)=\\
 &\frac12\left(\frac{\overline{\F^\delta_\mathbb{H}(v+\delta\bar{\lambda})}-
 \overline{\F^\delta_\mathbb{H}(v-\delta\bar{\lambda})}}{2\delta\bar{\lambda}}+\frac{\overline{\F^\delta_\mathbb{H}(v+\delta{\lambda})}-
 \overline{\F^\delta_\mathbb{H}(v-\delta{\lambda})}}{2\delta\lambda}\right)=\overline{[\bar{\partial}^\delta \F^\delta_\mathbb{H}](v)}=0.
\end{align*} 
To complete the proof, note that for $u$ on the real axis $\overline{\F^\delta_\mathbb{H}(u)}=\F^\delta_\mathbb{H}(u).$\end{proof}

Let us introduce a function $\F^\delta_{\mathbb{H}, \www^\delta}: \bb{} \to\mathbb{C}$ on the half-plane $\mathbb{H}^\delta.$ 
The function $\F^\delta_{\mathbb{H},\www^\delta}$ equals zero on the boundary, and $[\bar{\partial}^\delta\F^\delta_{\mathbb{H}, \www^\delta}](\www^\delta) = \frac{\lambda}{\delta^2}.$ 
Let $\F_1^\delta$ and $\F_2^\delta$ be two different discrete holomorphic functions that satisfy these two properties and tend to zero at infinity. Then the difference $\F_1^\delta-\F_2^\delta$ is discrete holomorphic everywhere in $\mathbb{H}^\delta$, vanishes on the boundary and tends to zero at infinity. Therefore $\F_1^\delta-\F_2^\delta\equiv0.$ Thus, there is a unique such discrete holomorphic function $\F^\delta_{\mathbb{H}, \www^\delta}$.

Let us consider the sum $\F^\delta_{\mathbb{C}, \www^\delta} + \F^\delta_{\mathbb{C}, \www^\delta-i\sqrt{2}\delta},$ where by $\www^\delta-i\sqrt{2}\delta$ we denote a white square symmetric to $\www^\delta$ with respect to the real axis that does not belong to $\mathbb{H}^\delta,$ see Fig.~\ref{figSch}. This sum tends to zero at infinity, since both $\F^\delta_{\mathbb{C}, \www^\delta}$ and $\F^\delta_{\mathbb{C}, \www^\delta-i\sqrt{2}\delta}$ tend to zero at the infinity. Note that $\F^\delta_{\mathbb{C}, \www^\delta-i\sqrt{2}\delta}$ is discrete holomorphic on $\mathbb{H}^\delta$, therefore $\F^\delta_{\mathbb{C}, \www^\delta} + \F^\delta_{\mathbb{C}, \www^\delta-i\sqrt{2}\delta}$ is holomorphic on $\mathbb{H}^\delta \smallsetminus \{\www^\delta\}$ and $[\bar{\partial}^\delta(\F^\delta_{\mathbb{C}, \www^\delta} +\F^\delta_{\mathbb{C}, \www^\delta-i\sqrt{2}\delta})](\www^\delta) = \frac{\lambda}{\delta^2}.$ 
Finally, note that 
\[
\im\F^\delta_{\mathbb{C}, \www^\delta}(u)=G^\delta(u,\www^\delta+\bar\lambda\delta) - G^\delta(u,\www^\delta-\bar\lambda\delta),
\] 
where 
$G^\delta(u,u')$ is the classical Green's function on $\mathbb{C}^\delta\cap\bb{1}^\delta$ satisfying $\Delta^\delta G^\delta(u,u')=\mathbb{1}_{u=u'}\cdot\frac{1}{2\delta^3}$. The Green's function is symmetric, therefore $\im[\F^\delta_{\mathbb{C}, \www^\delta} + \F^\delta_{\mathbb{C}, \www^\delta-i\sqrt{2}\delta}]$ vanishes on $\partial\mathbb{H}^\delta.$ Similarly the real part of $\F^\delta_{\mathbb{H},\www^\delta}$ vanishes on the boundary.
As a consequence we have
$\F^\delta_{\mathbb{H}, \www^\delta}(u) = \F^\delta_{\mathbb{C}, \www^\delta}(u) + \F^\delta_{\mathbb{C}, \www^\delta-i\sqrt{2}\delta}(u)$ for all $u\in \bb{}\cap\mathbb{H}^\delta.$

We will call a part of the boundary of hedgehog domain $\om^\delta$ {\it a right vertical straight part of the boundary} if all inner boundary squares along this part belong to the set $\partial^{+}_{\mathrm {int}}\om^\delta$. Similarly one can define {\it left vertical}, {\it upper horizontal} and {\it lower horizontal} straight parts of the boundary of hedgehog domain.

\begin{prop}\label{cnb}
1. Let $J$ be an open straight horizontal (or vertical) segment of the boundary of~$\om$. Then the uniform convergence in Theorem~\ref{convF} holds on compact subsets of $(\om\cup J)\smallsetminus\{v\}$.


2. Let $\om^\delta$ be a sequence of discrete hedgehog domains of mesh size $\delta$ approximating a simply connected domain $\om$.  Let $\www^\delta$ on a horizontal part of the boundary of $\om^\delta$ approximate a boundary point $v$, which lies on a straight horizontal segment of the boundary of $\om$. Then $F^\delta_{\operatorname{s-hol}}$ converges uniformly on compact subsets of  $\om\smallsetminus\{v\}$ to a continuous holomorphic function $f^v_\om$, where $f^v_\om$ is defined as in Lemma~$\ref{fboundary}$ (or is defined by~(\ref{notSmooth_b}), if $\om$ is not smooth).

\end{prop}

\begin{proof}
Reflect $\om^\delta$ across the lower horizontal straight part of the boundary to get a domain $\mho_\delta$. Glue domains $\om^\delta$ and $\mho_\delta$ together, note that the resulting domain is a hedgehog domain. The discrete holomorphic function $\check{\F}^\delta_{\operatorname{s-hol}}|_{\bb{}^\delta}$ which is zero on the lower horizontal straight part of $\dom^\delta$ extends to a discrete holomorphic function on this glued domain by discrete Schwarz reflection principle, 
see Lemma~\ref{Schwarz}. 
Then one can define a function ${\F}^\delta_{\operatorname{s-hol}}$ on the glued hedgehog domain as above. The argument of Theorem~\ref{convF} can then be applied in this case, with $\F^\delta_{\mathbb{C}, \www^\delta}$ replaced by $\F^\delta_{\mathbb{H}, \www^\delta}$. \end{proof}

\section{Dimers on hedgehog domains and the Gaussian Free Field}\label{5}
In~\cite{KGff} Kenyon proved that the scaling limit of the height function in the dimer model on Temperleyan domains is the Gaussian Free Field. In~\cite{rus} it is proven, that the same scaling limit appears for approximations by piecewise Temperleyan domains.
Our goal in this section is 
to show that the same holds for approximations by hedgehog domains.

\subsection{Asymptotic values of the coupling function}\label{sa_c_f}
Following~\cite{Kdom}, we define two functions $\ff{0} (z_1,z_2)$ and $\ff{1} (z_1,z_2)$. For a fixed $z_2$,
\begin{enumerate}
\item[$\rhd$] the function $\ff{0} (z_1, z_2)$ is analytic as a function of $z_1$, has a simple pole of residue $1/ 2\pi$ at $z_1=z_2$, and no other poles on ${\om}$;
\item[$\rhd$] $\ff{0}(z, z_2)|| \frac{1}{\sqrt{(n(z))}}, \quad z\in\dom$.
\end{enumerate}
The function $\ff{1}(z_1, z_2)$ has the same definition, except for a difference in the boundary conditions: $\ff{1}(z, z_2)|| \frac{i}{\sqrt{(n(z))}}, \quad z\in\dom$.
The existence and uniqueness of such functions can be shown using the technique described in Section~\ref{cont}, see Proposition~\ref{f}. In particular, we can write these functions on the upper half plane in the following way:
\[
\ff{0} (z, w)=\frac{1}{2\pi}\cdot\left( \frac{1}{z-w} + \frac{i}{z-\overline w}\right),
\]

\[
\ff{1} (z, w)=\frac{1}{2\pi}\cdot\left( \frac{1}{z-w} - \frac{i}{z-\overline w}\right).
\]

\begin{Th}\label{main-th2_1}
 Let $\om$ be a simply connected domain in $\mathbb{C}$. Assume that a sequence of discrete hedgehog domains $\om^\delta$ of mesh sizes $\delta$ approximates the domain~$\om$.
Let a sequence of white squares $v^\delta$ approximates a point $v\in \om$. Then the coupling function $\frac{1}{\delta}\Cmd(u,v)$ satisfies the following asymptotics: \\
for $v^\delta \in \ww{0}^\delta$
\[
\frac{1}{\delta}\Cmd(u,v^\delta) - \bar{\lambda}\cdot\F^\delta_{\mathbb{C}, v^\delta}(u)
= \ff{0}(u,v)-\frac{1}{2\pi(u-v)}+o(1);
\]
if $v^\delta \in \ww{1}^\delta$, then
\[
\frac{1}{\delta}\Cmd(u,v^\delta) - \bar{\lambda}\cdot\F^\delta_{\mathbb{C}, v^\delta}(u)
= \ff{1}(u,v)-\frac{1}{2\pi(u-v)}+o(1),
\]
where $\F^\delta_{\mathbb{C}, v^\delta}(u)$ is defined in Section~\ref{poc}.
\end{Th}
\begin{proof} Recall that $\F^\delta_{\mathbb{C}(z), \www^\delta}$ is asymptotically equal to $\frac{1}{2\pi}\cdot\frac{\lambda}{z-\www}$ as~$\delta\downarrow 0$. Recall that the function $F^\delta_{\operatorname{s-hol}}(\cdot)$ coincides with  $\frac{1}{\delta}\Cmd(\cdot,\www)$ on the set $\intbb{}^\delta\cup\partial^{\flat}_{\mathrm {int}}\bb{}^\delta\cup\partial^{\sharp}_{\mathrm {int}}\bb{}^\delta$.
Now, to obtain the first asymptotic one can use Theorem~\ref{convF}. The second one can be obtained similarly.
To see this note that for $v^\delta \in \ww{1}^\delta$ the function $ \frac{i}{\delta}\Cmd(\,\cdot\,,v^\delta)$ is discrete holomorphic everywhere in~$\ww{}^\delta\smallsetminus v^\delta$ with 
$[\bar{\partial}^\delta \frac{i}{\delta}\Cmd(\,\cdot\,,v^\delta)](v^\delta) = \frac{i\lambda}{\delta^2}$ and satisfies the same boundary conditions as $\frac{1}{\delta}\Cmd(\,\cdot\,,v^\delta)$ for $v^\delta \in \ww{0}^\delta$.
\end{proof}

\subsection{ Convergence to GFF} 
To obtain the convergence of the height function on hedgehog domains to the Gaussian free field it is enough to show  that the limits of moments of the height function in the Temperleyan and hedgehog cases are the same. As in our previous paper~\cite{rus} we give only the sketch of the proof of Corollary~\ref{main-cor2}. The novel part of the argument is in ~(\ref{f01}), then Lemma~\ref{eqlimits} completes the proof.

Due to \cite{Kdom} one can obtain the following result for hedgehog approximations. Let $\ff{+}(z, w)=\ff{0}(z, w)+\ff{1}(z, w)$ and $\ff{-}(z, w)=\ff{0}(z, w)-\ff{1}(z, w)$. 
\begin{prop}
Let $\gamma_1, \ldots, \gamma_m$ be a collection of pairwise disjoint paths running from the horizontal straight boundary segment of $\Omega$ to $z_1, \ldots, z_m$ respectively. Let $h(z_i)$ denote the height function at a point in a hedgehog domain $\Omega^\delta$ lying within $O(\delta)$ of $z_i$. Then 
\begin{align}\label{detdetdet}
\lim_{\delta\to 0}\mathbb{E}&[(h(z_1)-\mathbb{E}[h(z_1)])\cdot\ldots\cdot(h(z_m)-\mathbb{E}[h(z_m)])]~=~\\ 
&\sum_{\epsilon_1, \ldots, \epsilon_m \in \{-1,1\}} \epsilon_1\cdots\epsilon_m\int_{\gamma_1}\cdots\int_{\gamma_m}\det_{i,j\in [1,m]}(F_{\epsilon_i,\epsilon_j}(z_i, z_j))\,dz_1^{(\epsilon_1)}\cdots dz_m^{(\epsilon_m)},\label{detdet}
\end{align}
where $dz_j^{(1)}=dz_j$ and $dz_j^{(-1)}=d\overline{z_j}$, and
\[
F_{\epsilon_i, \epsilon_j}(z_i,z_j)=\begin{cases} 
0 &i=j\\
f_{+}(z_i,z_j) \quad &(\epsilon_i,\epsilon_j)=(1,1)\\
f_{-}(z_i,z_j) \quad &(\epsilon_i,\epsilon_j)=(-1,1)\\
\overline{f_{-}(z_i,z_j)} \quad &(\epsilon_i,\epsilon_j)=(1,-1)\\
\overline{f_{+}(z_i,z_j)} \quad &(\epsilon_i,\epsilon_j)=(-1,-1).
\end{cases}
\]
\end{prop}

\begin{proof}
The function $F^\delta_{\operatorname{s-hol}}(\cdot)$ coincides with  $\frac{1}{\delta}\Cmd(\cdot,\www)$ on the set $\intbb{}^\delta\cup\partial^{\flat}_{\mathrm {int}}\bb{}^\delta\cup\partial^{\sharp}_{\mathrm {int}}\bb{}^\delta$. Therefore, due to Theorem~\ref{convF} and Proposition~\ref{cnb} one has convergence of the coupling function up to the horizontal straight part of the boundary. The rest of the proof of the proposition mimics the proof of \cite[Proposition 20]{Kdom}.
\end{proof}

Recall that in the Temperleyan case~\cite{Kdom} one has $\ff{+}(z,w)=\frac{1}{\pi(z-w)}$ and $\ff{-}(z,w)=\frac{1}{\pi(z-\overline w)}$. In the hedgehog case we have
\begin{equation} \label{f01}
\begin{cases}
\ff{+}(z, w)=\frac{1}{\pi(z-w)}\\
\ff{-}(z, w)=i\cdot\frac{1}{\pi(z-\overline w)}.
\end{cases}
\end{equation}

\begin{lemma}\label{eqlimits}
 The limits of moments of the height function in Temperleyan case and hedgehog case are the same.
\end{lemma}
\begin{proof}
It is easy to check that the determinants in (\ref{detdet}) are the same for both cases.
\end{proof}
\begin{remark}[conformal covariance]\label{conf_cov} Let $\om$ and $\om'$ be simply connected domains. 
Let $\ff{+}^\om(z, w)$ and $\ff{-}^\om(z, w)$ be the functions defined as above for the region $\om$. The function $\ff{+}^\om(z, w)$ is holomorphic in both variables and the function $\ff{-}^\om(z, w)$ is holomorphic in $z$ and antiholomorphic in~$w$.
Let $\phi$ be a conformal mapping of $\om$ onto $\om'$. Then
\[
\ff{+}^\om(z, w) = \ff{+}^{\om'}(\phi(z), \phi(w))\cdot(\phi'(z))^\frac12\cdot(\phi'(w))^\frac12,
\]
\[
\ff{-}^\om(z, w) = \ff{-}^{\om'}(\phi(z), \phi(w))\cdot(\phi'(z))^\frac12\cdot\overline{(\phi'(w))}^\frac12.
\]
\end{remark}

By Lemma~\ref{eqlimits} the following proposition holds for hedgehog domains as well. Therefore the rest of the argument  of the proof of Corollary~\ref{main-cor2} is exactly as in~\cite[Theorem 1.1]{KGff}.

\begin{prop}[\cite{KGff}]
Let $\Omega$ be a simply connected domain. 
Let $z_1, \ldots, z_m$ (with $m$ even) be distinct points of $\om$. Let $\Omega^\delta$ be a hedgehog approximation of $\Omega$ and 
$h_{\Omega^\delta}$ be the height function of a uniform domino tiling in the domain $\Omega^\delta$. Then 
\begin{align*}
\lim_{\delta\to 0}\mathbb{E}[(h_{\Omega^\delta}(z_1)-\mathbb{E}[h_{\Omega^\delta}(z_1)])\cdot\ldots\cdot&(h_{\Omega^\delta}(z_m)-\mathbb{E}[h_{\Omega^\delta}(z_m)])]~=~\\ 
\left(-\frac{16}{\pi}\right)^{m/2}\sum_{pairings \,\alpha} &g_D(z_{\alpha(1)},z_{\alpha(2)})\cdot\ldots\cdot g_D(z_{\alpha(m-1)},z_{\alpha(m)}),
\end{align*}
where $g_D$ is the Green function with Dirichlet boundary conditions on $\Omega$.
\end{prop}

 And the following lemma completes the proof of Corollary~\ref{main-cor2}.
 
 \begin{lemma}[\cite{proba, KGff}]
A sequence of multidimensional random variables whose moments converge to the moments of a Gaussian, converges itself to a Gaussian.
\end{lemma}


\section{Double-dimer height function in hedgehog domains }\label{6}
In~\cite{rus} it was shown that there is a factorization of the gradient of the expectation of the height function in the double-dimer model into a product of two discrete holomorphic functions. In this section we use this factorisation and result of Theorem~\ref{convF} to show the convergence of the expectation of the double-dimer height function to the harmonic measure for approximations by hedgehog domains.

\subsection{A factorization of the double-dimer coupling function}
Let $\F\colon\clb\to\mathbb{C}$ be a function such that $\F|_{\db}=0$ and $\F$ is discrete holomorphic everywhere in $\ww{}$ except at the face $\www$ where one has $[\bar{\partial}\F](\www) = \lambda.$ Similarly, let $\G~\colon~\clw~\to~\mathbb{C}$ be a function such that $\G|_{\dw}=0$ and $\G$ is discrete holomorphic everywhere in $\bb{}$ except at the face $\bbb$ where one has $[\bar{\partial}\G](\bbb) = i.$ Let $u\in\bb{}$ and $v \in \ww{}$, then due to~\cite[Proposition~$3.11$]{rus} $C_{\operatorname{dbl-d}, \om}(u, v) = \mathrm{const}\cdot\F(u)\G(v),$ where  $\mathrm{const}=\frac{1}{4\G(\www)}$. 

Note that the function $F$ (resp., G) coincide with the dimer coupling function $C_\om(\cdot,\www)$ (resp., $C_\om(\bbb,\cdot)$) up to a multiplicative constant.

\subsection{Proof of Theorem~\ref{double}} 
From now onwards, let $\bbb$ be a point on a lower horizontal straight part of the boundary of the domain~$\om$, and $\www$ be a point on a right vertical straight part of $\dom$.
Let $f_\om := f_\om^{v_0} $ and $g_\om := f_\om^{u_0}$, where $f_\om^v(z)$ solves boundary value problem (F1)--(F2) described in Lemma~\ref{fboundary}. Note that due to 
Proposition~\ref{cnb}  the function $F^\delta$ converges to $f_\om$ and  the function $G^\delta$ converges to $g_\om$.
Now to complete the proof of Theorem~\ref{double} it remains to prove the following:
\begin{prop}\label{f*g}
Let $\om$  be a bounded simply connected domain in $\mathbb{C}$ with smooth boundary. Let $\www$~and~$\bbb$ be the points on the straight part of the boundary of the domain~$\om$.  Assume that the boundary arc $(\bbb\www)$ contains~$0$.
Then the function 
\[\int^w_0 \re[f_\om(z)g_\om(z)dz]\] 
is proportional to the harmonic measure $\operatorname{hm}_\om(w, (\www\bbb))$ in the domain~$\om$.
\end{prop}

\begin{proof} Let $\phi$ be a conformal mapping of the domain $\om$ onto the unit disk $\mathbb{D}$ such that $\www$ mapped onto $-i$ and $\bbb$ mapped onto $1$. Note that $\phi'(\www)>0$ and $\phi'(\bbb)>0$, since $\bbb$ (resp., $\www$) is a point on a lower horizontal (reps., right vertical) straight part of~$\dom$.
Let us consider the product of functions $f_\om(z)$ and $g_\om(z)$. It equals
\begin{align*}
f_\om(z)\cdot g_\om(z)=\frac{1}{2\pi}\left(\frac{\mu-\phi(v_0)\bar\mu}{\phi(z)-\phi(v_0)}\right)\cdot(\phi'(z))^\frac12\cdot(\phi'(v_0))^\frac12\times
\frac{1}{2\pi}\left(\frac{\mu-\phi(u_0)\bar\mu}{\phi(z)-\phi(u_0)}\right)\cdot(\phi'(z))^\frac12\cdot(\phi'(u_0))^\frac12\\
=\frac{ci\lambda\cdot\phi'(z)}{(\phi(z)-\phi(\www))(\phi(z)-\phi(\bbb))}=
\frac{\bar{\lambda} c}{\phi(\www)-\phi(\bbb)}\cdot\left(\log\frac{(\phi(z)-\phi(\www))}{(\phi(z)-\phi(\bbb))}\right)',
\end{align*}
hence $\int \re[fgdz]$ is proportional to 
$\frac{1}{\pi}\im\left[ \log \left(\frac{\phi(z)-\phi(\www)}{\phi(z)-\phi(\bbb)}\right)\right]$ which is the harmonic measure of~$(\www\bbb)$.
\end{proof}


\end{document}